\documentclass[twoside,11pt]{article}
\usepackage[letterpaper,top=2cm,bottom=2cm,left=3cm,right=3cm,marginparwidth=1.75cm]{geometry}
\usepackage{natbib}
\usepackage{url}            
\usepackage{booktabs}       
\usepackage{amsfonts}       
\usepackage{nicefrac}       
\usepackage{microtype}      
\usepackage{bm}
\usepackage{changepage}
\usepackage{wrapfig}
\usepackage{array, makecell} 
\usepackage{tabularx}
\usepackage{amsmath}
 
\usepackage{amsthm}
\usepackage{algorithm}
\usepackage{algorithmicx}
\usepackage[noend]{algpseudocode}
\usepackage{tikz}
\usepackage{subfigure}
\usepackage{mathrsfs}
\usepackage{amssymb}
\usepackage{xspace}
\usepackage{thmtools}
\usepackage{thm-restate}
\usetikzlibrary{arrows}
\usepackage{xcolor}
\usepackage{multirow}
\usepackage{threeparttable}
\usepackage{footmisc}

\newcolumntype{Y}{>{\arraybackslash}X}
\newcolumntype{L}[1]{>{\raggedright\let\newline\\\arraybackslash\hspace{0pt}}m{#1}}
\newcolumntype{C}[1]{>{\centering\let\newline\\\arraybackslash\hspace{0pt}}m{#1}}
\newcolumntype{R}[1]{>{\raggedleft\let\newline\\\arraybackslash\hspace{0pt}}m{#1}}

\usetikzlibrary{arrows}
\newtheorem{theorem}{Theorem}
\newtheorem{lemma}{Lemma}

\newtheorem{corollary}{Corollary}

\newtheorem{claim}{Claim}

\newtheorem{prop}{Property}

\newtheorem{assume}{Assumption}
\graphicspath{{figures/}}

\newcommand{\revone}{\color{black}}
\newcommand{\revtwo}{\color{black}}

\newcommand{\ie}{\textit{i.e.}\xspace}
\newcommand{\st}{\textit{s.t.}\xspace}

\newcommand{\ex}[1]{ \mathbb{E} \left[ #1 \right] }

\newcommand{\prob}[1]{ Pr \left( #1 \right) }
\newcommand{\Prob}{ Pr}

\newcommand{\E}{ \mathbb{E}}

\newcommand{\Lamb}[0]{\Lambda}
\newcommand{\lamb}[0]{\lambda}
\newcommand{\epsi}[0]{ \varepsilon }
\newcommand{\reals}[0]{ \mathbb{R}^+ }

\DeclareMathOperator*{\argmax}{arg\,max}

\newcommand{\sg}{\textsc{Greedy}\xspace}

\newcommand{\oh}[1]{O\left( #1 \right)}

\newcommand\numberthis{\addtocounter{equation}{1}\tag{\theequation}}

\newcommand{\opt}{\textsc{OPT}\xspace}

\newcommand{\sm}{\textsc{SM}\xspace}
\newcommand{\marge}[2]{\Delta \left( #1 \, | \, #2 \right) }

\algnewcommand{\LineComment}[1]{\State \(\triangleright\) #1}

\newcommand{\boost}{\textsc{ParallelGreedyBoost}\xspace}

\newcommand{\threshold}{\textsc{ThresholdSeq}\xspace}

\newcommand{\bicshort}{\textsc{BCG}\xspace}
\newcommand{\greedi}{\textsc{GreedI}\xspace}
\newcommand{\rgreedy}{\textsc{RandGreeDI}\xspace}

\newcommand{\rgshort}{\textsc{RG}\xspace}

\newcommand{\palgshort}{\textsc{PAlg}\xspace}

\newcommand{\dsb}{\textsc{R-DASH}\xspace }
\newcommand{\dsbfull}{\textsc{Randomized-DASH}\xspace }
\newcommand{\dsbshort}{\textsc{R-DASH}\xspace}
\newcommand{\mg}{\textsc{MED}\xspace}

\newcommand{\mgrg}{\textsc{MED+RG}\xspace}
\newcommand{\mgshort}{\textsc{MED}\xspace}
\newcommand{\mgfull}{\textsc{MemoryEfficientDistributed}\xspace}

\newcommand{\halfopt}{\frac{1}{2}(1-1/e-\epsi)}
\newcommand{\halfoptshort}{\frac{1}{2}(1-\frac{1}{e}-\epsi)}

\newcommand{\alg}{\textsc{Alg}\xspace}
\newcommand{\lag}{\textsc{LAG}\xspace}

\newcommand{\lat}{\textsc{ThreshSeqMod}\xspace}

\newcommand{\dat}{\textsc{T-DASH}\xspace }
\newcommand{\datfull}{\textsc{Threshold-DASH}\xspace }
\newcommand{\dagh}{\textsc{G-DASH}\xspace }
\newcommand{\dagfull}{\textsc{Greedy-DASH}\xspace }
\newcommand{\ddist}{\textsc{DistributedDistorted}\xspace }
\newcommand{\ddshort}{\textsc{DDist}\xspace }
\newcommand{\algsol}{\textsc{AlgSol}\xspace}
\newcommand{\algrel}{\textsc{AlgRel}\xspace}
\newcommand{\univ}{\mathcal N}
\newcommand{\tg}{\textsc{ThresholdGreedy}\xspace}
\newcommand{\tf}{\textsc{ThresholdFilter}\xspace}
\newcommand{\dls}{\textsc{LinearTime-Distributed}\xspace}
\newcommand{\dlsshort}{\textsc{L-Dist}\xspace}

\newcommand{\mgalg}{\textsc{MED+Alg}\xspace}

\newcommand{\pgshort}{\textsc{PG}\xspace}
\newcommand{\qsshort}{\textsc{LTC}\xspace}
\newcommand{\qs}{\textsc{Linear-TimeCardinality}\xspace}

\newcommand{\tgshort}{\textsc{TG}\xspace}

\newlength\myindent
\begin{document}
\allowdisplaybreaks[4]
\title{Scalable Distributed Algorithms for Size-Constrained Submodular Maximization in the MapReduce and Adaptive Complexity Models}

%

\author{%
 Yixin Chen
  chen777@tamu.edu \\
  Department of Computer Science \& Engineering\\
        Texas A\&M University\\
         College Station, TX
  \and
  Tonmoy Dey
 tdey@fsu.edu \\
 Department of Computer Science\\
  Florida State University\\
  Tallahassee, FL, USA
       \and
       Alan Kuhnle kuhnle@tamu.edu \\
       Department of Computer Science \& Engineering\\
         Texas A\&M University\\
         College Station, TX
}

\maketitle

\begin{abstract}
  Distributed maximization of a submodular function in the MapReduce (MR) model has received
  much attention, culminating in two frameworks that allow
  a centralized algorithm to
  be run in the MR setting without loss of approximation, as long as the centralized algorithm
  satisfies a certain consistency property -- which had previously only been known to be satisfied by
  the standard greedy and continous greedy algorithms. A separate line of work has studied parallelizability
  of submodular maximization in the adaptive complexity model, where each thread may have
  access to the entire ground set.
  For the size-constrained maximization of a monotone and submodular function,
  we show that several sublinearly adaptive (highly parallelizable) algorithms satisfy the consistency
  property required to work in the MR setting, which yields practical, parallelizable
  and distributed algorithms. Separately, we develop the first
  distributed algorithm with linear query complexity for this problem.
  Finally, we provide a method to increase the maximum cardinality
  constraint for MR algorithms at the cost of additional MR rounds.  
\end{abstract}

\section{Introduction}  \label{intro}
 Submodular maximization has become an important problem in data 
mining and machine learning with real world applications ranging from 
video summarization \citep{MirzasoleimanJ018} and mini-batch selection \citep{JosephRSB19}
to more complex tasks such as active learning \citep{RangwaniJAB21} and federated learning \citep{Balakrishnan0ZH22}.
In this work, we study the size-constrained maximization of a monotone,
submodular function (SMCC), formally defined in Section \ref{sec:prelim}. Because of the ubiquity
of problems requiring the optimization of a submodular function, a vast literature
on submodular optimization exists; we refer the reader to the surveys \citep{LiuCPZ20,Liu20}. 
\begin{table*}[t] 
  {
  \fontsize{9}{10}
  \selectfont
  
  \centering
  \begin{tabular}{ l|p{1.6cm}|p{0.9cm}|p{2.5cm}|p{3.0cm}|p{1.5cm}} 
    \hline
   
   Reference                                              &  Ratio                                                                                  & MR Rounds                                   & Adaptivity ($\Theta$)                                                                                                                         & Queries ($\Xi$)                                                                                                                             & $k_{max}$                               \\ 
   \hline
    \rgshort 
    &  $\frac{1}{2}(1-\frac{1}{e})$                                                           &  $\boldsymbol{2}$                           & {\fontsize{7}{10} \selectfont $O(k)$}                                                                                                         & {\fontsize{7}{10} \selectfont$O(nk+k^2\ell)$}                                                                                               & $O(\frac{n}{\ell^2})$                   \\ 
   
   \palgshort 
   & $\boldsymbol{ 1-\frac{1}{e}-\epsi  }$                                                & $O(\frac{1}{\epsi^2})$                      & {\fontsize{7}{10} \selectfont $O(\frac{k}{\epsi^2})$}                                                                                         & {\fontsize{7}{10} \selectfont$O(\frac{nk}{\epsi^2}+\frac{k^2 \ell^2}{\epsi^4})$ }                                                           & $O(\frac{n\epsi^2}{\ell^2})$            \\ 
   
   \ddshort 
   & $\boldsymbol{1-\frac{1}{e}-\epsi} $                                                     & $O(\frac{1}{\epsi})$                        & {\fontsize{7}{10} \selectfont $O(\frac{k}{\epsi})$}                                                                                           & {\fontsize{7}{10} \selectfont$O(\frac{nk}{\epsi}+\frac{k^2 \ell^2}{\epsi^2})$ }                                                             & $O(\frac{n\epsi}{\ell^2})$              \\ 
   
   \bicshort 
   & $1-\epsi$   & $r$     & {\fontsize{7}{10} \selectfont $O\left(\frac{rk}{\epsi^{2/r}}\log ^{2}(\frac{1}{\epsi^{1/r}})\right)$ }                                        & {\fontsize{7}{10} \selectfont$O(\frac{nkr}{\epsi^{1/r}} + \frac{k^2\ell r m^2(\epsi^{1/r})}{\epsi^{3/r}})$ }                                & $O(\frac{n\epsi^{1/r}}{\ell^2})$        \\ 
  
  \pgshort 
  & $0.545 - \epsi$   & $\boldsymbol{2}$                            & {\fontsize{7}{10} \selectfont $O(k^2)$ }                                                                                                      & {\fontsize{7}{10} \selectfont$O(nk + k^3)$}                                                                                                 & $O(\frac{n}{\ell^2})$                   \\ 
  
  
    \tgshort 
    & $\frac{1}{2}$                                                                           &  $\boldsymbol{2}$                           &  {\fontsize{7}{10} \selectfont ${ O(n\log(k)) }$}                                                                                             & {\fontsize{7}{10} \selectfont ${ O(n\log(k))}$}                                                                                             & $O(\frac{n}{\ell^2 \log(k)})$           \\ 
  
   \hline
   \dsbshort (Alg. \ref{alg:DSB})                                             & $\halfoptshort$                                                                         &  $\boldsymbol{2}$                           & {\fontsize{7}{10} \selectfont {$O(\frac{\log(k)\log(n)}{\epsi^4})$} }                                   & {\fontsize{7}{10} \selectfont {${ O(\frac{n\log(k)}{\epsi^4})}$}}                                                          & $O(\frac{n}{\ell^2})$                   \\ 
   
   \dagh (Alg. \ref{alg:DAG})                                    & $\boldsymbol{1-\frac{1}{e}-\epsi}$                                                      & $O(\frac{1}{\epsi})$                        & {\fontsize{7}{10} \selectfont {$O\left(\frac{\log(k)\log(n)}{\epsi^5} \right)$} }                 & {\fontsize{7}{10} \selectfont {$O(\frac{n\log(k)}{\epsi^5})$}}                                                        & $O(\frac{n\epsi}{\ell^2})$ \\ 
   
   \dat (Alg. \ref{alg:DAT})                                                  & $\frac{3}{8}-\epsi$                                                                     &  $\boldsymbol{2}$                           &  {\fontsize{7}{10} \selectfont {$\boldsymbol{ O(\frac{\log(n)}{\epsi^3}) }$} }                                      & {\fontsize{7}{10} \selectfont {${ O(\frac{n\log(k)}{\epsi^4})}$}}                                                          & $O(\frac{n\epsi}{\ell^2 \log(k)})$           \\ 

  \dlsshort (Alg.~\ref{alg:dls2})   & $\frac{1}{8}$            &  $\boldsymbol{2}$                           & {\fontsize{7}{10}$O(\frac{n}{\ell})$}    & {\fontsize{7}{10}$\boldsymbol{O(n)}$}     & $O(\frac{n}{\ell^2 \log(n)})$    \\       
   
   \hline
  \end{tabular}

  \caption{Theoretical comparison of our algorithms to MapReduce model algorithms:
  \rgshort of \citet{barbosa2015power},
  \palgshort of \citet{barbosa2016new},
  \ddshort of \citet{kazemi2021regularized},
  \bicshort of \citet{Epasto2017}, 
  \pgshort of \citet{mirrokni2015randomized},
  and \tgshort of \citet{liu2018submodular}.
  The column $k_{max}$ provides the maximum cardinality constraint that the algorithm can compute.
} \label{table:algs}
  }
  \end{table*}

  A foundational result of \citet{Nemhauser1978} shows that a
  simple greedy algorithm (\sg, pseudocode in Appendix~\ref{app:Greedy}) achieves the optimal
  approximation ratio for SMCC of $1 - 1/e \approx 0.63$ in the value query model, in which
  the submodular function $f$ is available to the algorithm as an oracle that returns
  $f(S)$ when queried with set $S$.
  However, because of the modern revolution in data \citep{MaoNJLCLLY0LYXX21,EttingerCCLZPCS21},
  { \revone
  there is a need for more scalable algorithms. Specifically, the standard greedy
  algorithm is a centralized algorithm that needs access to the whole dataset,
  makes many adaptive rounds (defined below) of queries to the submodular function,
  and has quadratic total query complexity in the worst case.}
  
  \textbf{Distributed Algorithms for SMCC.}
  Massive data sets are often too large to fit on a single machine and are
  distributed over a cluster of many machines.
  In this context, there has been
  a line of work developing algorithms for SMCC in the MapReduce model, defined
  formally in Section \ref{sec:prelim} (see Table \ref{table:algs} and
  the Related Work section below). 
  Of these, \palgshort \citep{barbosa2015power} is a general framework
  that allows a centralized algorithm \alg to be converted to the MapReduce
  distributed setting with nearly the same theoretical approximation ratio, as long as \alg satisfies a technical, randomized
  consistency property (RCP), defined in Section \ref{sec:prelim}.
  
 In addition, \ddist \citep{kazemi2019submodular},
 is also a general framework that can adapt any \alg
 satisfying RCP, as we show\footnote{In \citet{kazemi2019submodular}, an instantiated version of
 \ddist with a distorted greedy algorithm is analyzed.} in
 Section \ref{apx:dagfull}.
 However, to the best of our knowledge,
 the only centralized algorithms 
  to satisfy this property are
 the \sg \cite{Nemhauser1978} and the continuous greedy algorithm of
 \citet{Calinescu2007}.

 \textbf{Parallelizable Algorithms for SMCC.} A separate line of works, initiated by \citet{balkanski2018adaptive}
 has taken an orthogonal approach to scaling up the standard greedy algorithm: parallelizable algorithms for SMCC as measured
 by the \textit{adaptive complexity} of the algorithm, formally defined in Section \ref{sec:prelim}.
 In this model, each thread may have access to the entire
 ground set and hence these algorithms do not apply in a distributed setting.
 Altogether, these algorithms have exponentially improved the adaptivity of standard greedy from $O(n)$ to $O( \log (n) )$ while obtaining
 the nearly the same $1 - 1/e$ approximation ratio \cite{Fahrbach2018,balkanski2019exponential,Ene,Chekuri2018}.
 Recently, highly practical sublinearly adaptive algorithms have been developed without compromising
 theoretical guarantees \cite{Breuer2019,chen2021best}. 
 However, none of these algorithms have been shown to satisfy the RCP. 

 \textbf{Combining Parallelizability and Distributed Settings.}
 In this work, we study parallelizable, distributed algorithms
 for SMCC; each node of the distributed cluster likely has many
 processors, so we seek to take advantage of this situation by
 parallelizing the algorithm within each machine in the cluster.
 In the existing MapReduce algorithms, the number of processors
 $\ell$ could be set to the number of
 processors available to ensure the usage of all
 available resources -- this would consider each processor as a separate
 machine in the cluster. However, there are a number of disadvantages
 to this approach:
 1) A large number of machines severely restricts the size of the cardinality constraint $k$ -- since, in all
of the MapReduce algorithms, a set of size $k \ell$ must be stored on a single machine;
therefore, we must have $k \le O( \Psi / \ell ) = O( n / \ell^2 )$. For example,
with $n = 10^6$, $\ell = 256$ implies that $k \le 15 \cdot C$, for some $C$. However, if each machine has
8 processor cores, the number of machines $\ell$ can be set to 32 (and then parallelized on each machine),
which enables the use of 256 processors with $k \le 976 \cdot C$. As the number of cores and processors
per physical machine continues to grow, large practical benefits can be obtained by a parallelized
and distributed algorithm. 
2) Modern CPU architectures have many cores that all share a common memory,
and it is inefficient to view communication between these cores
to be as expensive as communication between separate machines in
a cluster. 

{\revone In addition to the MapReduce and adaptive complexity models,
  the total query complexity of the algorithm to the submodular oracle is also
  relevant to the scalability of the algorithm. Frequently, evaluation of the
  submodular function is expensive, so the time spent answering oracle queries dominates
the other parts of the computation.}
Therefore, the following questions are posed: 
\textit{Q1: Is it possible to design constant-factor approximation distributed 
  algorithms with 1) constant number of MapReduce rounds; 2)
  sublinear adaptive complexity (highly-parallelizable); and 3) nearly linear total query complexity? Q2: Can we design practical
        distributed algorithms that also meet these three criteria?} 

      \subsection{Contributions}
      In overview, the contributions of the paper are: 1) an exponential improvement in adaptivity
      (from $\Omega( n )$ to $O( \log n ))$
      for an algorithm with a constant number of MR rounds; 2)
      the first linear-time, constant-factor algorithm with constant MR rounds; 
      3) a way to increase the size constraint limitation at the cost of additional MR rounds; and
      4) an empirical evaluation on a cluster of 64 machines that demonstrates an order of magnitude improvement in runtime
      is obtained by combining the MR and adaptive complexity models.

    To develop MR algorithms with sublinear adaptivity,
    we first modify and analyze the low-adaptive algorithm
    \lat from \citet{chen2021best} and show that our modification
    satisfies the RCP (Section \ref{sec:pre}). 
    Next, \lat is used to create a low-adaptive greedy procedure \lag;
    \lag then satisfies the RCP, so can be used within \palgshort \cite{barbosa2016new}
    to yield a $1 - 1/e - \epsi$
    approximation in $O( 1 / \epsi^2 )$ MR rounds, with adaptivity $O( \log n )$. We show in Section~\ref{apx:dagfull}
    that \ddist \cite{kazemi2021regularized} also works with any \alg satisfying RCP, which
    yields an improvement in the number of MR rounds to achieve the same ratio. The resulting algorithm
    we term \dagh (Section \ref{sec:DFA}), which achieves ratio $1 - 1/e - \epsi$ with $O( \log^2 n )$ adaptivity and $O( 1 / \epsi )$
    MR rounds.

    Although \dagh achieves nearly the optimal ratio in constant MR rounds and sublinear adaptivity,
    the number of MR rounds is high.
    To obtain truly practical algorithms, we develop constant-factor algorithms with 2 MR rounds: 
    \dsb and \dat
    with $O(\log(k)\log(n))$ and $O(\log n )$ adaptivity, respectively, and with approximation ratios of $\halfopt$ ($\simeq$ 0.316) and $3/8-\epsi$. \dsb is our most practical
      algorithm and may be regarded as a parallelized version of
      \rgreedy \cite{barbosa2015power}, the first MapReduce algorithm for SMCC and the state-of-the-art MR algorithm in terms of empirical performance. \dat is a novel algorithm that improves the theoretical properties of \dsb at the cost of empirical performance (see Table~\ref{table:algs} and Section \ref{sec:exp}). Notably,
the adaptivity of \dat is $\oh{\log(n)}$ which is close to the best known adaptivity for
a constant factor algorithm
of \citet{chen2021best}, which has adaptivity of $O( \log(n / k) )$.

{\revone
  Although our MR algorithms are nearly linear time (within a polylogarithmic factor
  of linear),
  all of the existing MR algorithms are superlinear time, which raises the question of
  if an algorithm exists that has a linear query complexity and has constant MR rounds
  and constant approximation factor. We answer this question affirmatively by adapting
  a linear-time algorithm of \citet{Kuhnle2020a,chen2021best}, and showing that our adaptation
  satisfies RCP (Section \ref{sec:qs}). Subsequently, we develop the first MR algorithm (\dlsshort)
  with an overall linear query complexity.}


Our next contribution is \mg, a general plug-in framework for distributed algorithms, which increases the size of the maximum cardinality constraint at the cost of more MR rounds. As discussed above, the maximum $k$ value of any prior MR algorithm for SMCC
is $O( n / \ell^2 )$, where $\ell$ is the number of machines; \mg increases this to $O( n / \ell )$. We also show that under certain
assumptions on the objective function
(which are satisfied by all of the empirical applications evaluated in Section \ref{sec:exp}),
\mg can be run with $k_{max} = n$, which removes any cardinality
restrictions.
When used in conjunction with a $\gamma$-approximation 
MR 
algorithm, \mg provides an ($1-e^{-\gamma}$)-approximate solution. 


Finally, an extensive empirical evaluation of our algorithms and the current state-of-the-art
on a 64-machine cluster with 32 cores each
and data instances ranging up to 5 million nodes
show that \dsb is orders of magnitude faster than state-of-the-art MR algorithms and demonstrates an exponential improvement in scaling with larger $k$.
Moreover, we show that MR algorithms slow down with increasing number of machines past a certain point, even if enough memory is available, which further motivates distributing a parallelizable algorithm over smaller number of machines.
This observation also serves as a motivation for the development 
of the \mgalg framework. In our evaluation, we found that the 
\mgalg framework delivers solutions significantly faster than 
the \alg for large constraints, showcasing its superior performance 
and efficiency.

{\revone
A previous version of this work was published in a conference \cite{DeyCK23}.
In that version, RCP was claimed to be shown for the \threshold
algorithm of \citet{chen2021best}. Unfortunately, RCP does not hold for
this algorithm. In this version, we provide a modified version
of the \threshold algorithm of \citet{chen2021best}, for which we
show RCP.
Additionally, in this version, we introduce the first linear time MR algorithm, \dls (\dlsshort).
Finally, in our empirical evaluation, we expand upon the conference version by conducting two additional
experiments using a larger cluster comprising 64 machines. These additional experiments aim to evaluate the performance of \dlsshort and investigate the effectiveness of \mgshort in handling increasing cardinality constraints, respectively.}

\subsection{Organization}
The rest of this paper is organized as follows.
In Section \ref{sec:pre}, we present the low-adaptive
procedures and show they satisfy the RCP.
In Section \ref{sec:qs}, we analyze an algorithm
with linear query complexity and show it satisfies
the RCP. In Section \ref{sec:DFA}, we
show how to use algorithms satisfying the RCP
to obtain MR algorithms. In Section \ref{sec:dat},
we improve the theoretical properties of our 2-round MR
algorithm. In Section \ref{sec:DLS_proof2}, we detail
the linear-time MR algorithm. In Section \ref{sec:MG},
we show how to increase the supported constraint size
by adding additional MR rounds. In Section \ref{sec:exp},
we conduct our empirical evaluation. 

\subsection{Preliminaries} \label{sec:prelim} 
A submodular set function captures the diminishing gain property of 
adding an element to a set that decreases with increase in size of 
the set. Formally, let $N$ be a finite set of size $n$. A non-negative set function
$f : 2^{\mathcal N} \rightarrow \mathbb{R}^{+}$ is submodular iff for 
all sets $A \subseteq B \subseteq \mathcal N$ and $x \in \mathcal N \backslash B$, $f(A \cup {x}) - f(A) \ge f(B \cup {x}) - f(B)$. 
The submodular function is monotone if $f(A) \le f(B)$ for all 
$A \subseteq B$.
We use $\marge{x}{S}$ to denote the marginal
game of element $x$ to set $S$: $\marge{x}{S} = f( S \cup \{ x \} ) - f(S)$.
In this paper, we study the following optimization 
problem for submodular optimization under cardinality constraint 
(SMCC):
\begin{equation}
O \gets \arg \max_{S \subseteq \mathcal N} f(S)  \text{; subject to} \ |S| \le k
\end{equation}

\textbf{MapReduce Model.} The MapReduce (MR)
model is a formalization of distributed computation
into MapReduce rounds of communication between machines.
A dataset of size $n$ is distributed
on $\ell$ machines,
each with memory to hold at most $\Psi$ elements of the ground set. 
The total memory of the machines
is constrained to be $\Psi\cdot \ell = O(n)$. After
each round of computation, a machine may
send $O(\Psi)$ amount of data to other machines.
We assume $\ell \le n^{1-c}$ for some
constant $c \ge 1/2$.

Next, we formally define the RCP needed to use
the two frameworks to convert a centralized algorithm
to the MR setting, as discussed previously.
\revtwo
\begin{prop}[Randomized Consistency Property of \citet{barbosa2016new}]\label{prop:consistency} 
  Let $\mathbf{q}$ be a fixed sequence of random bits; and
  $\alg$ be a randomized algorithm with randomness determined
  by $\mathbf{q}$.
  Suppose $\alg(\univ, \mathbf{q})$ returns a pair of sets,
  $(\algsol(\univ, \mathbf{q}),\algrel(\univ, \mathbf{q}))$,
  where $\algsol(\univ, \mathbf{q})$ is the feasible solution 
  and $\algrel(\univ, \mathbf{q})$ is a set providing additional information.
  Let $A$ and $B$ be two disjoint subsets of $\univ$,
  and that for each $b \in B$,
  $\algrel(A\cup\{b\}, \mathbf{q})=\algrel(A, \mathbf{q})$.
  Also, suppose that $\alg( A, \mathbf{q} )$ terminates
  successfully.
  Then $\alg( A \cup B, \mathbf{q} )$ terminates
  successfully and $\algsol( A \cup B, \mathbf{q} ) = \algsol( A, \mathbf{q} )$.
 \end{prop}
 \color{black}

\textbf{Adaptive Complexity.}
The adaptive complexity of an algorithm is the minimum
number of sequential \textit{adaptive rounds} of at most polynomially
many queries, in which the queries to the
submodular function can be arranged, such that the queries
in each round only depend on the results of queries in previous rounds.

{\revone \subsection{Related Work}
\textbf{MapReduce Algorithms.}
Addressing the coverage maximization problem (a special case of SMCC),
\citet{ChierichettiKT10} proposed an approximation algorithm achieving a ratio of $1-1/e$, with polylogarithmic MapReduce (MR) rounds. This was later enhanced by \citet{BlellochST12}, who reduced the number of rounds to $\log^2 n$. \citet{KumarMVV13} contributed further advancements with a $1-1/e$ algorithm, significantly reducing the number of MR rounds to logarithmic levels. Additionally, they introduced a $1/2-\epsilon$ approximation algorithm that operates within $O(1/\delta)$ MR rounds, albeit with a logarithmic increase in communication complexity.

For SMCC, \citet{mirzasoleiman2013distributed} introduced the two-round distributed greedy algorithm (\greedi), demonstrating its efficacy through empirical studies across various machine learning applications. 
However, its worst-case approximation guarantee is $1/\Theta(\min\{\sqrt{k}, \ell\})$.
Subsequent advancements led to the development of constant-factor algorithms by \citet{barbosa2015power}. 
\citet{barbosa2015power} notably introduced randomization into the distributed greedy algorithm \greedi, resulting in the \rgreedy algorithm. 
This algorithm achieves a $\frac{1}{2}(1-1/e)$ approximation ratio 
with two MR rounds and $\oh{nk}$ queries.
\citet{mirrokni2015randomized} also introduced an algorithm with two MR rounds, 
where the first round finded randomized composable core-sets 
and the second round applied a tie-breaking rule.
This algorithm improve the approximation ratio from $\frac{1}{2}(1-1/e)$ to $0.545-\epsi$.
Later on, \citet{barbosa2016new} proposed a $\oh{1/\epsi^2}$ round algorithm
with the optimal $(1-1/e-\epsi)$ approximation ratio without data duplication.
Then, \citet{kazemi2021regularized} improved its space and communication complexity by a factor of $\oh{1/\epsi}$ with the same approximation ratio.
Another two-round algorithm, proposed by \citet{liu2018submodular}, achieved an $1/2-\epsi$ approximation ratio, but requires data duplication with four times more elements distributed to each machine in the first round. As a result, distributed setups have a more rigid memory restriction when running this algorithm. 

\textbf{Parallelizable Algorithms.} A separate line of works
consider the parallelizability of the algorithm,
as measured by its adaptive complexity.
\citet{balkanski2018adaptive} introduced the first
$\oh{\log (n)}$-adaptive algorithm, achieving 
a $(1/3-\epsi)$ approximation ratio with $\oh{nk^2\log^3(n)}$ query complexity for SMC
Then, \citet{balkanski2019exponential} enhanced the approximation ratio 
to $1-1/e-\epsi$ with $\oh{nk^2\log^2(n)}$ query complexity 
while maintaining the same sublinear adaptivity.
In the meantime, \citet{Ene}, \citet{Chekuri2018}, and \citet{Fahrbach2018} 
achieved the same approximation ratio and adaptivity 
but with improved query complexities: 
$\oh{n\text{poly}(\log(n))}$ for \citet{Ene},
$\oh{n\log(n)}$ for \citet{Chekuri2018},
and $\oh{n}$ for \citet{Fahrbach2018}.
Recently, highly practical sublinearly adaptive algorithms have been developed without compromising on the approximation ratio.
\textsc{FAST}, introduced by \citet{Breuer2019}, operates with $\oh{\log(n)\log^2\left(\log(k)\right)}$
adaptive rounds and $\oh{n\log\left(\log(k)\right)}$ queries,
while \textsc{LS+PGB}, proposed by \citet{chen2021best}, runs with $\oh{\log(n)}$
adaptive rounds and $\oh{n}$ queries.
However, none of these algorithms have been shown to satisfy the RCP. 

\textbf{Fast (Low Query Complexity) Algorithms.}
Since a query to the oracle for the submodular
function is typically an expensive operation,
the total query complexity 
of an algorithm is highly relevant to its scalability. 
The work by \citet{BadanidiyuruV14} reduced the query
complexity of the standard greedy algorithm from $\oh{kn}$ to
$\oh{n\log(n)}$, while nearly maintaining the ratio
of $1 - 1/e$. 
Subsequently, \citet{DBLP:conf/aaai/MirzasoleimanBK15} utilized
random sampling on the ground set 
to achieve the same approximation ratio in expectation
with optimal $\oh{n}$ query complexity.
Afterwards, \citet{Kuhnle2020a} obtained a deterministic algorithm
that achieves ratio $1/4$ in exactly $n$ queries, and
the first deterministic algorithm 
with $\oh{n}$ query complexity and $(1-1/e-\epsi)$ approximation ratio.
Since our goal is to develop practical algorithms in this work,
we develop paralellizable and distributed algorithms with nearly linear
query complexity (\ie within
a polylogarithmic factor of linear). Further, we develop the first
linear-time algorithm with constant number of MR rounds, although it doesn't parallelize well. 

\textbf{Non-monotone Algorithms.}
If the submodular function is no longer
assumed to be monotone, the SMCC problem becomes
more difficult. Here, we highlight a few works
in each of the categories of distributed, parallelizable,
and query efficient. 
\citet{DBLP:conf/stoc/LeeMNS09} provided the first constant-factor approximation algorithm
with an approximation ratio of $(1/4-\epsi)$ and a query complexity of $\oh{nk^5\log(k)}$.
Building upon that, \citet{DBLP:conf/wine/GuptaRST10} reduced the query complexity to $\oh{nk}$
by replacing the local search procedure with a greedy approach
that resulted in a slightly worse approximation ratio.
Subsequently, \citet{buchbinder2014submodular} incorporated randomness to achieve an expected
$1/e$ approximation ratio with $\oh{nk}$ query complexity.
Furthermore, \citet{buchbinder2017comparing} improved it to $\oh{n}$ query complexity while maintaining the same approximation ratio in expectation.
For parallelizable algorithms in the adaptive complexity model,
\citet{DBLP:conf/icml/EneN20} achieved the current best $(1/e-\epsi)$ approximation ratio with $\oh{\log(n)}$ adaptivity and $\oh{nk^2\log^2(n)}$ queries to the continuous oracle (a continuous relaxation of the original value oracle).
Later, \citet{DBLP:journals/jair/ChenK24} achieved the current best sublinear adaptivity and nearly linear query complexity with a $(1/6-\epsi)$ approximation ratio.
In terms of distributed algorithms in the MapReduce model,
\rgreedy proposed by \citet{barbosa2015power} 
achieved a $0.1$ approximation ratio within $2$ MR rounds
and $\oh{nk}$ query complexity.
Simulaneously, \citet{mirrokni2015randomized} also
presented a $2$-round algorithm with a $0.14$ approximation ratio
and $\oh{nk}$ query complexity.
Afterwards, \citet{barbosa2016new} proposed a
$2$-round continuous algorithm with a $0.232\left(1-\frac{1}{\ell}\right)$ approximation ratio.
}


\section{Low-Adaptive Algorithms That Satisfy the RCP}\label{sec:pre}
In this section, we analyze two
low-adaptive procedures, 
\lat (Alg. \ref{alg:threshold}), \lag (Alg. \ref{alg:LAG}), 
variants of low-adaptive procedures 
proposed in \citet{chen2021best}.
This analysis enables their use in the distributed,
MapReduce setting. For convenience, we regard the
randomness of the algorithms to be determined by
a sequence of random bits $\mathbf q$, which
is an input to each algorithm.
\begin{algorithm}[t]
  \caption{Low-Adaptive Threshold Algorithm (\lat)}
  \label{alg:threshold}
  \begin{algorithmic}[1]
  \Procedure{\lat}{$f, X, k, \delta, \epsi, \tau, \mathbf q$}
  \State \textbf{Input:} evaluation oracle $f:2^{\mathcal N} \to \reals$, subset $X \subseteq \univ$, constraint $k$, confidence $\delta$, error $\epsi$, threshold $\tau$,
  \revtwo
  a finite set of sequences of random bits $\mathbf q \gets \{\sigma_1, \ldots, \sigma_{M+1}\}$
  \color{black}
  \State Initialize $S_0 \gets \emptyset$ , $R_0\gets \emptyset$ $V_0 \gets X$, 
  $M \gets \left\lceil 4\left(1+\frac{1}{\beta \epsi}\right)\log \left(\frac n \delta\right) \right\rceil$, 
  $\beta \gets \epsi \left(16\log \left(\frac{4}{1-e^{-\epsi/2}}\right)\right)^{-1}$\label{lat:line-M}
  \For{ $j \gets 1$ to $M+1$ } \label{line:threForStart}
    \State Update $V_j \gets \{ x \in V_{j-1} : \marge{x}{S_{j-1}} \ge \tau \}$ \label{line:threFilterV} 
    \If{ $|V_j| = 0$ or $|S_{j-1}| = k$} 
    \revtwo
    \State \textbf{return} $S_{j-1}$, $R_{j-1}$, success \label{line:latReturn}
    \color{black}
    \EndIf
    \State $V_j \gets $ \textbf{random-permutation}$(V_j, \sigma_j)$\label{line:threPermute}
    \State $s \gets \min \{k-|S_{j-1}|, |V_j|\}$
    \revtwo
    \State $\Lambda \gets \left\{1,2,\ldots,\min\left\{s,\left\lceil \frac{1}{\epsi} \right \rceil \right\}\right\} \cup \left\{ \left\lfloor (1+\epsi)^u \right\rfloor: 
    1 \le \left\lfloor (1+\epsi)^u \right\rfloor \leq s
    , u \in \mathbb{N}\right\} \cup \{s\}$ \label{line:lambda}
    \color{black}
    \State $B[\lamb] \gets \textbf{False}, \forall \lamb\in \Lambda$
    \For{$\lamb \in \Lambda$ in parallel} \label{line:innerforStart}
      \State $T_{\lamb} \gets \{v_1, v_2, \ldots, v_{\lamb}\}$ 
      \If{$ \marge{T_{\lamb}}{S_{j-1}}/|T_{\lamb}| \geq  (1-\epsi) \tau $} \label{line:threIf2}
      \State $B[\lamb] \gets \textbf{True}$
      \EndIf
    \EndFor\label{line:innerforEnd}
    \revtwo
    \State $\lamb'_j \gets \min\left\{\lamb \in \Lambda: 
    B[\lamb] = \textbf{False}        
    \text{ OR } \lamb = s\right\}$ \label{line:LAT-index}
    \State $R_j \gets R_{j-1} \cup T_{\lamb'_j}$
    \State \textbf{if} $\lamb_j'\le \left\lceil \frac{1}{\epsi} \right\rceil$
    \textbf{then} $\lamb_j^*\gets \lamb'_j-1$
    \textbf{else} $\lamb_j^*\gets \lamb'_j$ \label{line:LAT-index2}
    \State $S_j\gets S_{j-1} \cup T_{\lamb_j^*}$
    \color{black}
  \EndFor \label{line:threForEnd}
  \revtwo
  \State \textbf{return} $S_{M+1}$, $R_{M+1}$, failure
  \color{black}
  \EndProcedure
\end{algorithmic}
\end{algorithm}

Observe that the randomness of 
both of \lat and \lag
only comes from the random permutations of $V_j$ on
Line \ref{line:threPermute} of \lat, since \lag employs \lat as a subroutine. Consider an equivalent version
of these algorithms in which the entire ground set
$\univ$ is permuted randomly, from which the
permutation of $V_j$ is extracted. That is,
if $\sigma$ is the permutation of $\univ$,
the permutation of $V$ is given by
$v < w$ iff $\sigma(v) < \sigma(w)$,
for $v, w \in V$.
Then, the random vector $\mathbf{q}$ specifies a sequence
of permutations of $\univ$: $(\sigma_1, \sigma_2, \ldots)$,
which completely determines the behavior of both algorithms.

\subsection{Low-Adaptive Threshold (\lat) Algorithm} \label{sec:lat}
This section presents the analysis of the low-adaptive threshold 
algorithm, \lat (Alg. \ref{alg:threshold}; 
a variant of \threshold from \citet{chen2021best}).
\revtwo In \lat, the randomness is explicitly dependent on a random vector $\mathbf{q}$.
This modification ensures the consistency property in the distributed setting.
\color{black}
Besides the addition of randomness \textbf{q},
\lat employs an alternative strategy of prefix selection within the for loop.
Instead of \revtwo identifying \color{black} the \textit{failure point}
(the average marginal gain is under the threshold)
after the final \textit{success point} 
(the average marginal gain is above the threshold)
in \citet{chen2021best},
\lat \revtwo determines \color{black} the first failure point.
This adjustment not only addresses the consistency problem 
but also preserves the theoretical guarantees below.

\begin{theorem} \label{theorem:threshold}
Suppose \lat is run with input $(f,X,k, \delta, \epsi, \tau,\mathbf{q} )$.
Then, the algorithm 
has adaptive complexity $O(\log (n/\delta)/\epsi^3)$ and
outputs \revtwo $S, R \subseteq \mathcal N$,
where $S$ is the solution set with $|S| \le k$ 
and $R$ provides additional information with $|R| = \oh{k}$. 
\color{black} 
The following properties hold:
1) The algorithm succeeds with probability at least $1 - \delta/n$.
2) There are $O(n/\epsi^3)$ oracle queries in expectation.
\revone
\revtwo
3) It holds that $f(S)/|S| \geq (1-2\epsi)\tau / (1+\epsi)$. \color{black}
4) If $|S| < k$, then $\marge{x}{S} < \tau$ for all $x \in \mathcal{N}$.
\end{theorem}

\subsubsection{Analysis of Consistency}\label{sec:lat-consis}
 \begin{lemma} \label{lemma:lat-consistency}
   \lat satisfies the randomized consistency property
   (Property \ref{prop:consistency}).
 \end{lemma}
 \begin{proof}
   Consider that the algorithm runs
   for all $M+1$ iterations of the outer
   \textbf{for} loop; if the algorithm
   would have returned at iteration $j$,
   the values $S_i$
   and $V_i$, for $i > j$ keep their values from 
   when the algorithm would have returned.
   The proof relies upon the fact that
   every call to $\lat( \cdot, \mathbf{q} )$
   uses same sequences of permutations of $\univ$:
   $\{\sigma_1,\sigma_2,\ldots,\sigma_{M+1}\}$.
   We refer to iterations of the outer \textbf{for}
   loop on Line \ref{line:threForStart}
   of Alg. \ref{alg:threshold} simply
   as iterations. 
   Since randomness only happens with the random permutation of
   $\mathcal N$ at each iteration in Alg.~\ref{alg:threshold}, 
   the randomness of $\lat( \cdot, \mathbf{q} )$ is determined by $\mathbf q$,
   which satisfies the hypothesis of Property \ref{prop:consistency}.

   \revtwo For the two sets returned by $\lat(\univ, \mathbf q)$,
   $S_{M+1} = \textsc{ThreshSeqModSol}(\univ, \mathbf q)$ 
   represents the feasible solution,
   and $R_{M+1} = \textsc{ThreshSeqModRel}(\univ, \mathbf q)$
   is the set that provides additional information.
   \color{black}
   We consider the runs of
   (1) $\lat( A, \mathbf q )$,
   (2) $\lat( A \cup \{ b \}, \mathbf q )$,
   and (3) $\lat( A \cup B, \mathbf q )$
   together. Variables of (1) are given the notation defined in
   the pseudocode; variables of (2) are given the
   superscript $b$; and variables of (3) are given
   the superscript $'$.
   \revtwo 
   Suppose for each $b\in B$, $R_{M+1}^{b} = R_{M+1}$.
   The lemma holds if $S_{M+1}' = S_{M+1}$
   and $\lat( A \cup B, \mathbf q )$ terminates successfully.
   \color{black}

   Let $P(i)$ be the
   statement \revtwo
   for iteration $i$ of the outer for loop on Line~\ref{line:threForStart}
   \color{black} that
   \begin{itemize}
   \item[(i)] $S_{i} = S'_{i}$, and
   \revtwo
   \item[(ii)] for all $b \in B$, $S_{i} = S^b_i$, $R_{i} = R^b_i$, and
   \color{black}
   \item[(iii)] $V_{i} = V'_{i} \setminus B$, and
   \item[(iv)] for all $b \in B$, $V_{i} = V^b_i \setminus \{b\}$.
   \end{itemize}
   \revtwo
   If $P(M+1)$ is true, 
   and $\lat( A , \mathbf q )$ and $\lat( A \cup B, \mathbf q )$ terminate at the same iteration, 
   implying that $\lat( A \cup B, \mathbf q )$ also terminates successfully, 
   then the lemma holds immediately.
   In the following, we prove these two statements by induction.

   For $i = 0$, it is clear that $S_0 = S_0' = R_0 = R_0^b = \emptyset$,
   and $V_0 = V_0'\setminus B = V_0^b\setminus \{b\} = \univ\setminus B$ for all $b\in B$.
   Therefore, $P(0)$ is true.
   Next, suppose that $P(i-1)$ is true.
   We show that $P(i)$ is also true,
   and if $\lat( A , \mathbf q )$ terminates at iteration $i$,
   $\lat( A \cup B , \mathbf q )$ also terminates at iteration $i$.

   Firstly, we show that (iii) and (iv) of $P(i)$ hold.
    Since $P(i-1)$ holds,
    it indicates that $S_{i-1} = S'_{i-1} = S_{i-1}^b$ for any $b\in B$ by (i) and (ii) of $P(i-1)$.
    So, (iii) and (iv) of $P(i)$ clearly hold since
    the sets $S_{i-1}, S_{i-1}', S_{i-1}^b$ involved in updating $V_i$, $V_i', V_i^b$
    on Line~\ref{line:threFilterV}
    are equal (after $B$ is removed).

   Secondly, we prove that (i) and (ii) of $P(i)$ hold and that
   if $\lat( A , \mathbf q )$ terminates at iteration $i$,
   $\lat( A \cup B , \mathbf q )$ also terminates at iteration $i$. 

   If $\lat( A , \mathbf q )$ terminates at iteration $i$ because 
   $|S_{i-1}| = k$, 
   then the other two runs of \lat also terminate at iteration $i$,
   since $|S'_{i-1}| = |S_{i-1}^b| = |S_{i-1}|=k$
   by (i) and (ii) of $P(i-1)$.
   Moreover, (i) and (ii) of $P(i)$ are true since $S'_{i}$, $S_{i}^b $, $S_{i}$,
   $R_{i}^b $, and $R_{i}$ do not update.

   If $\lat( A , \mathbf q )$ terminates at iteration $i$ because 
   $V_i = \emptyset$, then by (iii) and (iv) of $P(i)$,
   it holds that $V_i^b\setminus \{b\} = V_i'\setminus B = \emptyset$
   for all $b \in B$.
   Thus, $V_i^b $ must be empty;
   otherwise, $V_i^b = \{b\}$ and $b$ will be added to $R_{i-1}^b$ at iteration $i$
   which contradicts the fact that $b\not \in R_{M+1} = R_{M+1}^b$.
   Therefore, $b$ has been filtered out before iteration $i$ or
   will be filtered out at iteration $i$ in both $\lat( A\cup \{b\} , \mathbf q )$ and $\lat( A\cup B , \mathbf q )$ 
   since the sets $S_{i-1}^b$ and $S_{i-1}^B$ involved in updating $V_i^b$ and $V_i^B$
   are the same.
   Consequently, $V_i'$ is also an empty set and $\lat( A \cup B , \mathbf q )$ terminates
   at iteration $i$.
   Futhermore, (i) and (ii) of $P(i)$ are true since $S'_{i}$, $S_{i}^b $, $S_{i}$,
   $R_{i}^b $, and $R_{i}$ do not update.

    \begin{figure}
        \centering
        \includegraphics[width=0.8\textwidth]{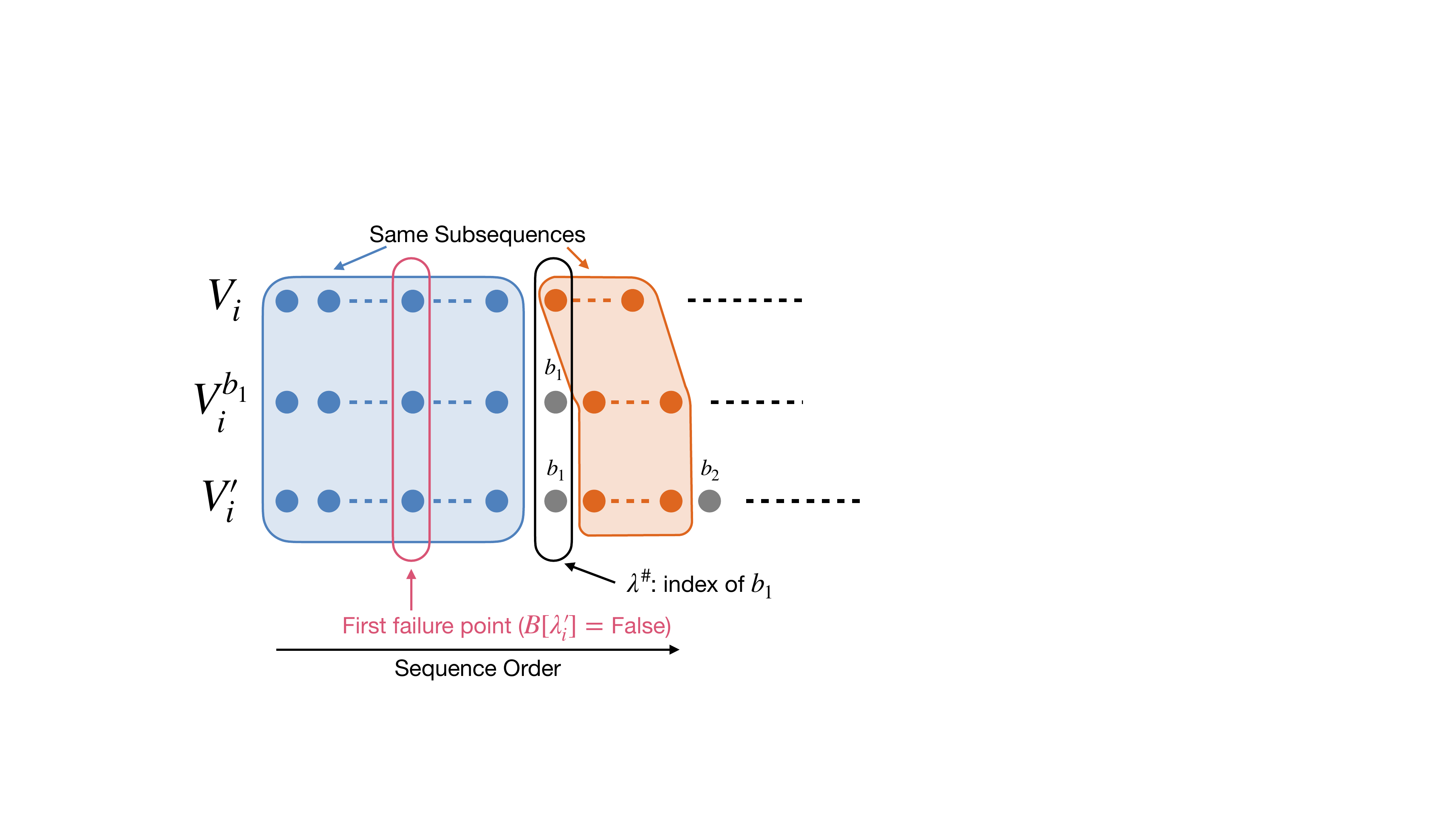}
        \caption{This figure depicts the relationship between
        $V_i$, $V_i^{b_1}$ and $V_i'$ in the circumstance that $|S_{i-1}| < k$, 
        $V_i \neq \emptyset$.}
        \label{fig:latproof}
    \end{figure}
    
    Finally, consider the case that $\lat( A , \mathbf q )$ does not terminate at iteration $i$,
    where $V_i \neq \emptyset$ and $|S_{i-1}|< k$. 
    \color{black}
    By (iii) and (iv),
    it holds that $V_i = V_i'\setminus B= V_i^b \setminus \{b\}$,
    for any $b\in B$.
    Let $b_1 \in B$ be the first element appeared in the sequence of $V_i'$,
    $\lamb^{\#}$ be the index of $b_1$ in the order of the permutation of $V_i'$.
    If such $b_1$ does not exist (\ie $V_i = V_i'=V_i^b$, for any $b\in B$), 
    let $\lamb^{\#} = |V_i|+1$.
    Fig.~\ref{fig:latproof} depicts the relationship between three sequences
    $V_i$, $V_i^{b_1}$ and $V_i'$.
    Let \revtwo
    $T_{\lambda'_i}$, $T_{\lambda'_i}^{b_1}$, $T_{\lambda'_i}'$
    be sets added to $R_{i-1}$, $R_{i-1}^{b_1}$, $R_{i-1}'$ \color{black}
    at iteration $i$ respectively.
    Since $V^{b_1}_i$ and $V_i'$ are subsets of the same random
    permutation of $\mathcal N$, 
    and $V_i'\setminus B= V_i^{b_1} \setminus \{b_1\}$,
    $\lamb^{\#}$ is also the index of $b_1$ in $V_i^{b_1}$.
    Also, it holds that 
    \begin{equation}\label{eq:pf1}
        T_{x} = T_{x}^{b_1} = T_{x}', \forall x < \lamb^{\#}.
    \end{equation}
    Since \revtwo $b_1 \not \in R_{M+1}^{b_1}$,
    it holds that $b_1 \not \in T_{\lamb'_i}^{b_1}$,
    which indicates that $\lamb'_i < \lamb^{\#}$.
    By the definition of $\lamb'_i$ on Line~\ref{line:LAT-index},
    \begin{equation*}
    \left\{
    \begin{aligned}
        &B^{b_1}[x] = \textbf{True}, \forall x < \lamb'_i,\\
        &B^{b_1}[\lamb'_i] = \textbf{False}.
    \end{aligned}
    \right.
    \end{equation*}
    Follows from Equality~\ref{eq:pf1}, and $\lamb'_i < \lamb^{\#}$,
    it holds that
    \begin{equation*}
    \left\{
    \begin{aligned}
        &B[x] = B'[x] = \textbf{True}, \forall x < \lamb'_i,\\
        &B[\lamb^{*}_i] = B'[\lamb'_i] = \textbf{False},
    \end{aligned}
    \right.
    \end{equation*}
    Therefore, $T_{\lambda^*_i}=T_{\lambda^*_i}^b=T_{\lambda^*_i}'$,
    and even further $T_{\lambda'_i}=T_{\lambda'_i}^b=T_{\lambda'_i}'$.
    So, Property (i) and (ii) hold.
     \color{black}

    \color{black}

 \end{proof}

\subsubsection{Analysis of Guarantees}
There are \revtwo three \color{black} adjustments in \lat compared with \threshold in \citet{chen2021best} altogether.
\revtwo First, we made the randomness $\mathbf q$ explicit \color{black}
allowing us to still consider each permutation
on Line~\ref{line:threPermute} as a random step.
\revtwo Therefore, the analysis of the theoretical guarantees is not influenced 
by adopting the randomness $\mathbf q$. \color{black}
Second, the prefix selection step is changed from \revtwo identifying \color{black} the
failure point after the final success point
to \revtwo identifying \color{black} the initial failure point.
\revtwo Note that, this change is necessary for the algorithm to satisfy RCP (Property\ref{prop:consistency}). \color{black}
However, by making this change, 
fewer elements could be filtered out at the next iteration.
Fortunately, we are still able to filter out a constant fraction
of the candidate set.
\revtwo
Third, the algorithm returns one more set that provides extra information.
Since the additional set does not affect the functioning of the algorithm,
it does not influence the theoretical guarantees.
In the following, we will provide the detailed analysis for the theoretical guarantees.
\color{black}

As demonstrated by Lemma~\ref{lemma:latinc} below, 
an increase in the number of selected elements results in 
a corresponding increase in the number of filtered elements.
Consequently, there exists a point, say $t$, such that
a given constant fraction of $V_j$ can be filtered out
if we are adding more than $t$ elements.
The proof of this lemma is quite similar to the proof of
Lemma 12 in~\citet{chen2021best} which can be found in
Appendix~\ref{apx:lat}.
\begin{restatable}{lemma}{lemmalatinc}
\label{lemma:latinc}
    At an iteration $j$, let $A_i=\{x \in V_j: \marge{x}{T_i\cup S_{j-1}}<\tau\}$
    after Line~\ref{line:threFilterV}.
    It holds that $|A_0| = 0$, $|A_{|V_j|}| = |V_j|$ and $|A_i| \le |A_{i+1}|$.
\end{restatable}

\revtwo Although we use a smaller prefix based on Line~\ref{line:LAT-index} and~\ref{line:LAT-index2} \color{black} compared to \threshold,
it is still possible that
with probability at least $1/2$,
a constant fraction of $V_{j}$ will be filtered out at the beginning
of iteration $j+1$, 
or the algorithm terminates with $|S_{j}|=k$;
given as Lemma~\ref{lemma:latprob} in the following.
Intuitively, if there are enough
such iterations, \lat will succeed.
Also, the calculation of query complexity comes from Lemma~\ref{lemma:latprob} directly.
Since the analyses of success probability, adaptivity and query complexity
simply follows the analysis in~\citet{chen2021best}, 
we provide these analyses in Appendix~\ref{apx:lat}.
Next, we prove that Lemma~\ref{lemma:latprob} always holds.
\begin{lemma}\label{lemma:latprob}
    At any iteration $j+1$, with probability at least $1/2$, 
    either $|V_{j+1}| \le (1-\beta\epsi)|V_j|$
    after Line~\ref{line:threFilterV},
    or the algorithm terminates with $|S_{j}| = k$ on Line~\ref{line:latReturn}.
\end{lemma}
\begin{proof}
    By the definition of $A_i$ in Lemma~\ref{lemma:latinc}, 
    $A_{\lamb_j^*}$ will be filtered out from $V_j$ at the next iteration.
    Equivalently, $A_{\lamb_j^*} = V_j \setminus V_{j+1}$.
    After Line~\ref{line:threPermute} at iteration $j+1$,
    by Lemma \ref{lemma:latinc},
    there exists a $t$ such that 
    $t = \min\{i \in \mathbb N: |A_i| \ge \beta\epsi |V_{j}|\}$.
    We say an iteration $j$ succeeds,
    if $\lamb_j^* \ge \min\{s, t\}$.
    In this case, it holds that either $|V_{j}\setminus V_{j+1}| \ge \beta\epsi |V_{j}|$,
    or $|S_j|=k$.
    \revone
    In other words, iteration $j$ succeeds if there does not exist
    $\lamb \le \min\{s, t\}$ such that $B[\lamb] = \textbf{False}$
    by the selection of $\lamb_j^*$ on \revtwo Line~\ref{line:LAT-index2}.
    \color{black}
    Instead of directly calculating the success probability,
    we consider the failure probability.
    Define an element $v_i$ \textit{bad} if $\marge{v_i}{S_{j-1}\cup T_{i-1}} < \tau$,
    and \textit{good} otherwise.
    Consider the random permutation of $V_j$ as a sequence of
    dependent Bernoulli trials, with \textit{success} if the element is bad
    and \textit{failure} otherwise.
    When $i \le \min\{s, t\}$, the probability of $v_i$ is bad is less than $\beta\epsi$.
    If $B[\lamb] = \textbf{False}$, there are at least $\epsi\lamb$ bad elements
    in $T_\lamb$.
    Let $\{Y_i\}_{i=1}^{\infty}$ be a sequence of independent and identically distributed Bernoulli trials,
    each with success probability $\beta \epsi$.
    \revtwo
    Next, we consider the following sequence of dependent Bernoulli trials $\{X_i\}_{i=1}^\infty$,
    where $X_i = 1_{v_i \text{ is bad}}$ when $i\le t$, and $X_i = Y_i$ otherwise.
    Define $B'[\lamb] = \textbf{False}$ if there are at least $\epsi \lamb$ bad elements in $\{X_i\}_{i=1}^\lamb$.
    Then, for any $\lamb \le \min\{s, t\}$, it holds that 
    $\prob{B[\lamb] = \textbf{False}} = \prob{B'[\lamb] = \textbf{False}}$.
    In the following, we bound 
    the probability that there are more than $\epsi$-fraction $1$s in
    $\{X_i\}_{i=1}^\lamb$ for any $\lamb \le s$,
    \color{black}
    \begin{align*}
        \revtwo \prob{B'[\lamb] = \textbf{False}} 
        &\revtwo \le \prob{\sum_{i=1}^{\lamb} X_i \ge \epsi\lamb}\\
        &\le \prob{\sum_{i=1}^{\lamb} Y_i \ge \epsi\lamb} \tag{Lemma~\ref{lemma:indep}}\\
        &\le \min\left\{\beta, e^{-\frac{(1-\beta)^2}{1+\beta}\epsi\lamb}\right\}
        \tag{Markov's Inequality, Lemma~\ref{lemma:chernoff}}\\
        &\le \min\left\{\beta, e^{-\epsi\lamb/2}\right\}\tag{$\beta < 1/16$}
    \end{align*}
    Subsequently, the probability of failure at iteration $j$ is calculated below,
    \begin{align*}
        \prob{\text{iteration } j \text{ fails}}
        &= \prob{\exists \lamb \le \min\{s, t\}, \st, B[\lamb] = \textbf{False}}\\
        &\revtwo \le \prob{\exists \lamb \le s, \st, B'[\lamb] = \textbf{False}}\\
        &\revtwo \le \sum_{\lamb\le s} \prob{B'[\lamb] = \textbf{False}} \\
        &\le \sum_{\lamb=1}^\infty \min\left\{\beta, e^{-\epsi\lamb/2}\right\}\\
        &\le a\beta + \sum_{\lamb=a+1}^\infty e^{-\epsi\lamb/2}
        \text{, where } a = \left\lfloor\frac{1}{4\beta} \right\rfloor=
        \left\lfloor \frac{2}{\epsi} \log\left(\frac{4}{1-e^{-\epsi/2}}\right) \right\rfloor\\
        & \le a\beta+\frac{e^{-\epsi(a+1)/2}}{1-e^{-\epsi/2}}
        \le \frac{1}{4}+\frac{1}{4} = \frac{1}{2}
    \end{align*}
\end{proof}
As for the Properties (3) and (4) in Theorem~\ref{theorem:threshold},
since the structure of the solution set is the same as 
the solution of \threshold.
So, the analysis is similar to what in~\citet{chen2021best}.
We provide these proofs in Appendix~\ref{apx:lat}.


\subsection{Low-Adaptive Greedy (\lag) Algorithm} \label{sec:LAG}

\begin{algorithm}[t]
   \caption{Low-Adaptive Greedy (\lag)}
   \label{alg:LAG}
   \begin{algorithmic}[1]
     \State \textbf{Input:} Evaluation oracle $f:2^{\mathcal N} \to \reals$, subset $C$, constraint $k$, 
     fixed sequence of random bits $(\mathbf q_1, \mathbf q_2, \ldots)$,
     accuracy parameter $\epsi$, constant $\alpha$, value $\Gamma$ such that $\Gamma \leq f(O) \leq \Gamma/\alpha$ 
     \If{$\Gamma = \varnothing$ or $\alpha = \varnothing$}
        \State $\Gamma \gets \max_{x \in \univ} f(x)$
        \State $\alpha \gets 1/k$
     \EndIf
     \State Initialize
     $\tau \gets \Gamma/\alpha k$,
     $\delta$ $\gets$ $1/(\log_{1-\epsi}(\alpha/3)) + 1)$, 
     $S \gets \emptyset$,\revtwo $R\gets \emptyset$,\color{black}
     $i \gets 0$
     \For{$i \gets 0$ to $\log_{1-\epsi}(\alpha/3)$} \label{line:boostWhileStart}
     \State $\tau \gets \tau (1-\epsi)^i$, $g(\cdot) \gets f(S \cup \cdot)$
     \State \revtwo $T_1, T_2 \gets \lat( g, C, k-|S|, \delta, \epsi/3, \tau, \mathbf{q}_{i})$\label{line:boostQuery}
     \State $S \gets S \cup T_1$, $R\gets R\cup T_2$ \label{line:boostUpdateA}
     \If{$|S| = k$}
     \State \textbf{return} $S, R$
     \EndIf
     \EndFor\label{line:boostWhileEnd}
     \State \textbf{return} $S, R$
   \end{algorithmic}
 \end{algorithm}

Another building block for our distributed algorithms
is a simple, low-adaptive greedy algorithm \lag
(Alg. \ref{alg:LAG}).
This algorithm is an instantiation of the \boost
framework of \citet{chen2021best}, and it relies heavily on
the low-adaptive procedure \lat (Alg. \ref{alg:threshold}).

\textbf{Guarantees.} \ 
Following the analysis of Theorem 3 in \citet{chen2021best},
\lag achieves ratio $1-1/e - \epsi$ in $O\left( \frac{\log (\alpha^{-1}) \log (n\log(\alpha^{-1})/\epsi)}{\epsi^4} \right)$
adaptive rounds and $O(\frac{n \log (\alpha^{-1})}{\epsi^4})$ total queries
with probability at least $1-1/n$.
 Setting $\Gamma = \max_{x \in \univ} f(\{x\})$ and $\alpha = 1/k$,
 \lag has adaptive complexity of $O\left( \frac{\log n \log k}{\epsi^4} \right)$
 and a query complexity of $O(\frac{n \log k}{\epsi^4})$


 \begin{restatable}{lemma}{lagconsistency}
  \label{lemma:lag-consistency}
  \lag satisfies the randomized consistency property (Property \ref{prop:consistency}).
 \end{restatable}
 \begin{proof}[Proof of Lemma \ref{lemma:lag-consistency}]
  Observe that the only randomness in \lag is from the
  calls to \lat. Since $\lag(A, \mathbf q)$ succeeds,
  every call to \lat must succeed as well. Moreover,
  considering that $\mathbf q$ is used to permute the underlying
  ground set $\univ$, changing the set argument $A$ of $\lag$
  does not change the sequence received by each call to $\lat$. Order the calls to $\lat$: $\lat( \cdot, \mathbf q_1 )$,
  $\lat( \cdot, \mathbf q_2 )$, \ldots, $\lat( \cdot, \mathbf q_m )$. 
  \revtwo
  Since $\lag\textsc{Rel}( A\cup \{b\}, \mathbf{q}) = \lag\textsc{Rel}( A, \mathbf{q})$,
  it holds that $b\not \in \lat\textsc{Rel}( A \cup \{ b \}, \mathbf q_i)$, 
  \color{black}
  for each $b \in B$ and each $i$. 
  \revtwo If $b\not \in \lat\textsc{Rel}( A \cup \{ b \}, \mathbf q_i)$ is equivalent to
  $\lat\textsc{Rel}( A \cup \{ b \}, \mathbf q_i) = \lat\textsc{Rel}( A, \mathbf q_i)$ (as shown in Claim~\ref{claim:lat-equal} in the following), \color{black}
  then by application of Lemma \ref{lemma:lat-consistency},
  $\lat( A \cup B, \mathbf q_i )$ terminates successfully
  and $\lat\textsc{Sol}( A \cup B, \mathbf q_i ) = \lat\textsc{Sol}( A, \mathbf q_i )$ for each $i$. Therefore, $\lag\textsc{Sol}(A \cup B, \mathbf{q}) = \lag\textsc{Sol}( A, \mathbf{q})$ and the former call terminates successfully.
\end{proof}
\revtwo
Next, we provide Claim~\ref{claim:lat-equal} and its analysis below.
\begin{claim}
\label{claim:lat-equal}
Let $\mathbf q$ be a fixed sequence of random bits,
$A \subseteq \univ$, and $b\in \univ\setminus A$.
Then, $b\not \in \lat\textsc{Rel}( A \cup \{ b \}, \mathbf q)$
if and only if $\lat\textsc{Rel}( A \cup \{ b \}, \mathbf q) = \lat\textsc{Rel}( A, \mathbf q)$.
\end{claim}
\begin{proof}
It is obvious that if $\lat\textsc{Rel}( A \cup \{ b \}, \mathbf q_i) = \lat\textsc{Rel}( A, \mathbf q)$,
then $b\not \in \lat\textsc{Rel}( A \cup \{ b \}, \mathbf q)$.
In the following, we prove the reverse statement.

Following the analysis of RCP for \lat in Section~\ref{sec:lat-consis},
we consider the runs of (1) $\lat( A, \mathbf q )$,
and (2) $\lat( A \cup \{ b \}, \mathbf q )$.
Variables of (1) are given the notation defined in
the pseudocode; variables of (2) are given the
superscript $b$.
We analyze the following statement $P(i)$ for each iteration $i$ 
of the outer for loop on Line~\ref{line:threForStart} in Alg.~\ref{alg:threshold}.
\begin{itemize}
   \item[(i)] for all $b \in B$, $S_{i} = S^b_i$, $R_{i} = R^b_i$, and
   \item[(ii)] for all $b \in B$, $V_{i} = V^b_i \setminus \{b\}$.
   \end{itemize}
Since $b\not \in \lat\textsc{Rel}( A \cup \{ b \}, \mathbf q) = R^b_{M+1}$,
then $b\not \in R^b_i$ for each $i$.
Thus, the analysis in Section~\ref{sec:lat-consis} also holds in this case.
According to the inductive method, we are able to prove that $P(M+1)$ is true for each $i$ which indicates $\lat\textsc{Rel}( A \cup \{ b \}, \mathbf q) = \lat\textsc{Rel}( A, \mathbf q)$.
\end{proof}
\color{black}
\section{Consistent Linear Time (\qs) Algorithm}\label{sec:qs}
In this section, 
we present the analysis of \qs (Alg. \ref{alg:qs}, \qsshort), 
a consistent linear time algorithm which is an extension of the
highly adaptive linear-time algorithm (Alg. 3) in~\citet{chen2021best}.
Notably, to ensure consistency within a 
distributed setting and bound the solution size,
\qsshort incorporates the randomness \textbf{q},
and initializes the solution set with the maximum singleton.
The algorithm facilitates the creation of linear-time MapReduce algorithms, 
enhancing overall computational efficiency beyond the 
capabilities of current state-of-the-art superlinear algorithms.
\begin{algorithm}[t]
  \caption{\qs (\qsshort)}
  \label{alg:qs}
  \begin{algorithmic}[1]
  \Procedure{\qsshort}{$f, \mathcal N, k, \mathbf q$}
  \State \textbf{Input:} evaluation oracle $f$, constraint $k$, fixed sequence of random bits $\mathbf q \gets (\sigma)$.
  \State $\mathcal N \gets \textbf{random-permutation}(\mathcal N, \sigma)$
  \State Initialize $a \gets \text{first element in }\argmax_{x \in \mathcal N}{ f(x)}$ , $S \gets \{a\}$\label{line:qsini}
  \For{ $x \in \mathcal N$ } \label{line:qsfor}
        \State \textbf{if} {$\marge{x}{S}\ge \frac{f(S)}{k} $} \textbf{then}
        $S\gets S \cup \{x\}$ \label{line:qsif}
    \EndFor
    \State \textbf{return} $S$
      \EndProcedure
\end{algorithmic}
\end{algorithm}

\begin{theorem}\label{thm:qs}
    Let $(f,k)$ be an instance of \sm.
    The algorithm \qs outputs $S\subseteq \mathcal N$
    such that the following properties hold:
    1) There are $O(n)$ oracle queries 
    \revone
    and $O(n)$ adaptive rounds.
    \color{black}
    2) Let $S'$ be the last $k$ elements in $S$.
    It holds that $f(S')\ge \frac{1}{2} f(S)\ge \frac{1}{4} f(O)$,
    where $O$ is an optimal solution to the instance $(f,k)$.
    3) The size of $S$ is limited to $O(k\log(n))$.
\end{theorem}

\subsection{Analysis of Consistency}
The highly adaptive linear-time algorithm (Alg. 3) outlined in~\citet{chen2021best}
commences with an empty set and adds elements to it iff.
$\marge{x}{S} \ge f(S)/k$.
Without introducing any randomness,
this algorithm can be deterministic only if the order
of the ground set is fixed.
Additionally, to limit the solution size,
we initialize the solution set with the maximum singleton,
a choice that also impacts the algorithm's consistency.
However, by selecting the first element that maximizes
the objective value, the algorithm can maintain its consistency.
In the following, we provide the analysis of randomized consistency property of $\qsshort$.
\begin{lemma}\label{lemma:qscst}
    \qsshort satisfies the randomized consistency property (Property~\ref{prop:consistency}).
\end{lemma}
\begin{proof}
    Consider the runs of (1) $\qsshort(A, \textbf{q})$, 
    (2) $\qsshort(A\cup \{b\}, \textbf{q})$, 
    and (3) $\qsshort(A\cup B, \textbf{q})$.
    With the same sequence of random bits \textbf{q},
    after the random permutation,
    $A$, $A\cup \{b\}$ and $A \cup B$
    are the subsets of the same sequence.
    For any $x \in \mathcal N$, let $i_x$ be the index of $x$.
    Then, $\mathcal N_{i_x}$ are the elements before $x$ (including $x$).
    Define $S_{i_x}$ be the intermediate solution of (1) after we consider element $x$.
    If $x \not \in A$, define $S_{i_x} = S_{i_x-1}$.
    Similarly, define $S_{i_x}'$ and $S_{i_x}^b$ be the intermediate solution of (2) and (3).
    If $S_{i_x}=S_{i_x}^b=S_{i_x}'$, for any $x \in \univ$ and $b \in B$,
    $\qsshort(A\cup B, \textbf{q})= \qsshort(A, \textbf{q})$.
    Next, we prove the above statement holds.

    For any $b\in B$,
    since $\qsshort(A, \textbf{q})=\qsshort(A\cup \{b\}, \textbf{q})$,
    it holds that either $f(b) < \max_{x\in A}f(x)$,
    or $f(b)=\max_{x\in A}f(x)$, and $i_a < i_b$, where $a$
    is the first element in $\arg\max_{x\in A}f(x)$.
    Therefore, $a$ is also the first element in $\arg\max_{x\in A\cup B}f(x)$.
    Furthermore, it holds that $S_0 = S_0^b = S_0' = \{a\}$.

    Suppose that $S_{i_x-1}=S_{i_x-1}^b=S_{i_x-1}'$.
    If $x\in A$ or $x \not \in A\cup B$,
    naturally, $S_{i_x}=S_{i_x}^b=S_{i_x}'$.
    If $x \in B$, $S_{i_x}=S_{i_x-1}$ immediately.
    Since $\qsshort(A, \textbf{q})=\qsshort(A\cup \{b\}, \textbf{q})$,
    it holds that $\marge{x}{S_{i_x-1}^b} < f(S_{i_x-1}^b)/k$ and
    $S_{i_x}^b=S_{i_x-1}^b$.
    So, $\marge{x}{S_{i_x-1}'} < f(S_{i_x-1}')/k$.
    Then, $S_{i_x}'=S_{i_x-1}'$.
    Thus, $S_{i_x}=S_{i_x}^b=S_{i_x}'$.
\end{proof}
\clearpage
\subsection{Analysis of Guarantees}
\revone
\textbf{Query Complexity and Adaptivity.}
Alg.~\ref{alg:qs} queries on Line~\ref{line:qsini} and~\ref{line:qsif},
where there are $n$ oracle queries on Line~\ref{line:qsini}
and 1 oracle query for each element received on Line~\ref{line:qsif}.
Therefore, the query complexity and the adaptivity are both $O(n)$.
\color{black}

\textbf{Solution Size.}
Given that $f$ is a monotone function,
it holds that $\arg\max_{S\in \univ}f(S) = \univ$.
Furthermore, 
since $f(\{a\}) = \max_{x \in \mathcal N} f(\{x\})$, the 
submodularity property implies that 
$f(\{a\}) \ge \frac{1}{n}\sum_{x \in \mathcal N}f(\{x\})\ge \frac{1}{n} f(\univ)$. 
Each element added to $S$ contributes to an increase in the solution
value by a factor of $(1+\frac{1}{k})$.
Therefore,
\begin{align*}
    &f(\univ) \ge f(S) \ge \left(1+\frac{1}{k}\right)^{|S|}f(a)
    \ge \left(1+\frac{1}{k}\right)^{|S|}\frac{f(\univ)}{n}\\
    \Rightarrow & |S|\le \log_{1+\frac{1}{k}}(n)\le (k+1)\log(n)
    \tag{$\log(x) \ge 1-\frac{1}{x}$}
\end{align*}

\textbf{Objective Value.}
For any $x \in \univ$,
let $S_x$ be the intermediate solution before we process $x$.
It holds that $\marge{o}{S_o} < f(S_o)/k$
for any $o \in O\setminus S$. Then,
\[f(O)-f(S) \overset{(1)}{\le}\sum_{o \in O\setminus S}\marge{o}{S}
\overset{(2)}{\le} \sum_{o \in O\setminus S}\marge{o}{S_o}
< \sum_{o \in O\setminus S} f(S_o)/k\le f(S),\]
where Inequality (1) follows from monotonicity and submodularity,
Inequality (2) follows from submodularity.

For any $x\in S'$, it holds that $\marge{x}{S_x} \ge f(S_x) / k$.
Then, 
\begin{align*}
    &f(S)-f(S\setminus S') = \sum_{x \in S'}\marge{x}{S_x}
    \ge \sum_{x \in S'} f(S_x) / k
    \overset{(1)}{\ge} \sum_{x \in S'} f(S\setminus S') / k
    =f(S\setminus S')\\
    \Rightarrow &\hspace{1em} f(S') \overset{(2)}{\ge} f(S)-f(S\setminus S') \ge \frac{1}{2} f(S)
\end{align*}
where Inequalities (1) and (2) follow from submodularity.

 \begin{algorithm}[t]
  \caption{\dsbfull (\dsb)}
  \label{alg:DSB}
  \begin{algorithmic}[1]
     \State {\bfseries Input:} Evaluation oracle $ f:2^{\mathcal N} \to \mathbb{R}$, constraint $k$, error $\epsi$, available machines $M \gets \{1, 2,...,\ell\}$
     
     \For{$e \in \univ$ }\label{line:assignEle}
       \State Assign $e$ to machine chosen uniformly at random
     \EndFor
     \For{$i \in M$}\label{alg1:distLS}
     \LineComment{On machine $i$}
     \State Let $\univ_i$ be the elements assigned to machine $i$
     \State \revtwo $S_i, R_i \gets \lag (f, \mathcal N_i, k, \epsi )$ \label{alg1:Ti}      
     \State Send $S_i, R_i$ to primary machine
     \EndFor
     \LineComment{On primary machine}
     \State Gather $R \gets \bigcup_{i=1}^{\ell} R_i$ \label{alg1:T} 
     \State $T \gets \lag (f, R, k,  \epsi)$ on machine $m_1$\label{line:A}
     \color{black}
     \State \textbf{return} $V \gets \argmax{\{ f(T), f(S_1)}\} $
  \end{algorithmic}
\end{algorithm}


\section{Low-Adaptive Algorithms with Constant MR Rounds (\dagfull and \dsbfull)} \label{sec:DFA}
Once we have the randomized consistency property of \lag, we can
almost immediately use it to obtain parallelizable MapReduce algorithms.

\subsection{\dsbfull}
\dsb (Alg. \ref{alg:DSB}) 
is a two MR-rounds
algorithm obtained by plugging \lag into the \rgreedy
algorithm of
\citet{mirzasoleiman2013distributed,barbosa2015power}. 
\dsb runs in two MapReduce rounds,
$O(\log{}(k)\log{}(n))$ adaptive rounds, and guarantees the ratio 
$\halfopt$ ($\simeq$ 0.316). 

\textbf{Description.} \ The ground set is initially distributed at 
random by \dsb across all machines $M$. In its first MR round, 
\dsb runs \lag
on every machine
to obtain \revtwo $S_i, R_i$ \color{black} in $\oh{\log(k)\log(|\univ_i|)}$ adaptive rounds. 
The solutions from every machine are then returned to the primary machine, 
where \lag 
selects the output solution 
that guarantees $\halfopt$ 
approximation in \revtwo $\oh{\log(k) \log (|R|)}$ \color{black} adaptive rounds
as stated in Corollary \ref{cor:dash}.
First, we provide Theorem~\ref{thm:dash} and its analysis
by plugging any randomized algorithm \alg which follows
RCP (Property~\ref{prop:consistency})
into the framework of \rgreedy.
The proof is a minor modification 
of  the proof of \citet{barbosa2015power} to incorporate randomized consistency.
Then, we can get Corollary~\ref{cor:dash} for \dsb immediately.

 \begin{restatable}{theorem}{dash} \label{thm:dash}
  Let $(f,k)$ be an instance of \sm where \revtwo $k = \oh{\psi/\ell}$. \color{black}
  \alg is a randomized algorithm which satisfies the
  Randomized Consistency Property with $\alpha$ approximation ratio,
  $\Phi(n)$ query complexity, and $\Psi(n)$ adaptivity.
  By replacing \lag with \alg,
    \dsb returns set $V$ with two MR rounds,
    $2\Psi(n/\ell)$
    adaptive rounds,
    $(\ell +1)\Phi(n/\ell)$ total queries,
    $O(n)$ communication complexity
     such that
    \[ \E[ f(V) ] \ge \frac{\alpha}{ 2 } \opt. \]
  \end{restatable}
   \begin{proof}
   \revone
   \textbf{Query Complexity and Adaptivity.}
   \dsb runs with two MR rounds. 
   In the first MR round, \alg is invoked
   $\ell$ times in parallel, each with a ground set size of $n/\ell$.
   So, during the first MR round, the number of queries is $\ell \Phi(n/\ell)$ and the number of adaptive rounds is $\Psi(n/\ell)$.
   Then, the second MR round calls \alg once, handling at most $n/\ell$ elements.
   Consequently, the total number of queries of \dsb is $(\ell +1)\Phi(n/\ell)$,
   and the total number of adaptive rounds is $2\Psi(n/\ell)$.
   
    \textbf{Approximation Ratio.}
    \color{black}
    Let \dsb be run with input
    $(f,k,\epsi, M)$. 
    Since $\ell \le n^{1-c}$,
    $\ex{ |\univ_i| } = n / \ell \ge n^c$.
    By an application of Chernoff's bound to
    show that the
    size $|\univ_i|$ is concentrated.
  
    Let $\univ(1/\ell)$ denote the random
    distribution over subsets of $\univ$
    where each element is included independently
    with probability $1/\ell$.
    For $x \in \univ$, let \revtwo
    $$p_x = \begin{cases}
      \Prob_{X \sim \univ(1 / \ell), \mathbf q} 
      \left[ \algrel( X \cup \{x\}, \mathbf q) = \algrel( X, \mathbf q ) \right]  \text{, if } x \in O \\
      0 \hspace{13em} \text{, otherwise }
    \end{cases}.$$ \color{black}
    Consider \revtwo $S_1 = \algsol( \univ_1, \mathbf{q} )$
    and $R_1 = \algrel( \univ_1, \mathbf{q} )$ on machine $m_1$.
    Let $O_1 = \{ o \in O : \algrel( \univ_1 \cup \{ o \}, \mathbf q ) = \algrel( \univ_1, \mathbf q ) \}$.  \color{black}
    Since \alg follows the Randomized Consistency Property,
    \revtwo
    $\algsol( \univ_1 \cup O_1, \mathbf q ) = \algsol(\univ_1, \mathbf q) = S_1$.
    \color{black}
    Therefore, 
    $f( S_1 ) \ge \alpha f( O_1 )$.
    Next, \revtwo let $O_2 = O \cap R$, where $R = \bigcup_{i=1}^{\ell} R_i$. \color{black}
    It holds that $f( T ) \ge \alpha f( O_2 )$.
    Let $o \in O$; $o$ is assigned to $\univ_c$ on some machine
    $m_c$. It holds that, \revtwo
    \begin{align*}
      \Pr[ o \in O_2 ] & = \Pr[ o \in \algrel( \univ_c ) | o \in \univ_c ] \\
                     & = \Pr[ o \in \algrel( \univ_c \cup \{ o \} )] \\
                     &= 1 - p_o. 
    \end{align*}
    \color{black}
    Therefore,
    \begin{align*}
      \ex{f(V)} &\ge \frac{1}{2}\left(\E[ f(S_1) ] + \E[ f(T) ]\right)\\
      &\ge \frac{\alpha}{2}\left( \E[ f( O_1 ) ] +\E[ f( O_2 ) ]\right)\\
       &\ge \frac{\alpha}{2}\left(F( \mathbf p ) + F( \mathbf{1}_O - \mathbf p ) \right)\\
      &\ge \frac{\alpha}{2} F( \mathbf{1}_O ) 
      ,\numberthis \label{inq:dash-lovasz}
    \end{align*}
  where Inequality \ref{inq:dash-lovasz} follows from Lemma \ref{lemma:lovasz}
  and $F$ is convex. 
  \end{proof}

\begin{corollary}\label{cor:dash}
    Let $(f,k)$ be an instance of \sm where \revtwo $k = \oh{\psi/\ell}$. \color{black}
    \dsb returns set $V$ with two MR rounds,
    $O\left(\frac{1}{\epsi^4} \log(k)\log\left(n\right)\right)$
    adaptive rounds,
    $O\left(\frac{n \log(k)}{\epsi^4}\right)$ total queries,
    $O(n)$ communication complexity,
    and probability at least $1-n^{1-2c}$
     such that
    \[ \E[ f(V) ] \ge \frac{1-1/e-\epsi}{ 2 } \opt. \]
\end{corollary}

\subsection{\dagfull} \label{apx:dagfull}
\begin{algorithm}[t]
   \caption{$\dagfull (\dagh)$}
   \label{alg:DAG}
  \begin{algorithmic}[2]
     \State {\bfseries Input:} evaluation oracle 
     $ f:2^{\mathcal N} \to \mathbb{R}$, constraint $k$, 
     error $\epsi$, available machines $M \gets \{1, 2,...,\ell\}$
     \State $ S\gets \emptyset, C_0 \gets \emptyset $
    \For{$r \gets 1$ to $\lceil\frac{1}{\epsi} \rceil$ }\label{alg2:iterRounds}
        \State $X_{r,i} \gets$ Elements assigned to machine $i$ chosen uniformly at random in round $r$
        \State $ \mathcal N_{r,i} \gets X_{r,i} \cup C_{r-1}$
        \For{$i \in M$ in parallel }\label{alg2:distGB}
        \revtwo
          \State $S_{r,i}, R_{r,i} \gets \lag (f, \mathcal N_{r,i}, k, \epsi )$ \label{alg2:AlgSoli}
          \State Send $S_{r,i}, R_{r,i}$ to each machine
          \color{black}
        \EndFor
        \State $S \gets \argmax{\{ f(S), f(S_{r,1}), \cdots , f(S_{r,\ell})}\} $
        \revtwo
        \State $ C_r \gets \bigcup_{i=1}^{\ell} R_{r,i} \cup C_{r-1} $ \label{alg2:AlgRel}
        \color{black}
        \EndFor
     \State \textbf{return} \textit{S}
  \end{algorithmic}
\end{algorithm}
Next, we obtain the nearly the optimal ratio by applying \lag and
randomized consistency to
\ddist proposed by \citet{kazemi2021regularized}. \ddist is a distributed 
algorithm for regularized submodular maximization
that relies upon (a distorted version of) the
standard greedy algorithm. For SMCC, the distorted
greedy algorithm reduces
to the standard greedy algorithm \sg.
In the following, we show that any algorithm
that satisfies randomized consistency
can be used in place of \sg as stated in Theorem~\ref{thm:GDASH}.
By introducing \lag into \ddist, \dagh
achieves the near optimal ($1 - 1/e - \epsi$) 
expected ratio in $\lceil \frac{1}{\epsi} \rceil$ MapReduce rounds,
$O( \log (n) \log (k) )$ adaptive rounds, and $O(n \log(k ))$
total queries. We also generalize this to any algorithm
that satisfies randomized consistency. 

\textbf{The Framework of \citet{kazemi2021regularized}.} \ 
\ddist has $\lceil1/\epsi \rceil$ MR rounds where each round $r$ works as follows:
First, it distributes the ground set into $m$ machines uniformly at random.
Then, each machine $i$ runs 
\sg (when modular term $\ell(\cdot)=0$) 
on the data $\univ_{r,i}$ that combines the
elements distributed before each round $X_{r,i}$
and the elements forwarded from the previous rounds $C_{r-1}$
to get the solution $S_{r,i}$. 
At the end, the final solution,
which is the best among $S_{r,1}$ and all the previous solutions, is returned.
To improve the adaptive rounds of \ddist, we replace standard greedy with \lag
to get \dagh.
\begin{restatable}{theorem}{thmGDASH}
\label{thm:GDASH}
    Let $(f,k)$ be an instance of \sm where \revtwo $k = \oh{\epsi\psi/\ell}$. \color{black}
    \alg is a randomized algorithm which satisfies the
    Randomized Consistency Property with $\alpha$ approximation ratio,
    $\Phi(n)$ query complexity, and $\Psi(n)$ adaptivity.
    By replacing \lag with \alg,
    \dagh returns set $S$ with $\oh{1/\epsi}$ MR rounds and 
    $\frac{1}{\epsi}\Psi\left(\frac{n}{\ell}\right)$
    adaptive rounds,
    $\frac{\ell}{\epsi}\Phi\left(\frac{n}{\ell}\right)$ total queries,
    $O\left(\frac{n}{\epsi}\right)$ communication complexity
    such that,
    \[\ex{f(S)} \ge (\alpha-\epsi) \opt.\]
\end{restatable}
\begin{proof}
\revone
\textbf{Query Complexity and Adaptivity.}
\dagh operates with $\lceil \frac{1}{\epsi} \rceil$ MR rounds,
where each MR round calls \alg $\ell$ times in parallel.
As each call of \alg works on the ground set with a size of at most $\frac{n}{\ell}$,
the total number of queries for \dagh is $\frac{\ell}{\epsi}\Phi\left(\frac{n}{\ell}\right)$,
and the total number of adaptive rounds is $\frac{1}{\epsi}\Psi\left(\frac{n}{\ell}\right)$.

\textbf{Approximation ratio.}
\color{black}
Let \dagh be run with input $(f,k,\epsi, M)$. 
Similar to the analysis of Theorem~\ref{thm:dash},
The size of $\univ_i$ is concentrated.
Let $O$ be the optimal solution.
For any $x \in \mathcal N$,
define that,
$$p_x^r = \begin{cases}
  \Prob_{X \sim \univ(1 / \ell), \mathbf q} 
  \left[ x \not \in C_{r-1} \text{ and }\right. 
  \\ \revtwo \hspace{1em}  \left.x \in \algrel(X \cup C_{r-1} \cup \{x\}, \mathbf{q}) \right] 
  \hspace{1em}\text{, if } x \in O \\
  0 \hspace{13em} \text{, otherwise }
\end{cases}.$$
Then we provide the following lemma.
\begin{lemma}\label{lemma:dagprob}
For any $x \in O$ and $1 \le r \le 1/\epsi$,
$\prob{x \in C_r} = \sum_{r'=1}^{r} p_x^{r'}$.
\end{lemma}
\begin{proof}
\begin{align*}
\prob{x \in C_r} &= \sum_{r'=1}^r \prob{x \in C_{r'} \backslash C_{r'-1}}\\
&= \revtwo \sum_{r'=1}^r \prob{x \in \cup_{i=1}^{\ell} R_{r',i}\backslash C_{r'-1}}\\
&=\revtwo \sum_{r'=1}^r \sum_{i=1}^{\ell}\frac{1}{\ell}
\prob{x \in R_{r',i}\backslash C_{r'-1}| x \in \mathcal{N}_{r',i}}\\
&=\sum_{r'=1}^r \sum_{i=1}^{\ell}\frac{1}{\ell}
\Prob \left( x \not \in C_{r'-1} \text{ and }\right. \\
       & \revtwo \quad \left. x \in \algrel( X_{r',i}\cup C_{r'-1} \cup \{x\}, \mathbf q)\right) \\
&=\sum_{r'=1}^r p_x^{r'}
\end{align*}
\end{proof}
The rest of the proof bounds $f(S_{r,1})$ in the following two ways.

First, let \revtwo $O_{r,1} = \{o \in O: o \not \in 
\algrel\left(X_{r,1} \cup C_{r-1}\cup \{o\}, \mathbf q\right) \}$, \color{black}
and $O_{r,2} = \left(C_{r-1} \cap O \right)\cup O_{r,1}$.
Since \alg follows the Randomized Consistency Property,
it holds that \revtwo 
$\algsol(X_{r,1} \cup C_{r-1}, \mathbf q)=
\algsol(X_{r,1} \cup C_{r-1} \cup O_{r,1}, \mathbf q)=S_{r,1}$, \color{black}
and $f\left(S_{r,1}\right) \ge \alpha f\left(O_{r,2}\right)$.
So, for any $o \in O$, \revtwo
\begin{align*}
\prob{o \in O_{r,2}} = \Prob \left(o \in C_{r-1} \text{ or }
o \not \in \algrel\left(X_{r,1} \cup C_{r-1}\cup \{o\}, \mathbf q\right)\right)
= 1- p_o^{r}.
\end{align*}
\color{black}
Therefore,
\begin{align*}
\ex{f\left(S_{r,1}\right)}\ge \alpha \ex{f\left(O_{r,2}\right)}
\ge \alpha F\left(\mathbf{1}_O - \mathbf{p}^r\right).\numberthis \label{eq:o2}
\end{align*}

Second, let $O_{r,3} = C_{r-1} \cap O$.
Similarly, it holds that
$f\left(S_{r,1}\right) \ge \alpha f\left(O_{r,3}\right)$.
And for any $o \in O$, by Lemma~\ref{lemma:dagprob}, it holds that,
\begin{equation*}
\prob{o \in O_{r,3}} = \prob{o \in C_{r-1}}
= \sum_{r'=1}^{r-1} p_x^{r'},
\end{equation*}
Therefore,
\begin{align*}
\ex{f\left(S_{r,1}\right)}\ge \alpha\ex{f\left(O_{r,3}\right)}
\ge \alpha F\left(\sum_{r'=1}^{r-1} \mathbf{p}^{r'}\right).\numberthis \label{eq:o3}
\end{align*}

By Inequalities~\ref{eq:o2} and \ref{eq:o3},
we bound the approximation ratio of \dagh by the following,
\begin{align*}
\ex{f(S)} & \ge \epsi \sum_{r=1}^{1/\epsi} \ex{f\left(S_{r,1}\right)}\\
& \ge \epsi \alpha \cdot
\left(F\left(\sum_{r'=1}^{1/\epsi-1} \mathbf{p}^{r'}\right) + 
 \sum_{r=1}^{1/\epsi-1}F\left(\mathbf{1}_O - \mathbf{p}^{r'}\right)\right)\\
 &\overset{(a)}{\ge} \alpha(1-\epsi) F\left(\mathbf{1}_O\right)\\
 &\ge (\alpha-\epsi) \opt,
\end{align*}
where Inequality (a) follows from Lemma~\ref{lemma:lovasz} and $F$ is convex.
\end{proof}
\begin{corollary}\label{cor:GDASH}
    Let $(f,k)$ be an instance of \sm where \revtwo $k = \oh{\epsi\psi/\ell}$. \color{black}
  \dagh returns set $S$ with $\oh{1/\epsi}$ MR round and 
  $O\left(\frac{1}{\epsi^5}\log(k)\log(n)\right)$ 
  adaptive rounds,
  $O\left(\frac{n \log(k)}{\epsi^5}\right)$ total queries,
  $O\left(\frac{n}{\epsi}\right)$ communication complexity,
  and probability at least $1-\epsi^{-1}n^{1-2c}$
  such that,
  \[\ex{f(S)} \ge (1-1/e-\epsi)\opt.\]
\end{corollary}

\section{\datfull: Two MR-Rounds Algorithm with Improved Theoretical Properties}\label{sec:dat}
In this section, we present a novel $(0.375 - \epsi$)-approximate, two MR-rounds algorithm \datfull (\dat)
that 
achieves nearly optimal $O(\log(n))$ adaptivity in $O(n\log (k))$ total time complexity.

\textbf{Description.} \ The \dat algorithm (Alg. \ref{alg:DATOPT})
is a two MR-rounds algorithm that runs \lat concurrently on every machine with
a specified threshold value of $\alpha \opt / k$.
The primary machine then builds up its solution
$S_1$ by adding elements with $\lat$ from the pool
of solutions returned by the other machines.
Notice that there is a small amount of data duplication
as elements of the ground set
are not randomly partitioned in the same way as in the
other algorithms.
This version of the algorithm
requires to know the $\opt$ value; in Appendix~\ref{sec:DATOPT} we show how to remove
this assumption. 
In Appendix~\ref{app:Vondrak}, we further discuss the similar two MR round algorithm of ~\citet{liu2018submodular}, that provides an improved $1/2$-approximation but requires four times the data duplication of \dat.


\begin{algorithm}[t]
  \caption{\datfull Knowing \opt}
  \label{alg:DATOPT}
  \begin{algorithmic}[1]
     \State {\bfseries Input:} Evaluation oracle $ f:2^{\mathcal N} \to \mathbb{R}$, constraint $k$, error $\epsi$, available machines $M \gets \{1, 2,...,\ell\}$, and \opt
          \State Initialize $\delta \gets 1/(\ell+1)$, $\mathbf q \gets$ a fixed sequence of random bits.
          \State Set $\alpha \gets \frac{3}{8}$, 
     $(\mathbf q_1, \ldots, \mathbf q_{\ell+1}) \gets \mathbf q$
     \For{$e \in \univ$ do }\label{alg0:assignEle}
       \State Assign $e$ to each machine independently with probability $1/\ell$
     \EndFor
      
     \For{$i \in M$}\label{alg0:distLAT}
     \LineComment{On machine $i$}
     \State Let $\univ_i$ be the elements assigned to machine $i$
     \State \revtwo $S_{i}, R_i \gets \lat (f, \mathcal N_i, k, \delta, \epsi,(\alpha+\epsi)\opt/k,\mathbf q_i)$ \label{alg0:Ti} \color{black}
     \State \textbf{if} $|S_i| = k$ \textbf{then} return $T'\gets S_i$
     \State \revtwo Send $S_{i}, R_i$ to primary machine \color{black}
     \EndFor
     \LineComment{On primary machine} 
     \State \revtwo $R \gets \bigcup_{i=1}^{\ell} R_i$ \color{black}
     \State Let $g(\cdot) \gets f(S_1 \cup \cdot) - f(S_1)$
     
     \State \revtwo $T \gets \lat (g, R, k-|S_1|,  \delta, \epsi, (\alpha+\epsi)\opt/k, \mathbf q_{\ell+1})$ \color{black} \label{alg0:T}
     
     \State $T' \gets S_1 \cup T$ \label{alg0:TFinal}
     \State \textbf{return} $T'$
  \end{algorithmic}
\end{algorithm}
\begin{theorem}
Let $(f,k)$ be an instance of \sm where \revtwo $k = \oh{\frac{\psi}{\ell}}$. \color{black}
  \dat knowing \opt returns set $T'$ with two MR rounds,
  $O\left(\frac{1}{\epsi^3}\log\left(n\right)\right)$
  adaptive rounds,
  $O\left(\frac{n}{\epsi^3}\right)$ total queries,
  $O(n)$ communication complexity,
  and probability at least $1-n^{-c}$
   such that
  \[\ex{f(T')} \ge \left(\frac{3}{8}-\epsi\right)\opt.\]
\end{theorem}
\begin{proof}
\revone
\textbf{Query Complexity and Adaptivity.}
\dat runs with two MR rounds, where the first MR round invokes \lat $\ell$ times in parallel, and the second MR round invokes \lat once.
By Theorem~\ref{theorem:threshold}, 
each call of \lat with a ground set of size at most $\frac{n}{\ell}$
queries $\oh{\frac{n}{\ell\epsi^3}}$ within $\oh{\frac{\log(n/\ell)}{\epsi^3}}$ adaptive rounds. 
The total number of queries is $\oh{\frac{n}{\epsi^3}}$,
and the total number of adaptive rounds is $\oh{\frac{\log(n)}{\epsi^3}}$.

\textbf{Approximation Ratio.}
\color{black}
  In Algorithm~\ref{alg:DATOPT}, 
  there are $\ell+1$ independent calls of \lat.
  With $|\mathcal N_{i}|\ge n^c$,
  the success probability of each call of \lat
  is larger than $1-\frac{1}{n^c(\ell+1)}$.
  Thus, Algorithm~\ref{alg:DATOPT} succeeds
  with probability larger than $1-n^{-c}$.
  For the remainder of the analysis, 
  we condition on the event that all calls to \lat succeed.

  In the case that $|T'| = k$, by \revtwo Theorem~\ref{theorem:threshold}
  in Section~\ref{sec:lat} and $\tau = \left(\frac{3}{8}+\epsi\right)\frac{\opt}{k}$, it holds that
  $f(T') \ge \frac{1-2\epsi}{1+\epsi}\tau \cdot k \ge
   \left(\frac{3}{8} - \epsi\right)\opt .$ \color{black}
  Otherwise, we consider the case that $|T'| < k$ in the following.
  \revtwo Let $(\textsc{TSMSol}(\univ, \mathbf q), \textsc{TSMRel}(\univ, \mathbf q))$ be the pair of sets returned by $\lat(\univ, \mathbf q)$.
  For any $x \in \univ$, let
    \[p_x = \begin{cases}
      \Prob_{X \sim \univ(1 / \ell), \mathbf q} 
      \left[ \textsc{TSMRel}( X \cup \{x\}, \mathbf q)\right. 
      \\\hspace{6em}  \left. = \textsc{TSMRel}( X, \mathbf q ) \right] 
      \hspace{1em}\text{, if } x \in O \\
      0 \hspace{13em} \text{, otherwise }
    \end{cases}.\]
  Let $O_1 = \{o \in O: o \not \in \textsc{TSMRel} (N_1 \cup \{o\}, q)\}$, $O_2 = R \cap O$. 
  By Randomized Consistency Property, it holds that
  $S_1 = \textsc{TSMSol} (N_1, q) = \textsc{TSMSol} (N_1 \cup O_1, q)$. \color{black}
  For any $o \in O_1$, $o$ is not selected in $S_1$. Since, $|S_1| < k$, 
  \revtwo by Property (4) in Theorem~\ref{theorem:threshold}, \color{black} 
  $\marge{o}{T'} < \marge{o}{S_1} < \tau.$
  Also, for any $o \in O_2\backslash T$, $o$ is not selected in $T$. Similarly, $\marge{o}{T'} < \tau.$
  Then, we can get,
  \begin{align*}
    &f(O_1 \cup O_2) - f(T') \le f(O_1 \cup O_2 \cup T') - f(T')\\
    & \le \sum_{o \in O_1 \cup O_2 \backslash T'}\marge{o}{T'}
     \le k \cdot \tau = \revtwo \left(\frac{3}{8}+\epsi\right)\opt.\numberthis \label{ineq:dat1}
  \end{align*}
  Next, we provide the following lemma to complete the rest of the proof.
\begin{restatable}{lemma}{unionprob}
\label{lemma:union_prob}
For any $o \in O$, it holds that 
   $\prob{o \in O_1 \cup O_2} \ge 3/4$.
\end{restatable} 
  Then,
  by Lemma~\ref{lemma:lovasz} in Appendix~\ref{sec:lov},
  $\ex{f(O_1 \cup O_2)} \ge 3/4 \cdot F(\mathbf{1}_O)
  =3/4\cdot \opt.$
  From this and Inequality \ref{ineq:dat1},
  we can bound the approximation ratio for Algorithm \dat
  knowing \opt by follows, \revtwo
  $\ex{f(T')} \ge \left(\frac{3}{8}-\epsi\right)\opt.$ \color{black}
\end{proof}

 \section{\dls (\dlsshort)} \label{sec:DLS_proof2}
 \begin{algorithm}[t]
   \caption{\dls}
   \label{alg:dls2}
   \begin{algorithmic}[1]
      \State {\bfseries Input:} Evaluation oracle $ f:2^{\mathcal N} \to \mathbb{R}$, constraint $k$, error $\epsi$, available machines $M \gets \{1, 2,...,\ell\}$
      \For{$e \in \univ$ }\label{line:assignEle2}
        \State Assign $e$ to machine chosen uniformly at random
      \EndFor
      \For{$i \in M$}\label{alg1:distLS2}
      \LineComment{On machine $i$}
      \State Let $\univ_i$ be the elements assigned to machine $i$
      \State $S_i \gets \qsshort (f, \mathcal N_i, k,  \mathbf q)$ \label{alg1:Ti2}      
      \State Send $S_i$ to primary machine
      \EndFor
      \LineComment{On primary machine}
      \State Gather $S \gets \bigcup_{i=1}^{\ell} S_i$ \label{alg1:T2} 
      \State $T_1 \gets \qsshort (f, S, k, \mathbf q)$ on machine $m_1$\label{line:A2}
      \State $T_1' \gets $ last $k$ elements in $T_1$
      \State $T_2 \gets \tg\left(f, T_1, k, \epsi, f(T_1'), 1/2\right)$ \Comment{Post-Processing}
      \State $S_1' \gets$ last $k$ elements added to $S_1$
      \State \textbf{return} $V \gets \argmax{\{f(S_1'), f(T_1'), f(T_2)}\} $
   \end{algorithmic}
  \end{algorithm}

In this section, we introduce \dlsshort, the first linear 
time MapReduce algorithm with constant approximation. 
We derive \dlsshort from \dsbshort by incorporating the analysis of 
\qsshort (Alg. \ref{alg:qs}), an algorithm with $O(n)$ query complexity, 
into the distributed setting. 
This integration allows \dlsshort 
to operate within two MapReduce rounds, 
demonstrating an adaptive complexity of 
$O(\frac{n}{\ell})$ and a query complexity of $O(n)$. 
Also, as stated in Theorem \ref{thm:dash2}, 
it guarantees an approximation ratio of $\frac{1}{8}$. 
Additionally, we improve the objective value by implementing
\tg (Alg.~\ref{alg:tg} in Appendix~\ref{apx:tg}) on the set returned at the second MapReduce round.
The integration of \qsshort into \dls 
provides enhanced capabilities and improves the 
efficiency of the overall algorithm.

\textbf{Description:} 
Initially, \dlsshort involves randomly distributing the 
ground set across all machines $M$. In the first MR round, 
\dls applies \qsshort on each machine to obtain $S_i$ 
within $\oh{n/\ell}$ query calls. 
The solutions from all machines are then returned to the primary machine, 
where \qsshort selects the output solution $T_1$.
By selecting the best last $k$ among all solutions, 
\dlsshort ensures a $\frac{1}{8}$-approximation. 
Besides returning the last $k$ elements of the solution $T_1$
which is an $1/2$-approximation of $T_1$,
we employ \tg, a $(1-1/e-\epsi)$ approximation algorithm,
to boost the objective value.
We provide the theoretical guarantees of \dlsshort as Theorem \ref{thm:dash2}.
The proof involves a 
minor modification of the proof provided in Theorem \ref{thm:dash}.

 \begin{restatable}{theorem}{dash2} \label{thm:dash2}
  Let $(f,k)$ be an instance of \sm where $k < \frac{\psi}{\ell\log(\psi)}$.
    \dlsshort returns a set $V$ with two MR rounds,
    $\oh{\frac{n}{\ell}}$
    adaptive rounds,
  $\oh{n}$
  total queries,
    $O(n)$ communication complexity
     such that
    \[ \E[ f(V) ] \ge \frac{1}{8} \opt. \]
  \end{restatable}
\revone
\subsection{Analysis of Query Complexity and Adaptivity}
\dlsshort runs with two MR rounds.
   In the first MR round, \qsshort is invoked
   $\ell$ times in parallel, each with $\oh{n/\ell}$ queries and $\oh{n/\ell}$ adaptive rounds by Theorem~\ref{thm:qs}.
   So, during the first MR round, the number of queries is $\oh{n}$ and the number of adaptive rounds is $\oh{n/\ell}$.
   Then, the second MR round calls \qsshort and \tg once respectively, handling at most $n/\ell$ elements.
   By Theorem~\ref{thm:qs} and~\ref{thm:tg}, the number of queries is $\oh{n/\ell}$ and the number of adaptive rounds is $\oh{n/\ell}$.
   Consequently, the total number of queries of \dlsshort is $\oh{n}$,
   and the total number of adaptive rounds is $\oh{n/\ell}$.

\color{black}
\subsection{Analysis of Approximation Ratio}
   Let \dlsshort be executed with input $(f,k,\epsi, M)$. 
   Since $\ell \le n^{1-c}$, it follows that the expected 
   size of each subset $|\univ_i|$ satisfies 
   $\ex{ |\univ_i| } = n / \ell \ge n^c$. 
   To ensure that the size $|\univ_i|$ is concentrated, 
   we apply Chernoff's bound. 
   Let $\univ(1/\ell)$ denote the random
   distribution over subsets of $\univ$
   where each element is included independently
   with probability $1/\ell$. Let $\mathbf p \in [0, 1]^n$ be the following vector.
   For $x \in \univ$, let
   $$p_x = \begin{cases}
     \Prob_{X \sim \univ(1 / \ell), \mathbf q} 
     \left[ \qsshort( X \cup \{x\}, \mathbf q) = \qsshort( X, \mathbf q ) \right] 
     \text{, if } x \in O \\
     0 \hspace{13em} \text{, otherwise }
   \end{cases}.$$
   Consider $S_1 = \qsshort( \univ_1, \mathbf{q} )$ on machine $m_1$.
  Let $O_1 = \{ o \in O : \qsshort( \univ_1 \cup \{ o \}, \mathbf q ) = \qsshort( \univ_1, \mathbf q ) \}$. By Lemma \ref{lemma:qscst},
  $\qsshort( \univ_1 \cup O_1, \mathbf q ) = S_1$.
  Therefore, by Theorem~\ref{thm:qs}, it holds that $f( S_1' ) \ge  f( O_1 )/4$.
  Next, let $O_2 = O \cap S$, where $S = \bigcup_{i=1}^{\ell} S_i$.
  Similarly, by Theorem~\ref{thm:qs}, it holds that $f( T_1' ) \ge f( O_2 )/4$.
  Let $o \in O$; $o$ is assigned to $\univ_c$ on some machine
  $m_c$. It holds that,
   \begin{align*}
     \Pr[ o \in O_2 ] &= \Pr[ o \in \qsshort( \univ_c ) | o \in \univ_c ] \\
                    &= \Pr[ o \in \qsshort( \univ_c \cup \{ o \} )] \\
                    &= 1 - p_o. 
   \end{align*}
 
 Therefore,
  \begin{align*}
    \ex{f(V)} &\ge \frac{1}{2}\left(\E[ f(S_1') ] + \E[ f(T_1') ]\right)\\   
    &\ge \frac{1}{2}\left(\E[ f( O_1 )/4 ] + \E[ f( O_2 )/4 ]\right)\\
    &\ge \frac{1}{8}\left(F( \mathbf p ) + F( \mathbf{1}_O - \mathbf p ) \right)\\
    &\ge \frac{1}{8} F( \mathbf{1}_O ) = \frac{1}{8} \opt,\numberthis \label{inq:dash-lovasz-2}
  \end{align*}
 where 
 inequality \ref{inq:dash-lovasz-2} follows from Lemma \ref{lemma:lovasz}
 and $F$ is convex.
  
\subsection{Post-Processing}
In this algorithm, we employ a simple post-processing procedure.
Since $T_1'$ is an $1/2-$approximation of $T_1$ with size $k$,
$T_1'$ is also an $1/2-$approximation of the instance $(f,k)$ on
the ground set $T_1$.
By running any linear time algorithm that has a better approximation ratio
on $T_1$,
we are able to boost the objective value returned by the algorithm
with the same theoretical guarantees.
By Theorem~\ref{thm:tg}, \tg achieves $(1-1/e-\epsi)-$approximation
in linear time with a guess on the optimal solution.
Therefore, with input $\alpha=1/2$ and $\Gamma=f(T_1')$,
utilizing \tg can effectively enhance the objective value.

\section{Towards Larger $k$: A Memory-Efficient, Distributed Framework (\mg)} \label{sec:MG}
In this section, we propose a general-purpose plug-in framework for distributed algorithms, \mgfull (\mg, Alg. \ref{alg:MetaDASH}). \mg increases
the largest possible constraint value from $O( n / \ell^2 )$ to $O(n / \ell)$ in the value oracle model. Under some
additional assumptions, we remove all restrictions on the constraint value. 


\def\NoNumber#1{{\def\alglinenumber##1{}\State #1}\addtocounter{ALG@line}{-1}}
\begin{algorithm}[t]
\caption{$\mg (f, k, M, \alg, \mathbf q )$}
\label{alg:MetaDASH}
\begin{algorithmic}[1]
 \State {\bfseries Input:} evaluation oracle $ f:2^{\mathcal N} \to \mathbb{R}$, constraint $k$, available machines $M \gets \{1, 2,...,\ell\}$, MR algorithm \alg, random bits $\mathbf q$
  \State $n \gets |\mathcal N|$, $\Psi \gets$ Memory capacity (\# elements) of primary machine
 \State Choose $k' \gets \max \{k' \in \mathbb{N} : k' \le \frac{\Psi}{\ell}\}$
 \State $m \gets \lceil k / k' \rceil$, $(\mathbf q_1,\ldots, \mathbf q_m) \gets \mathbf q$
 \For{$i \gets 1 \ to \ m$  }\label{line:iterDASH}
   \State Let $g(\cdot) \gets f(\cdot \cup S_i) - f(S_i)$
   \State $A_i \gets \alg (g, M, \min\{k', k - |S_i|\}, \mathbf{q}_i )$
   \State $S_{i+1} \gets S_i \cup A_i$

 \EndFor
 \State \textbf{return} $S_m$
 
\end{algorithmic}
\end{algorithm}


As discussed in Section \ref{intro}, 
the $k$ value is limited to a fraction of the machine memory
for MapReduce algorithms: $k \le O( n / \ell^2 )$,
since those algorithms need to merge a group of solutions and pass
it to a single machine. \mg works around this limitation as follows:
\mg can be thought of as a greedy algorithm
that uses an approximate greedy selection through $\alg$.
One machine manages the partial solutions $\{ S_i \}$, built up
over $m$ iterations. 
In this way, each call to $\alg$ is within the constraint
restriction of \alg, \ie $O( n / \ell^2 )$, but a larger solution
of up to size $O(n / \ell )$ can be constructed.

The restriction on $k$ of \mg of $O(n / \ell)$ comes from passing the data of the current
solution to the next round.
Intuitively, if we can send some auxiliary information about the function
instead of the elements selected, the $k$ value can be unrestricted.
\begin{assume}\label{assume:1}
  Let $f$ be a set function with ground set $\univ$ of size $n$. If for all $S \subseteq \univ$,
  there exists a bitvector $\mathbf v_S$, such that
  the function $g( X ) = f( S \cup X ) - f(S)$ can be
  computed from $X$ and $\mathbf v_S$, then $f$ satisfies Assumption~\ref{assume:1}.
\end{assume}
We show in Appendix~\ref{appendix:obj} that all four applications evaluated in Section \ref{sec:exp}
satisfy Assumption~\ref{assume:1}. As an example, consider MaxCover,
which can be expressed as $f(S) = \sum_{i \in \mathcal N}f_i(S)$,
where $f_i(S) = \mathbf{1}_{\{i \text{ is covered by } S\}}$.
Let $g_i(X) = f_i(S \cup X) - f_i(S)$,
and $\mathbf v_S = (f_1(S), \ldots, f_n(S))$.
Then, $f_i(S \cup X) = \mathbf{1}_{\{i \text{ is covered by } S\}} \lor 
\mathbf{1}_{\{i \text{ is covered by } X\}}$,
where the first term is given by $\mathbf v_S$ and the second term
is calculated by $X$.
Therefore, since $g(X) = \sum_{i \in \mathcal N} g_i(X)$,
$g(X)$ can be computed from $X$ and $\mathbf v_S$.

\begin{restatable}{theorem}{metagreedy}
  \label{thm:metagreedy}
  Let ($f,k$) be an instance of SMCC distributed over $\ell$ machines. 
  For generic objectives, where the data of current
  solution need to be passed to the next round,
  $k \le \min\{\frac{\ell\psi-n}{\ell-1}+1, \psi-\ell+1\}$;
  for special objectives, where only one piece of compressed data
  need to be passed to the next round,
  $k \le n$.
  Let $S_{m}$ be the set returned by \mg . 
  Then 
  $\ex{f(S_{m})} \ge (1-e^{-\gamma})\opt$, 
  where $\gamma$ is the expected approximation of \alg.
\end{restatable}
\begin{proof}
  Let $O$ be an optimal
  and $O_1, O_2, ..., O_{m}$ be a partition of $O$ into $m$ pieces, each of size $\le k'$. Also
  $\forall_i \hat{f}_i$ is monotone SMCC since $f$ is monotone SMCC.

  \begin{align*}
    \ex{f(S_{i+1}) - f(S_i)|S_i} & = \ex{f(S_i \cup A_i) - f(S_i)|S_i} \nonumber \\
                        & = \ex{\hat{f}_i(A_i)|S_i} \nonumber \\
                        & \overset{(a)}{\ge} \frac{1}{m} \sum_{j=1}^{m} \gamma \hat{f}_i(O_j) \nonumber \\
                        & \ge \frac{\gamma}{m} \hat{f}_i(O) \nonumber\\
                        & = \frac{\gamma}{m} (f(S_i \cup O) - f(S_i)) \nonumber \\
                        & \ge \frac{\gamma}{m} (OPT - f(S_i)) \nonumber \label{eq:mg},  \nonumber
  \end{align*}
  where Inequality (a) follows from  $\hat{f}_i$ is SMCC; 
  \alg is $\gamma$-approximation.
  Unfix $S_i$, it holds that
  \begin{align*}
    OPT - \ex{f(S_{i+1})} & \le  \left( 1 - \frac{\gamma}{m}\right) [OPT - \ex{f(S_i) }] \nonumber \\
   & \le \left( 1 - \frac{\gamma}{m}\right) ^{i} [OPT - f(\emptyset)] \nonumber \\
   & = \left( 1 - \frac{\gamma}{m}\right) ^{i} \cdot OPT \nonumber \\ \nonumber \\
      \therefore \ex{f(S_{m})} & \ge \left[ 1 - (1 - \frac{\gamma}{m})^m\right] \cdot OPT \nonumber \\                       
      & \ge (1 - e^{-\gamma}) \cdot OPT 
  \end{align*}

  To analyze the memory requirements of \mg,
  consider the following. To compute
  $g(\cdot) \gets f(\cdot \cup S_i) - f(S_i)$,
  the current solution $S_i$ need to be passed
  to each machine in $M$ for the next call of \alg.
  Suppose $|S| = k - x$ where $1\le x \le k$.
  The size of data stored on any non-primary machine in cluster $M$,
  as well as on the primary machine of $M$ can be bounded as follows:
  \begin{equation*}
  \begin{aligned}
  &k-x+(n-k+x)/\ell \le \Psi\\
  &k-x+(k\ell)/m \le \Psi
  \end{aligned}
  \, \Rightarrow\,
  \begin{aligned}
  &k \le (\ell \Psi-n)/(\ell-1)+1\\
  &k \le \Psi - \ell +1
  \end{aligned}.
  \end{equation*}
  Therefore, if $\Psi \ge 2n/ \ell$, \mg can
  run since $\ell \le n / \ell$ in our MapReduce model.
  Under the alternative assumption, it is clear
  that \mg can run for all $k \le n$. 
\end{proof}


\section{Empirical Evaluation} \label{sec:exp}

This section presents empirical comparison of \dsb, \dat, \dagh, \dlsshort
to the state-of-the-art 
distributed algorithms. Distributed algorithms 
include \rgreedy (\rgshort) of \citet{barbosa2015power}
and \ddist (\ddshort) of 
\citet{kazemi2021regularized}.

\begin{table}[]
  \centering
  \begin{threeparttable}
      \caption{Environment Setup} \label{table:datasets}
    \begin{tabular}{l|C{2.5cm}|C{2.5cm}|C{2.5cm}}
      \hline
      \textbf{Experiment}                       & \textbf{Number of Machines ($\ell$)\tnote{*}} & \textbf{Number of Threads Per Machine ($\varUpsilon$)} & \textbf{Dataset Size (n)} \\ 
      \hline
      Experiment 1 (Fig.~\ref{fig:Fig1})        & 8                                    & 4                                                      & 10K - 100K                \\
      Experiment 2 (Fig.~\ref{fig:FigExp1J1})   & 64                                   & 32                                                     & 1M - 5M                   \\
      Experiment 3 (Fig.~\ref{fig:scaleB})      & 8                         & 4                                                      & 3M                        \\
      Experiment 3 (Fig.~\ref{fig:scaleC})      & 8                                & 4                                                      & 100K                      \\
      Experiment 4 (Fig.~\ref{fig:FigExp2J1})   & 32                                   & 1                                                      & 100K                      \\ 
      \hline
    \end{tabular}
    \begin{tablenotes}\footnotesize
      \item[*] Prior MR (greedy) algorithms were assessed using $\ell \cdot \varUpsilon$ machines, utilizing each thread as an individual core. Thus, all algorithms were configured to use all available cores.
      \item[$\ddag$] We kept a variant of \dsbshort evaluated on 8 machines (with 4 cores each) as a baseline. All other algorithms use 2 machines (with 4 cores each).
    \end{tablenotes}
  \end{threeparttable}
\end{table}



\textbf{Environment.} Our experiments utilized diverse computational setups. Experiments 1, and 3 employed an eight-machine cluster, each with four CPU cores, totaling 32 cores. Experiment 2 was conducted on a larger 64-machine cluster, each featuring 32 CPU cores. Experiment 4 operated on a cluster of 32 single-core machines. Notably, in Experiment 1, our algorithms used $\ell = 8$ machines, while prior MapReduce (MR) algorithms employed $\ell = 32$, fully utilizing 32 cores. 
MPICH version 3.3a2 was installed on every machine, and we used the python library $mpi4py$ for implementing and parallelizing all algorithms with the Message Passing Interface (MPI). These algorithms were executed using the $mpirun$ command, with runtime tracked using $mpi4py.MPI.Wtime()$ at the algorithm's start and completion.

\textbf{Datasets.} 
Table \ref{table:datasets} presents dataset sizes for our experiments. In Experiment 1 and 3 (Fig.~\ref{fig:scaleC}), we assess algorithm performance on small datasets, varying from $n$ = 10,000 to 100,000. These sizes enable evaluation of computationally intensive algorithms, such as \ddshort, \dat, and \dagh. Experiment 2, 3 (Fig. \ref{fig:scaleB}), and 4 focus on larger datasets, ranging from $n$ = 50,000 to 5 million.



\textbf{Applications.} We evaluated 
image summarization (ImageSumm), influence maximization (InfMax), 
revenue maximization (RevMax) and maximum coverage (MaxCover). 
Details are provided in Appendix \ref{app:app}.

\revone
\textbf{Experiment Objectives.} The primary objective of our experiment set is to comprehensively evaluate the practicality of the distributed algorithms across varying cluster and dataset sizes. The specific objectives of each experiment are outlined below:

\begin{itemize}
\item Experiment 1\footnote[1]{Results published in Thirty-Seventh AAAI Conference on Artificial Intelligence, AAAI 2023}: Baseline experiment aimed at assessing the performance of the algorithms using small datasets and cluster setup.
\item Experiment 2: Assess the performance of the algorithms on large datasets and cluster setup.
\item Experiment 3$^*$: Investigate the influence of the number of nodes in a cluster on algorithm performance.
\item Experiment 4: Examine the impact of increasing cardinality constraints on the performance of \mg.
\end{itemize}
\color{black}

\subsection{Experiment 1 - Comparative Analysis on Small Datasets}\label{subsec:exp1}
The results of Experiment 1 (Fig. \ref{fig:Fig1}) show that all algorithms provide similar solution
values (with \dat being a little worse than the others).
However, there is a large difference in runtime,
with \dsb the fastest by orders of magnitude. 
The availability of only 4 threads per machine severely limits 
the parallelization of \dat, resulting in longer runtime;
access to $\log{}_{1+\epsi}(k)$ threads per machine 
should result in faster runtime than \dsb.

\begin{figure*}[h]
  \centering
  \setcounter{subfigure}{0}
        \subfigure[]{\includegraphics[height=0.175\textheight]{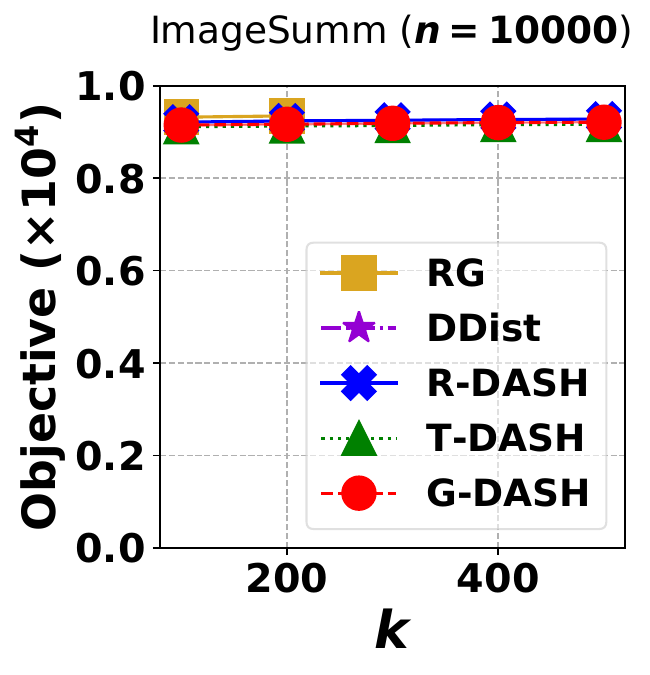}} 
        \subfigure[]{\includegraphics[height=0.175\textheight]{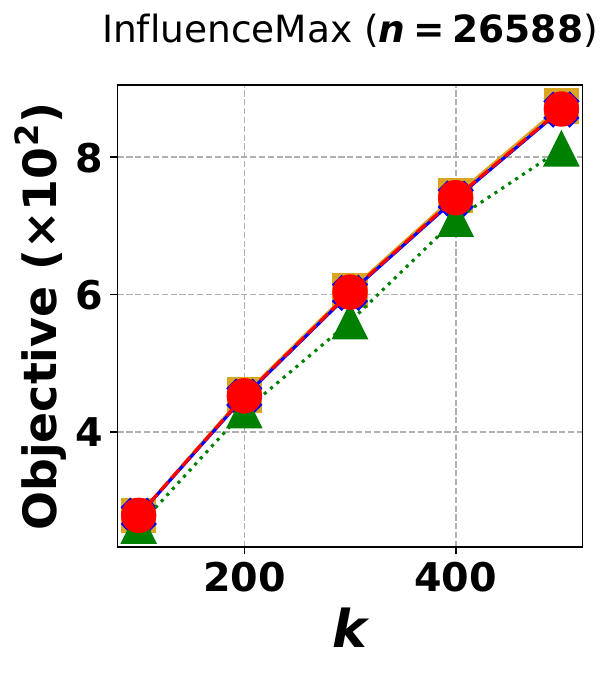}} 
        \subfigure[]{\includegraphics[height=0.175\textheight]{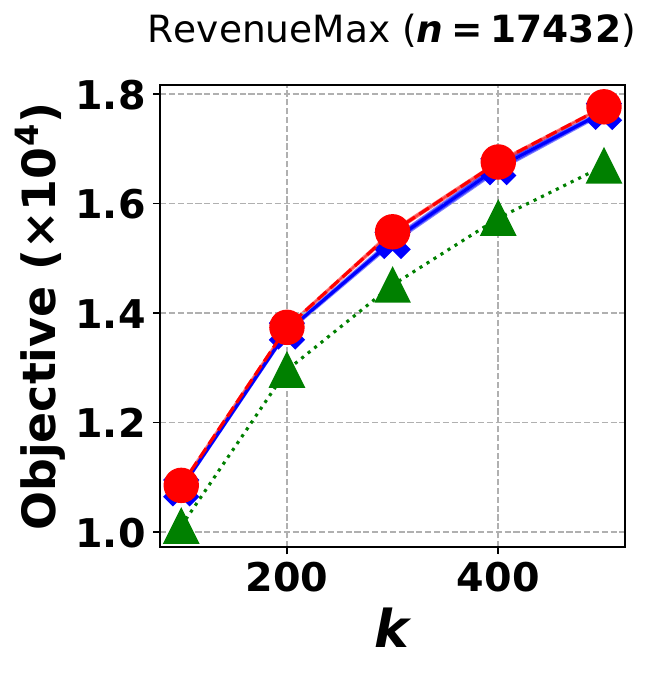}} 
        \subfigure[]{\includegraphics[height=0.175\textheight]{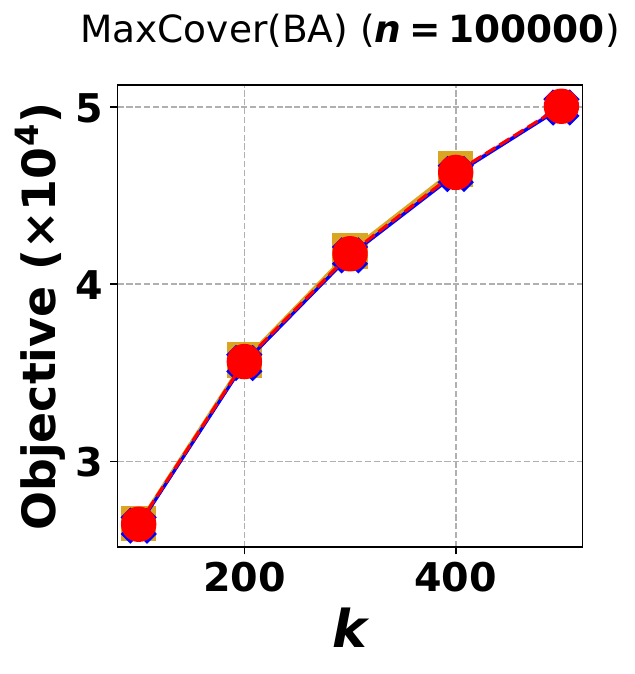}} 
        \subfigure[]{\includegraphics[height=0.175\textheight]{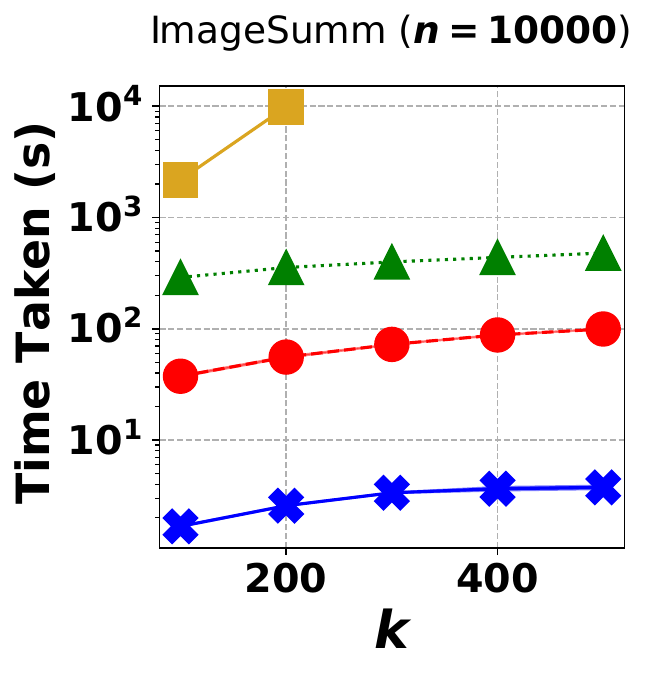}} 
        \subfigure[]{\includegraphics[height=0.175\textheight]{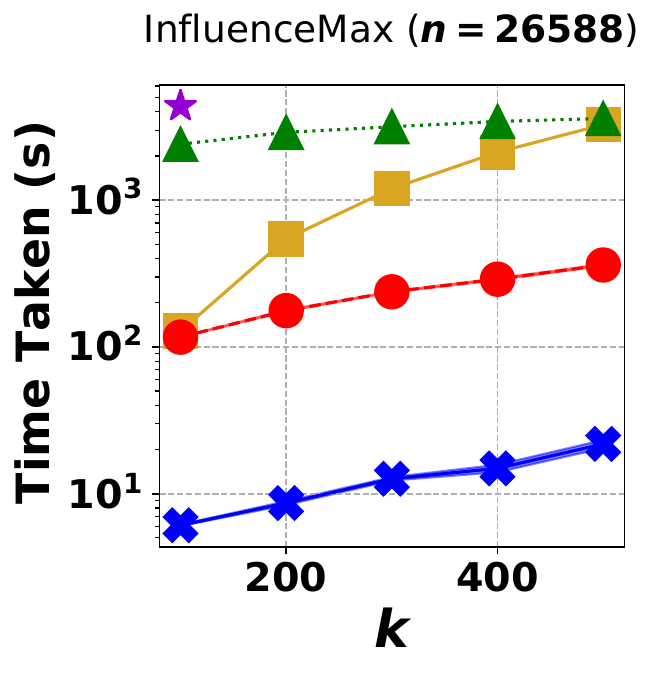}} 
       \subfigure[]{\includegraphics[height=0.175\textheight]{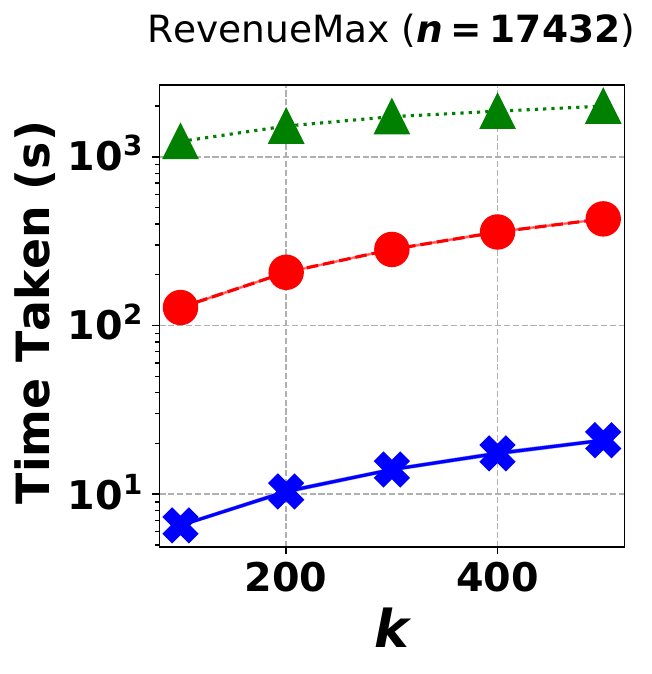}} 
        \subfigure[]{\includegraphics[ height=0.175\textheight]{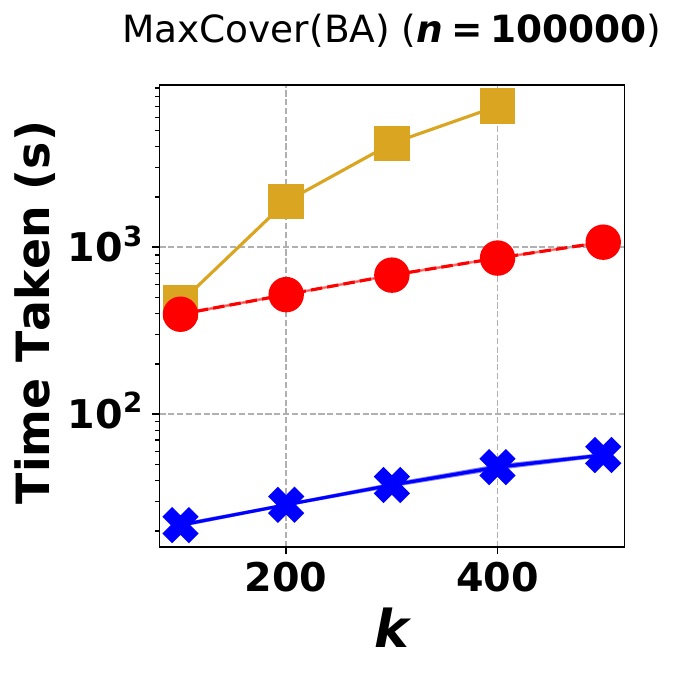}}
  \caption{Performance comparison of
distributed algorithms on 
ImageSumm, InfluenceMax, RevenueMax and MaxCover;  
\rgreedy (\rgshort) is run with \sg as the algorithm \alg \cite{barbosa2015power} to ensure 
the $\frac{1}{2}(1-1/e)$  ratio.
All \sg implementations used lazy greedy to improve the runtime.
$Timeout$ for each application: 6 hours per algorithm.} \label{fig:Fig1} 
\end{figure*}

\subsection{Experiment 2 - Performance Analysis of \dsbshort, \rgreedy and \dlsshort on a Large Cluster}\label{subsec:exp4}

\begin{figure*}[h]
  \centering
  \setcounter{subfigure}{0}
  \subfigure[]{\includegraphics[width=0.31\textwidth]{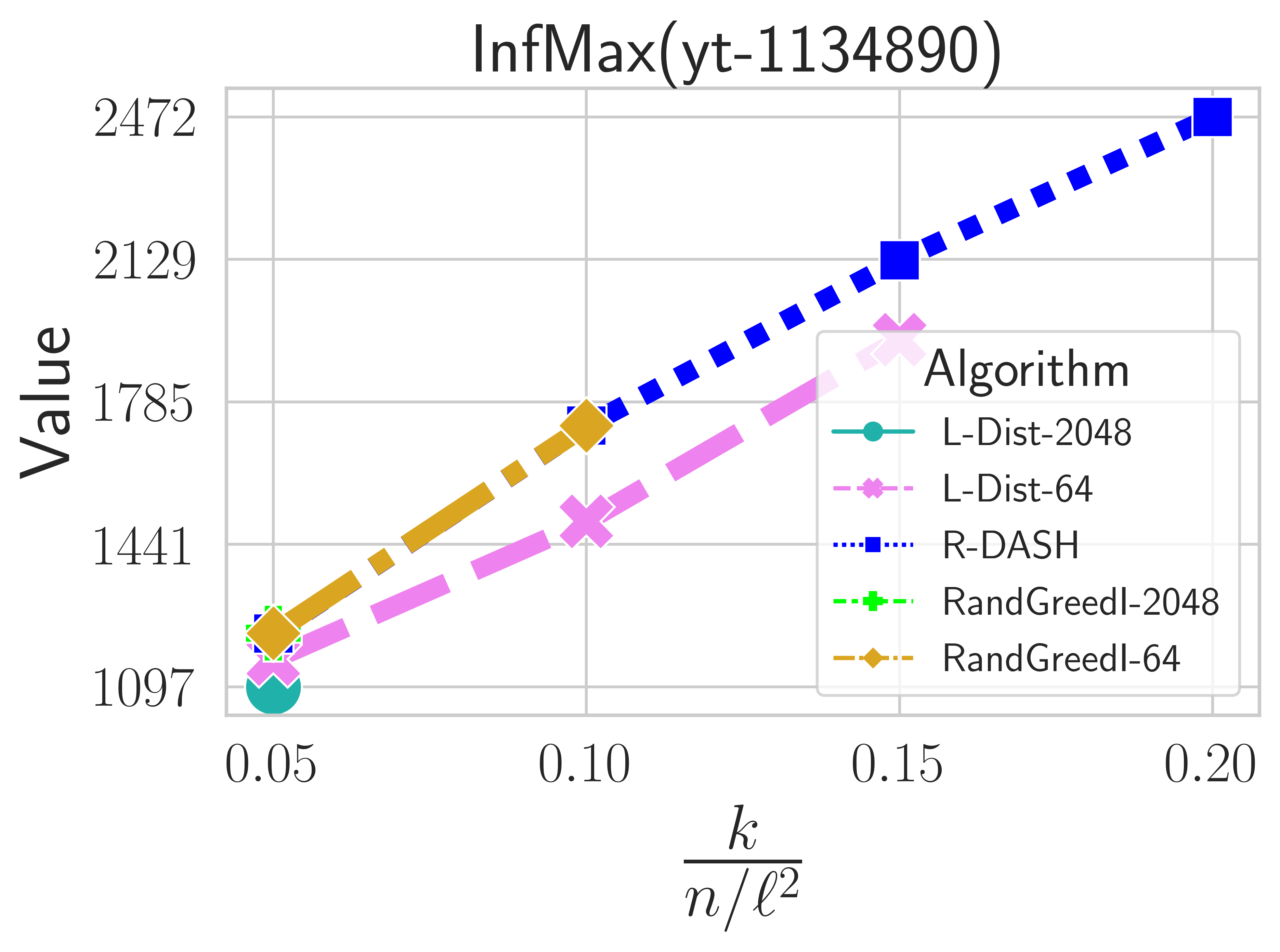} \label{fig:FigExp1-1J1}} 
  \subfigure[]{\includegraphics[width=0.31\textwidth]{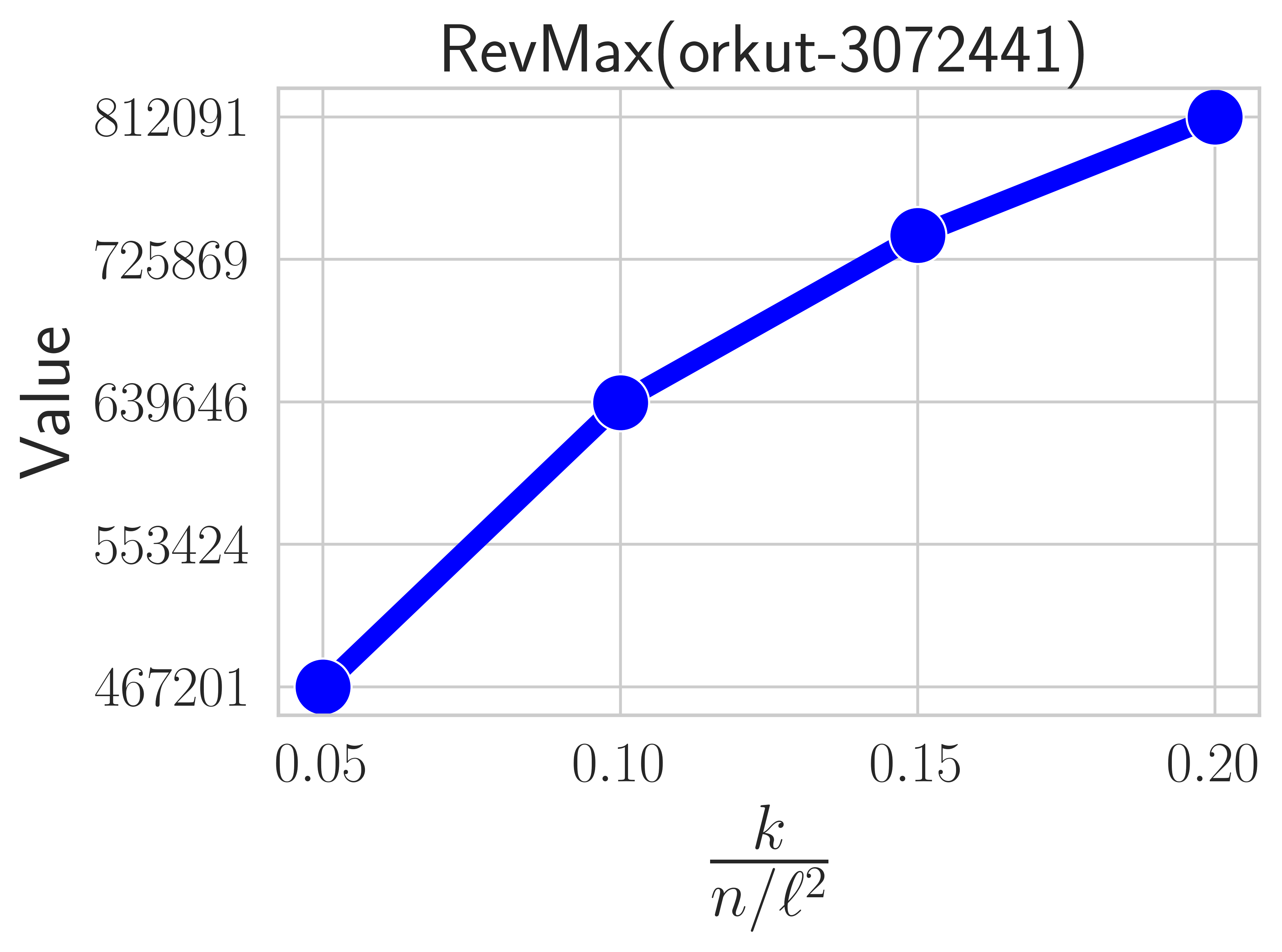} \label{fig:FigExp1-2J1}} 
  \subfigure[]{\includegraphics[width=0.31\textwidth]{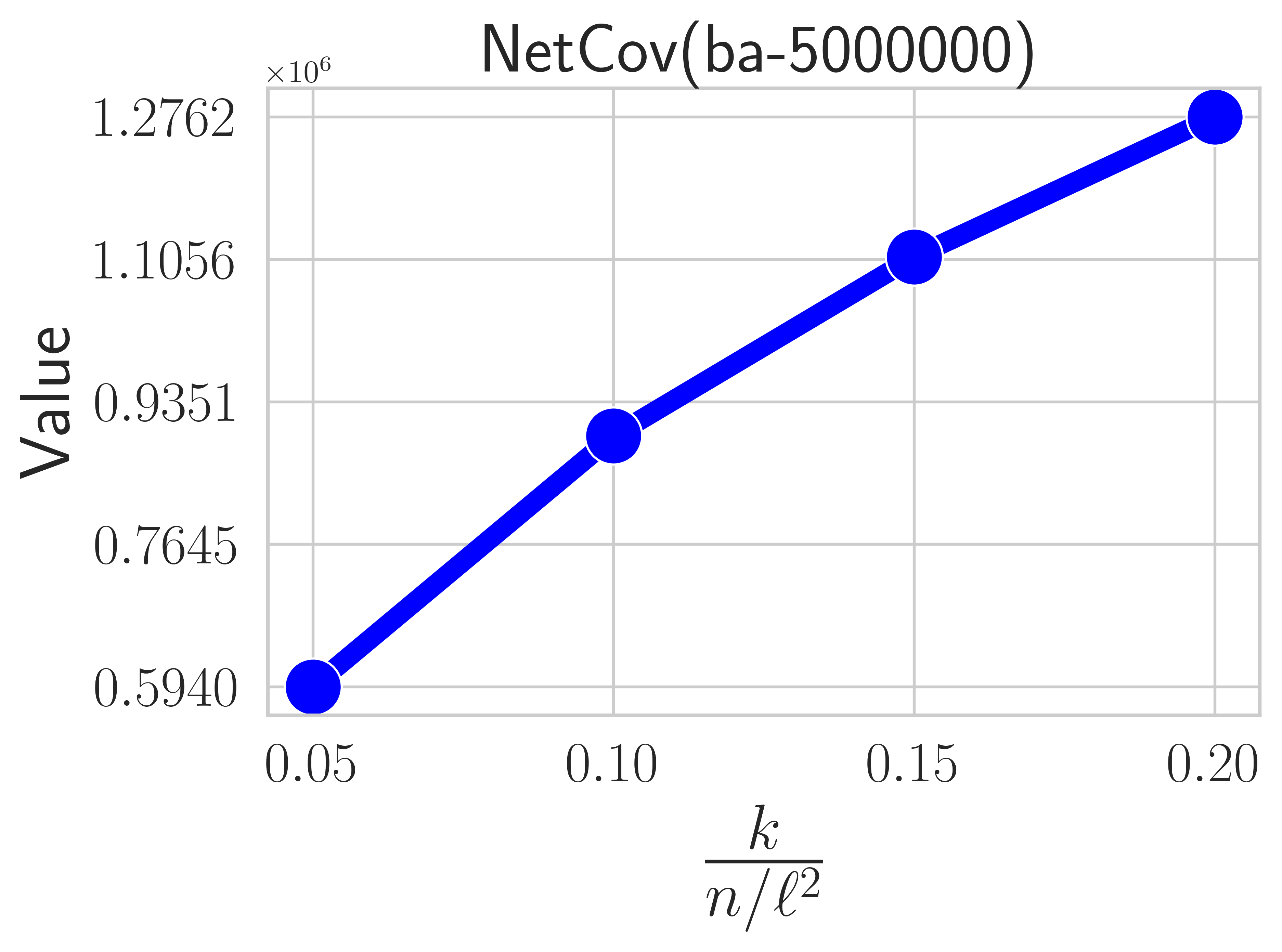} \label{fig:FigExp1-3J1}} 
\subfigure[]{\includegraphics[width=0.31\textwidth]{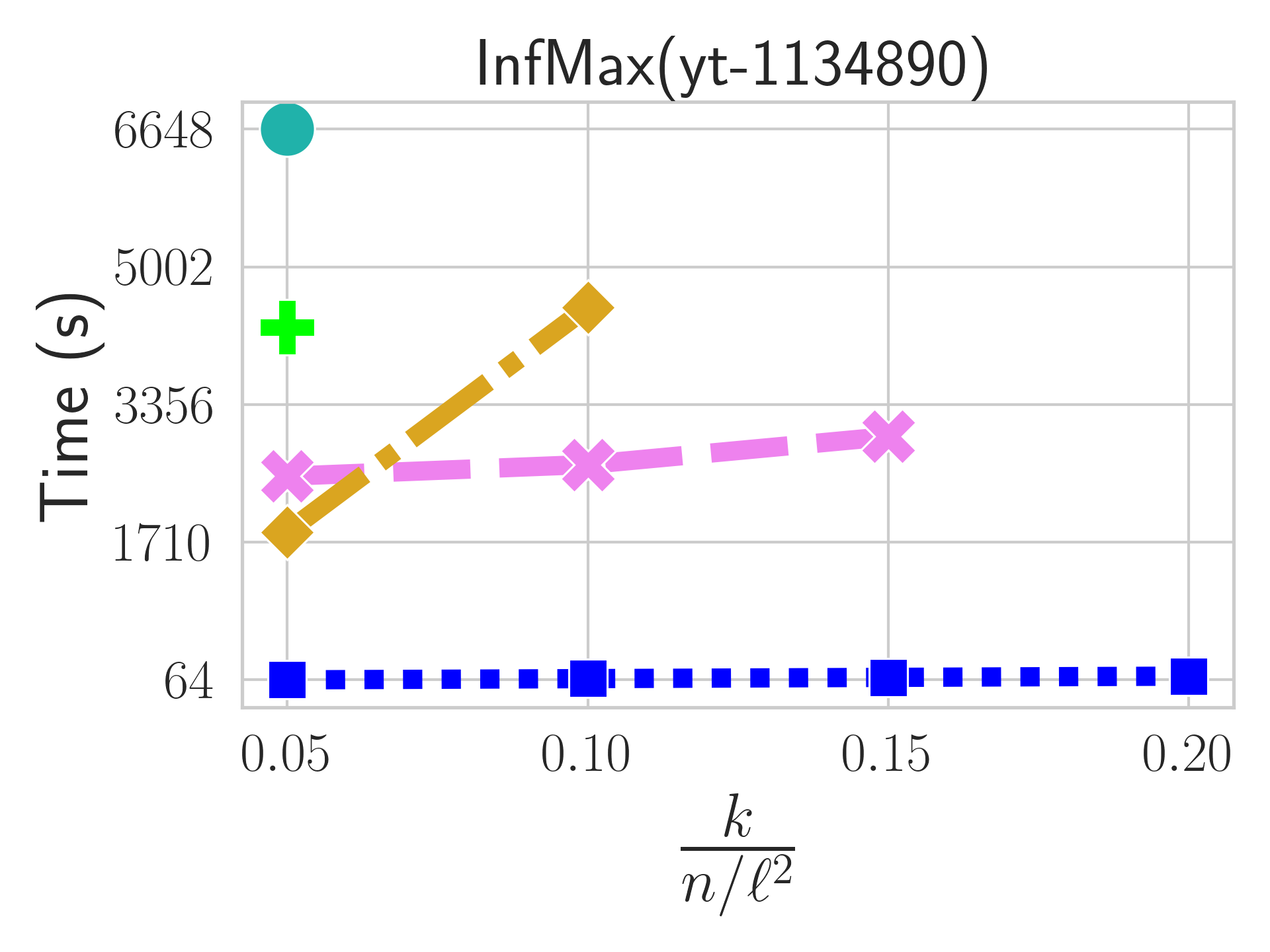} \label{fig:FigExp1-4J1}} 
\subfigure[]{\includegraphics[width=0.31\textwidth]{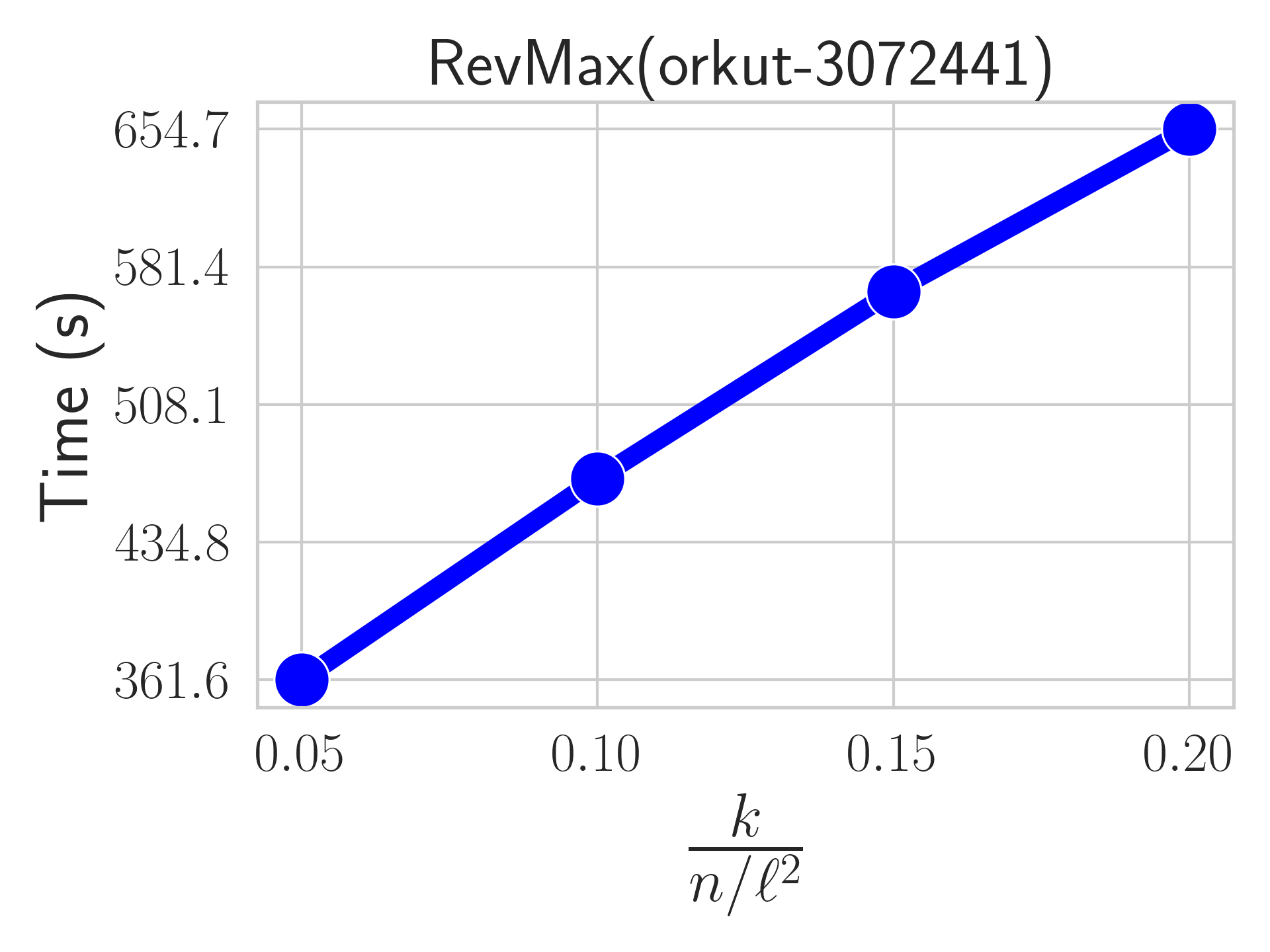} \label{fig:FigExp1-5J1}} 
\subfigure[]{\includegraphics[width=0.31\textwidth]{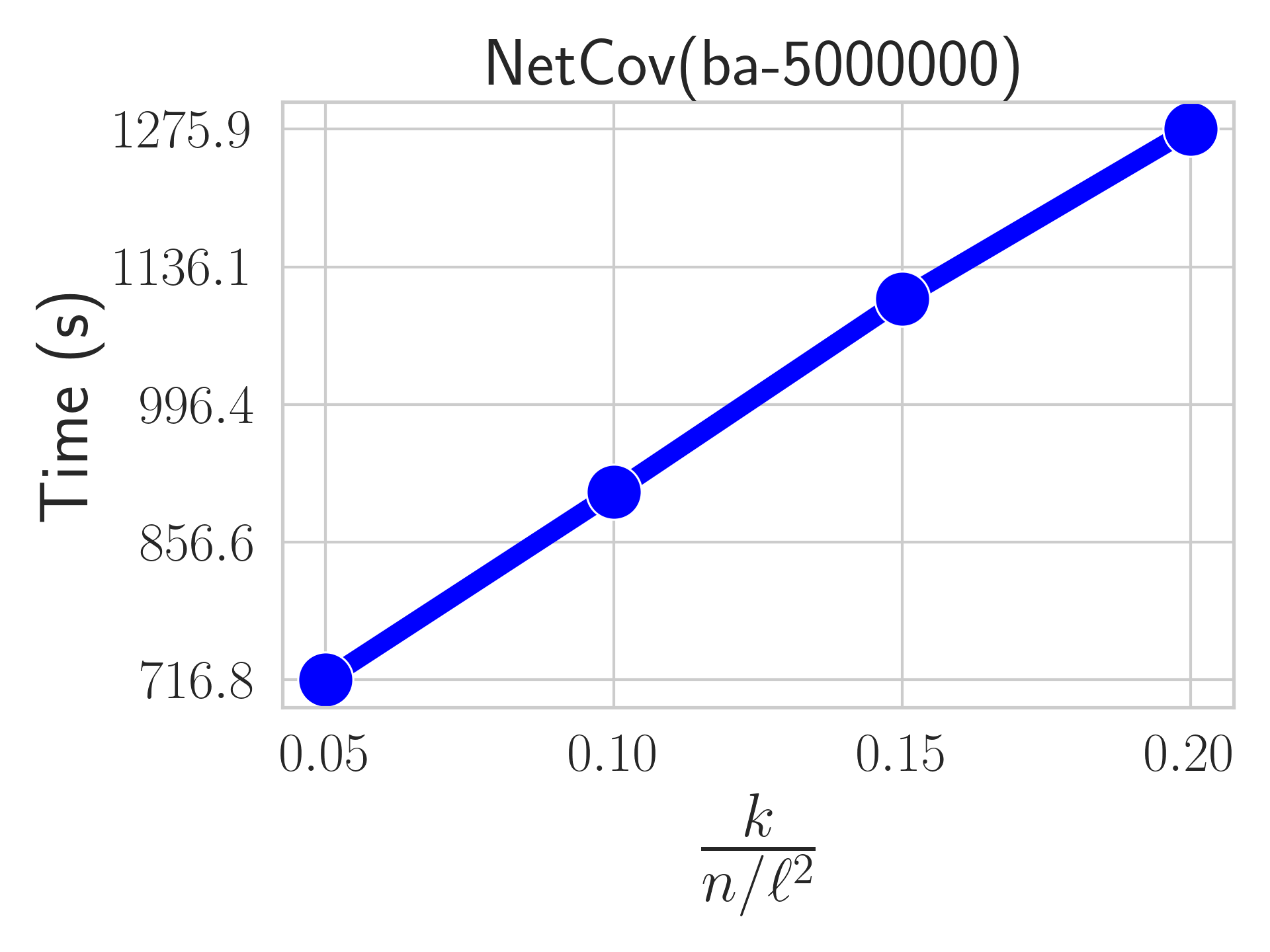} \label{fig:FigExp1-6J1}} 
  \caption{Empirical comparison of \dsbshort and \dlsshort. The plotted metrics are solution value (Fig. \ref{fig:FigExp1-1J1}-\ref{fig:FigExp1-3J1}) and runtime (Fig. \ref{fig:FigExp1-4J1}-\ref{fig:FigExp1-6J1}).} \label{fig:FigExp1J1}
\end{figure*}


This section presents a comparative analysis of \dsbshort, \dlsshort, and \rgreedy 
on a large 64-node cluster, each node equipped with 32 cores. 
We assess two versions of \dlsshort and \rgreedy, one with $\ell = 64$ 
and the other with $\ell = 2048$, with a predetermined time limit of 
3 hours for each application. The plotted results depict the instances 
completed by each algorithm within this time constraint.

In terms of solution quality, as depicted in Figure \ref{fig:FigExp1-1J1}, statistically
indistinguishable results, while \dlsshort exhibits an average drop of 9\% in solution quality 
across completed instances. Examining runtime (Figure \ref{fig:FigExp1-4J1}), \dsbshort consistently
outperforms \dlsshort and \rgreedy by orders of magnitude. As discussed in Experiment 3, for both
\dlsshort and \rgreedy, instances with fewer distributions ($\ell = 64$) display faster runtimes
compared to instances using a larger setup ($\ell = 2048$), attributed to $k$ values surpassing 
$\frac{n}{\ell^2}$ ($\simeq 0.27$ for InfMax). The subsequent experiment demonstrates how \mgshort 
effectively addresses this challenge. Comparing the fastest variant of each algorithm, we observe that 
for the smallest $k$ values, \rgreedy outperforms \dlsshort. However, as $k$ increases, \rgreedy's 
runtime linearly grows, while \dlsshort's runtime remains unaffected due to its linear time 
complexity, enabling \dlsshort to outperform \rgreedy for larger constraints.


\subsection{Experiment 3 - Scalability Assessment} \label{subsec:exp3}
Figure \ref{fig:scaleB} illustrates a linear speedup for \dsb 
as the number of machines $\ell$ increases. Figure \ref{fig:scaleC} 
highlights an intriguing observation that, despite having sufficient 
available memory, increasing $\ell$ can result in inferior performance 
when $k > \frac{n}{\ell^2}$. Specifically, as depicted in Figure \ref{fig:scaleC}, 
we initially witness the expected faster execution of $\rgreedy$ with 
$\ell = 32$ compared to $\rgreedy$ with $\ell = 8$. However, once $k > \frac{n}{32^2}$, 
the relative performance of $\rgreedy$ with $\ell = 32$ rapidly deteriorates. 
This decline can be attributed to the degradation of $\rgreedy$'s optimal performance 
beyond $k = \frac{n}{\ell^2}$.

When running $\rgreedy$ with a single thread on each machine, 
the total running time on $\ell$ machines can be computed based on 
two components. First, the running time for one machine in the first MR round, 
which is proportional to $(n/\ell-(k-1)/2)k$. Second, the running time for the 
primary machine in the second MR round (post-processing step), which is proportional 
to $(k\ell-(k-1)/2)k$. Consequently, the total running time is proportional to $nk/\ell + \ell k^2-k(k-1)$. 
Optimal performance is achieved when $\ell=\sqrt{n/k}$, which justifies the preference 
for parallelization within a machine to maintain a lower $\ell$ rather than distributing 
the data across separate processors.

Furthermore, in Experiment 5 (Section~\ref{sec:Exp5}), we demonstrate that utilizing 
\mgshort enables MR algorithms to produce solutions much faster with no compromise 
in solution value, particularly when solving for $k > \frac{n}{\ell^2}$. These results 
provide further support for the advantage of incorporating \mgshort in achieving 
efficient and effective parallelization in MR algorithms.

\begin{figure}[h]
  \centering
  \subfigure[]{\label{fig:scaleB} \includegraphics[height=0.19\textheight]{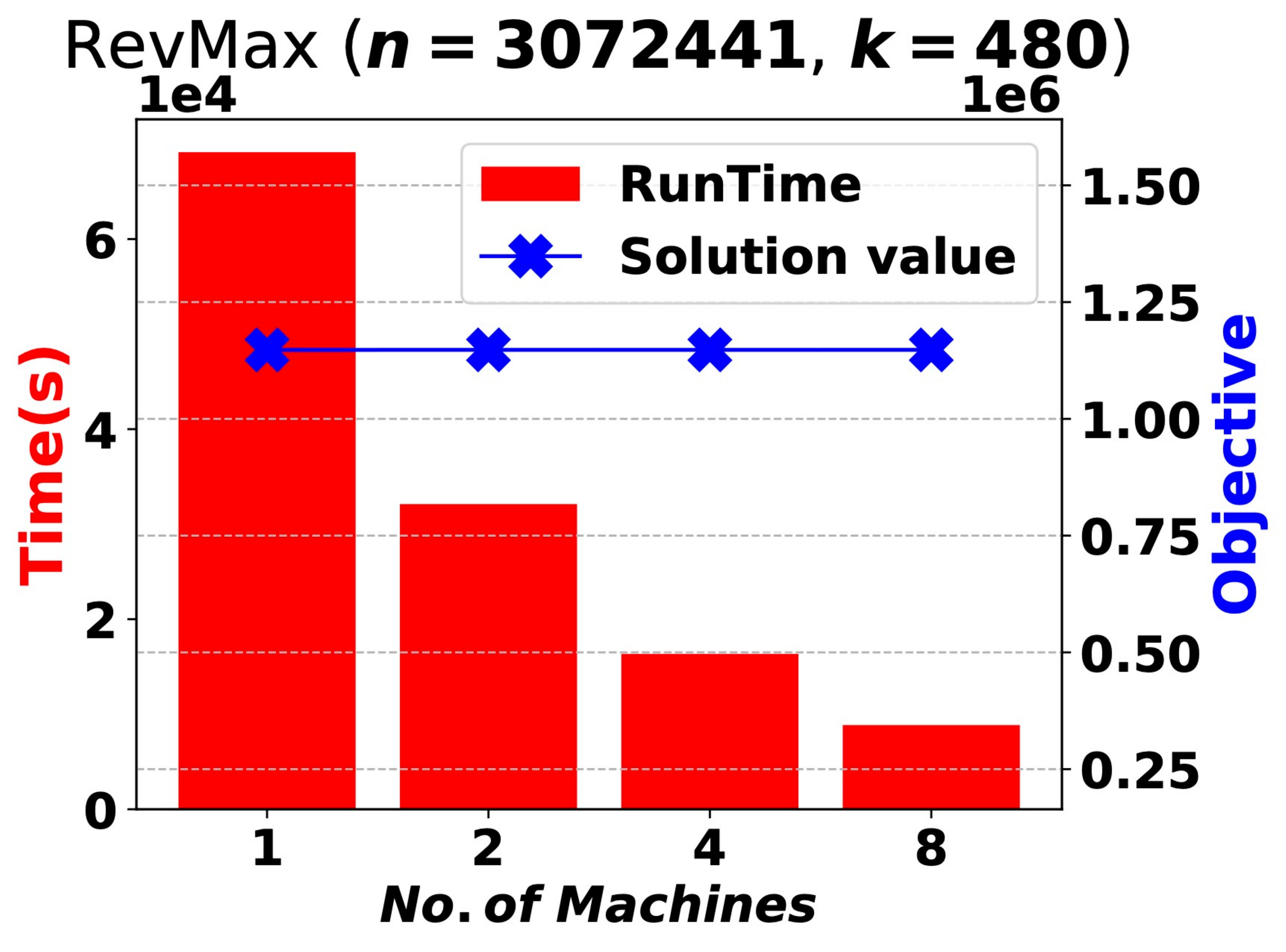}} \hspace{0.35em}
  \subfigure[]{\label{fig:scaleC} \includegraphics[height=0.19\textheight]{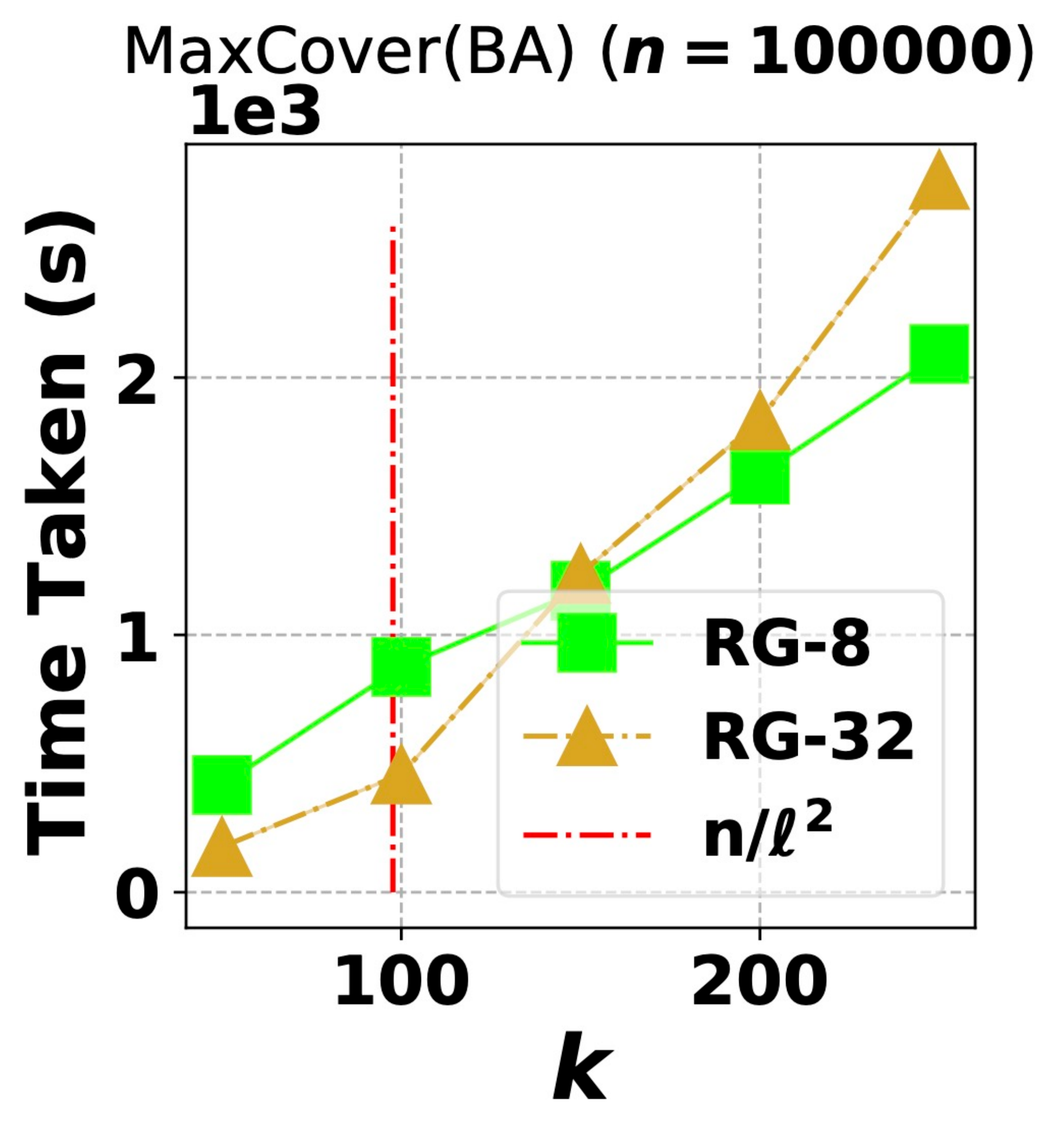}} 
  \caption{ \textbf{(a):} Scalability of \dsb vs. $\ell$ \textbf{(b):} \rgreedy with $\ell=8$ Vs. $\ell=32$. } \label{fig:Fig4m}
\end{figure}

\subsection{Experiment 4 - Performance Analysis of MED}\label{sec:Exp5}

\begin{figure*}[h]
  \centering
  \setcounter{subfigure}{0}
  \subfigure[]{\includegraphics[width=0.31\textwidth]{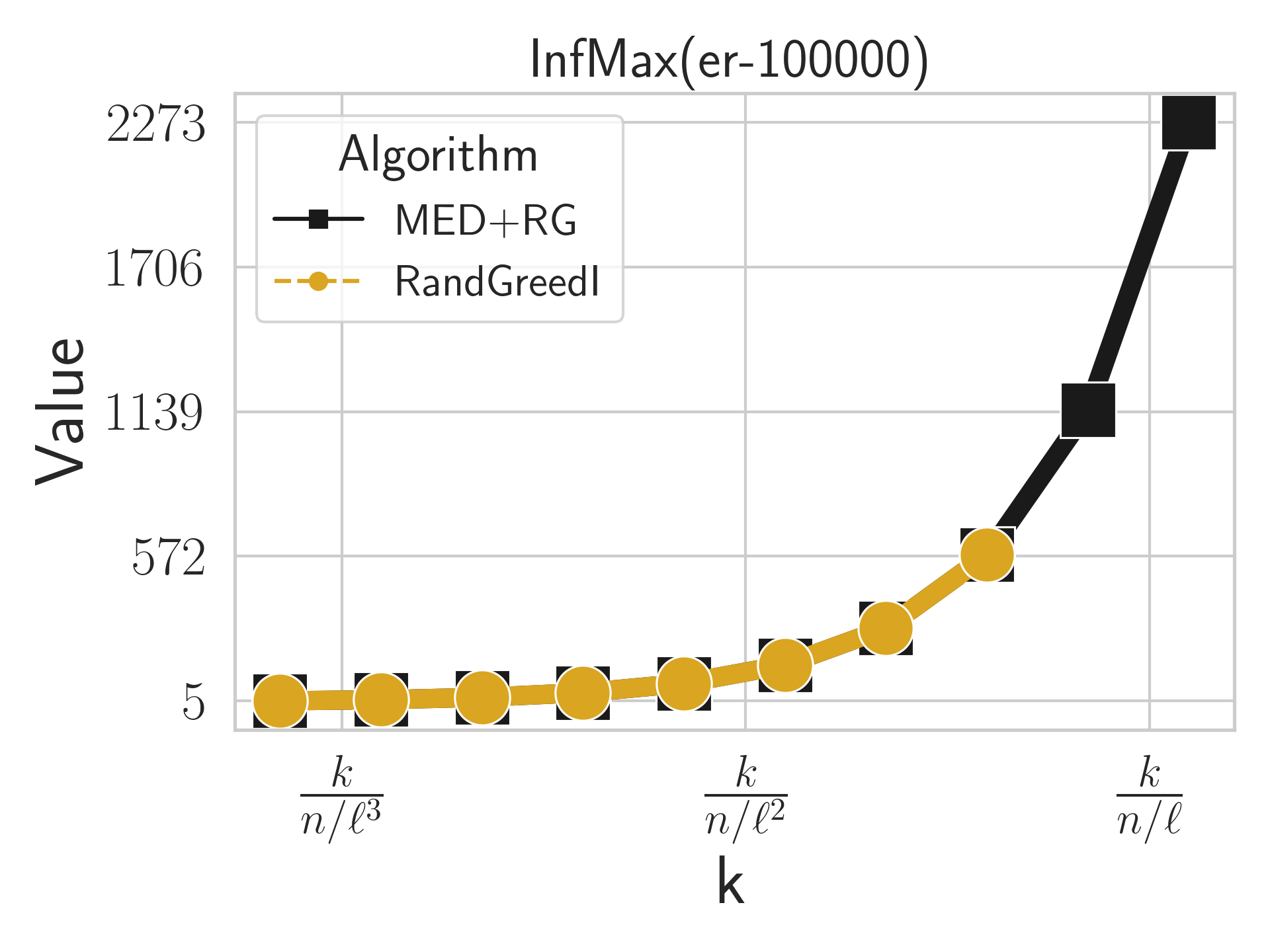} \label{fig:FigExp2-1J1}} 
  \subfigure[]{\includegraphics[width=0.31\textwidth]{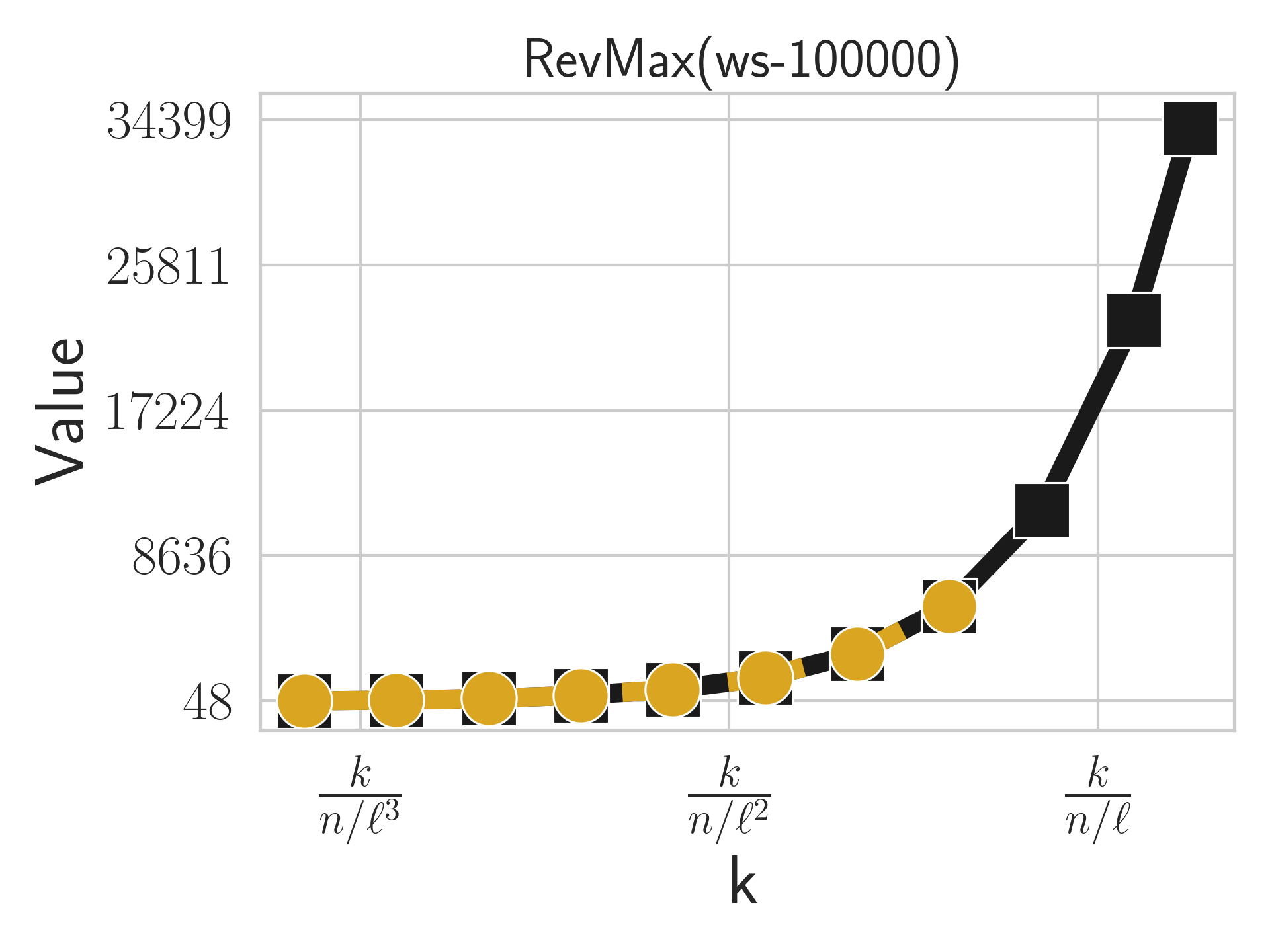} \label{fig:FigExp2-2J1}} 
  \subfigure[]{\includegraphics[width=0.31\textwidth]{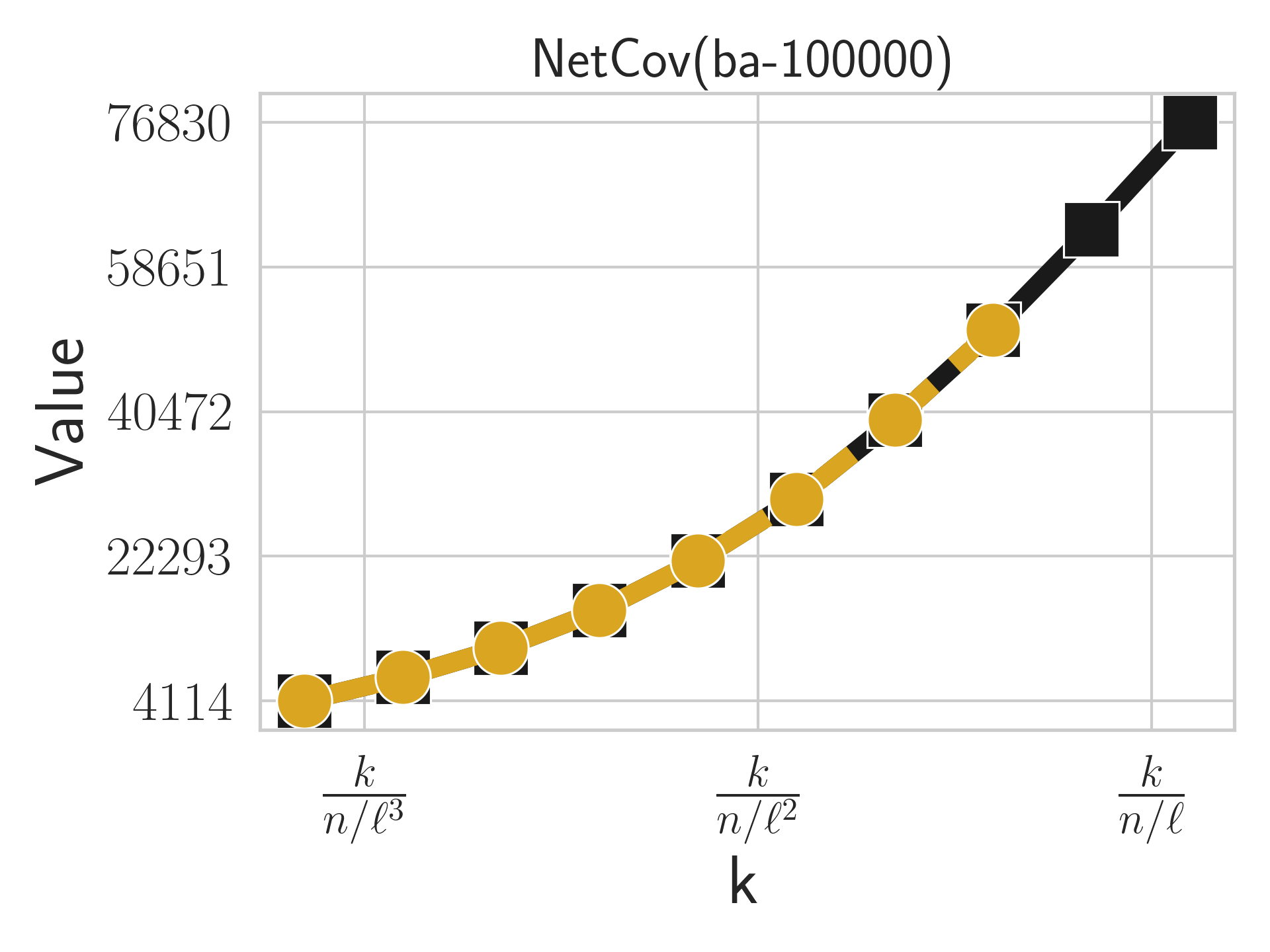} \label{fig:FigExp2-3J1}} 
\subfigure[]{\includegraphics[width=0.31\textwidth]{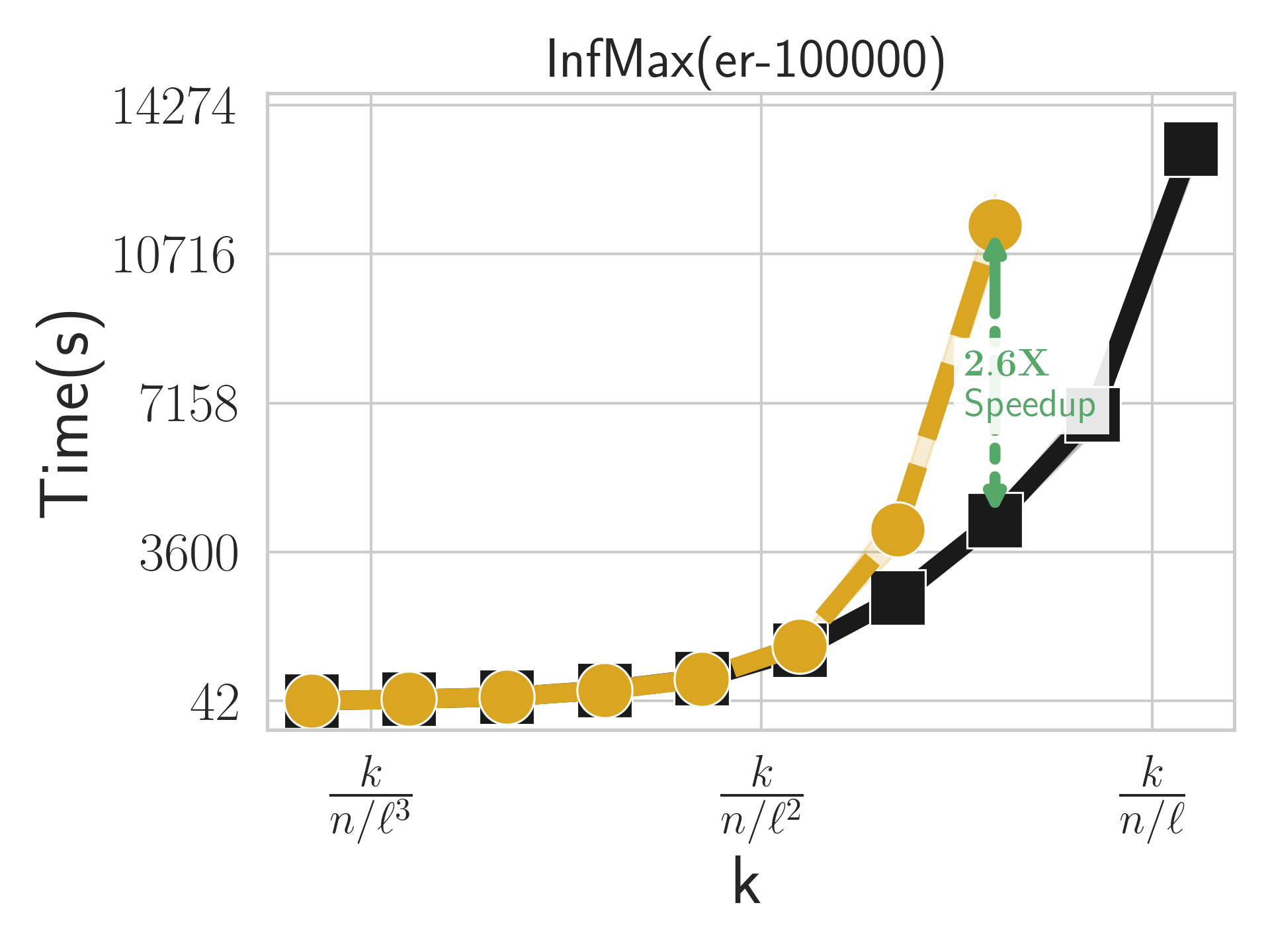} \label{fig:FigExp2-4J1}} 
\subfigure[]{\includegraphics[width=0.31\textwidth]{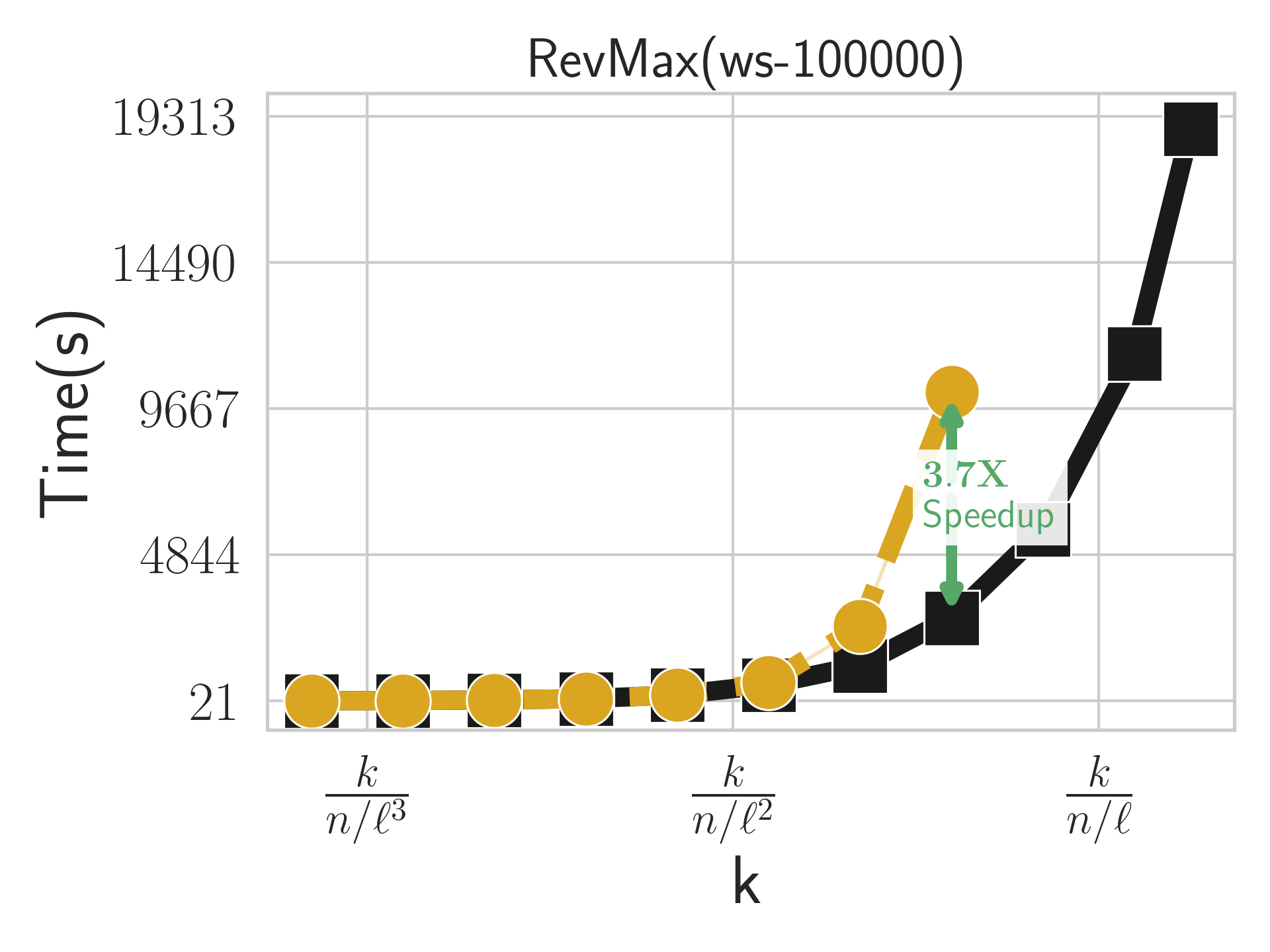} \label{fig:FigExp2-5J1}} 
\subfigure[]{\includegraphics[width=0.31\textwidth]{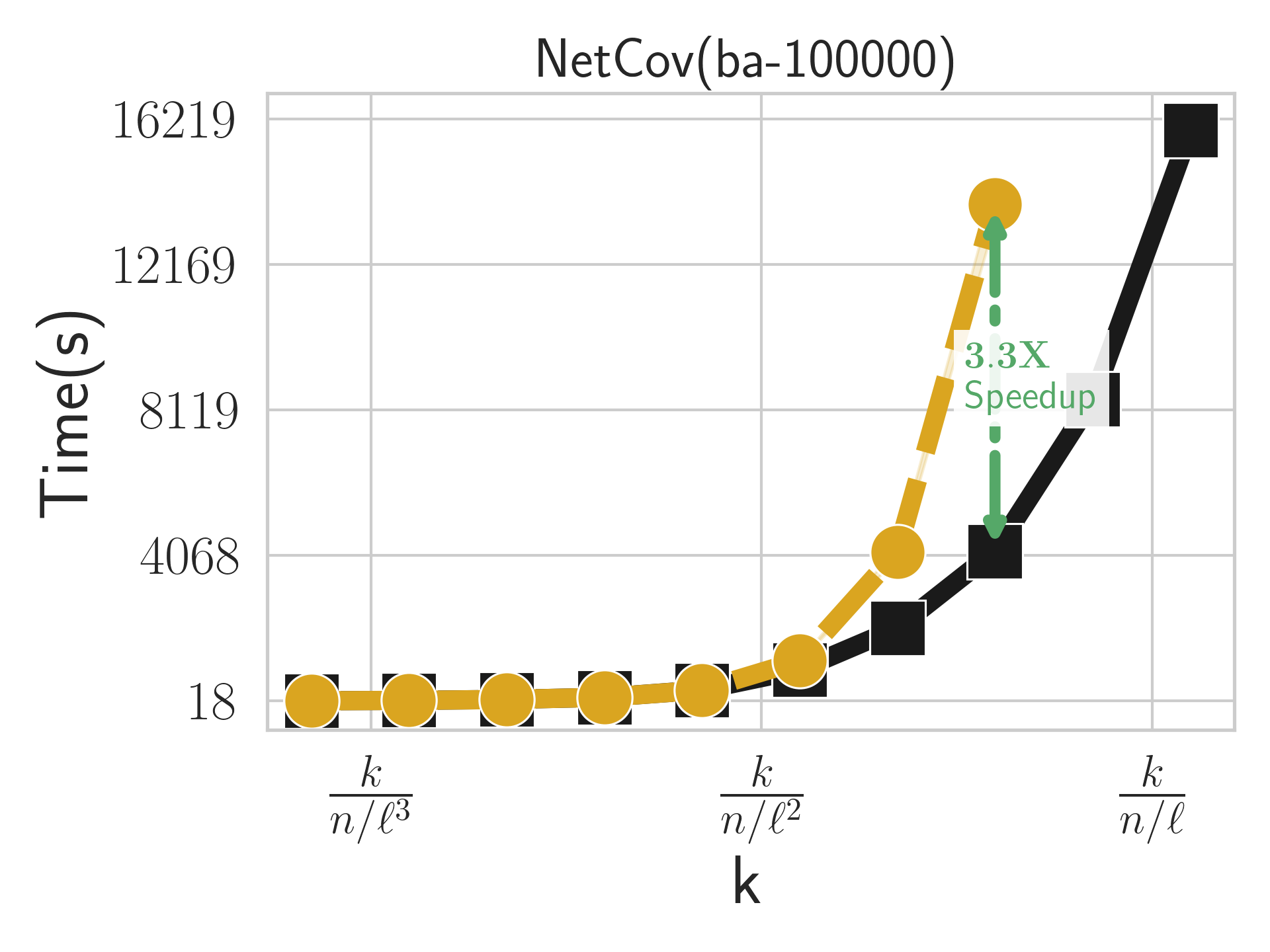} \label{fig:FigExp2-6J1}} 
  \caption{Empirical comparison of \mgrg to \rgreedy. The plotted metrics are solution value (Fig. \ref{fig:FigExp2-1J1}-\ref{fig:FigExp2-3J1}) and runtime (Fig. \ref{fig:FigExp2-4J1}-\ref{fig:FigExp2-6J1}).} \label{fig:FigExp2J1}
\end{figure*}

This section presents experimental results comparing the performance of 
\mgrg and the vanilla \rgreedy algorithm. In terms of solution value, 
\mgrg consistently provides nearly identical solutions to the vanilla 
\rgreedy algorithm across all instances for the three applications studied. 
Regarding runtime, the following observations are made: Initially, both 
algorithms exhibit similar execution times up to a threshold of $k=\frac{n}{\ell^2}$. 
However, beyond this threshold, the performance gap between \mgrg and \rgreedy widens 
linearly. \mgrg achieves notable average speedup factors of 1.8, 2.2, and 2.3 over 
\rgreedy for the respective applications. Moreover, beyond $k=\frac{n}{\ell^2}$, 
\mgrg outperforms \rgreedy in terms of completing instances within a 12-hour timeout, 
completing 77\% more instances of $k$ across all three applications. These findings 
highlight the promising performance of \mgrg, demonstrating comparable solution values 
while significantly improving runtime efficiency compared to the vanilla \rgreedy algorithm.

The empirical findings from this experiment provide insights into the capabilities of the 
\mgalg framework. Our results indicate that \mgalg achieves solution quality that is 
practically indistinguishable from the vanilla \alg algorithm, even for significantly 
larger values of $k$. Additionally, even with sufficient available memory to run \alg beyond 
$k=\frac{n}{\ell^2}$, \mgalg demonstrates noteworthy computational efficiency, 
surpassing that of the \alg algorithm. This combination of comparable solution quality and 
improved computational efficiency positions \mgalg as a highly promising framework. These 
findings have significant implications and underscore the potential of \mgalg to address 
complex problems in a distributed setting with large-scale values of $k$.


\section{Discussion and Conclusion}
{\revone
Prior to this work, no MR algorithms for SMCC could parallelize within a machine; these algorithms require many sequential
queries. Moreover, increasing the number of machines to the number of threads available in the cluster
may actually harm the performance, as we showed empirically; intuitively, this is because the size of the
set on the primary machine scales linearly with the number of machines $\ell$.
In this paper, we have addressed this limitation by introducing a suite of algorithms
that are both parallelizable and distributed.
Specifically, we have presented R-DASH, T-DASH, and G-DASH, which are the first MR algorithms with sublinear adaptive complexity (highly parallelizable). Moreover, our algorithms make have nearly linear query complexity over the entire
computation.
We also provide the first distributed algorithm with $O(n)$ total query complexity, improving on
the $O(n \text{polylog}(n))$ of our other algorithms and the algorithm of \citet{liu2018submodular}.

When \rgreedy was introduced by \citet{mirzasoleiman2013distributed}, the empirical performance of the
algorithm was emphasized, with theoretical guarantees unproven until the work of \citet{barbosa2015power}.
Since that time, \rgreedy has remained the most practical algorithm for distributed SMCC.
Our R-DASH algorithm may be regarded as a version of \rgreedy that is 1) parallelized; 2)
nearly linear time in total (\rgreedy is quadratic); and 3) empirically orders of magnitude
faster in parallel wall time in our evaluation.

R-DASH achieves approximation ratio of $(1 - 1/e)/2$ in two MR rounds; the first round merely
being to distribute the data. We provide G-DASH to close the gap to $(1-1/e)$ ratio in a constant number of rounds. However,
MR rounds are expensive (as shown by our Experiment 1). The current best ratio achieved in two rounds
is the $0.545$-approximation algorithm of \citet{mirrokni2015randomized}. Therefore,
a natural question for future research is what is the best ratio achievable in two MR rounds.
Moreover, non-monotone or partially monotone objective functions and more sophisticated constraint systems
are directions for future research. 

\color{black}
\section*{Acknowledgements}
Yixin Chen and Tonmoy Dey contributed equally to this work.
The work of Tonmoy Dey was partially supported by Florida State University;
the work of Yixin Chen and Alan Kuhnle was partially 
 supported by Texas A \& M University.
 The authors have received no third-party funding in direct support of this work. 
 The authors have no additional revenues from other sources related to this work.

\appendix
\section{Lov\'{a}sz Extension of Submodular Function.}
\label{sec:lov}
Given submodular function $f$,
the Lov\'{a}sz extension $F$ of $f$
is defined as follows: For $\mathbf{z} \in [0,1]^{|\univ|}$,
\[ F( \mathbf{z} = (z_i)_{i \in \univ} ) = \E_{\lambda  \sim \mathcal{U}[0,1]}[ f( \{ i : z_i \ge \lambda \} ) ]. \]
The Lov\'{a}sz extension satisfies the following properties:
(1) $F$ is convex; (2) $F( c \mathbf z ) \ge c F( \mathbf z )$
for any $c \in (0,1)$. Moreover, we will require the following
simple lemma:
\begin{lemma}[\citet{barbosa2015power}] \label{lemma:lovasz} Let $S$ be a random set,
  and suppose $\E[ \mathbf{1}_S ] = c \cdot \mathbf{z}$, for
  $c \in (0,1)$. Then $\E[ f(S) ] \ge c \cdot F( \mathbf{z} )$.
\end{lemma}

\section{Probability Lemma and Concentration Bounds}
\begin{lemma} \label{lemma:indep}(\citet{chen2021best})
  Suppose there is a sequence of $n$ Bernoulli trials:
  $X_1, X_2, \ldots, X_n,${}
  where the success probability of $X_i$
  depends on the results of
  the preceding trials $X_1, \ldots, X_{i-1}$.
  Suppose it holds that $$\prob{X_i = 1 | X_1 = x_1, X_2 = x_2, \ldots, X_{i-1} = x_{i-1} } \ge \eta,$$ where $\eta > 0$ is a constant and $x_1,\ldots,x_{i-1}$ are arbitrary.

  Then, if $Y_1,\ldots, Y_n$ are independent Bernoulli trials, each with probability $\eta$ of
  success, then $$\prob {\sum_{i = 1}^n X_i \le b } \le \prob{\sum_{i = 1}^n Y_i \le b }, $$
  where $b$ is an arbitrary integer.

  Moreover, let $A$ be the first occurrence of success in sequence $X_i$.
  Then, $$\ex{A} \le 1/\eta.$$
\end{lemma}

\begin{lemma}[Chernoff bounds
\cite{mitzenmacher2017probability}]\label{lemma:chernoff}
    Suppose $X_1$, ... , $X_n$ are independent binary random variables such that 
    $\prob{X_i = 1} = p_i$. Let $\mu = \sum_{i=1}^n p_i$, and 
    $X = \sum_{i=1}^n X_i$. Then for any $\delta \geq 0$, we have
    \begin{align}
        \prob{X \ge (1+\delta)\mu} \le e^{-\frac{\delta^2 \mu}{2+\delta}}.
    \end{align}
    Moreover, for any $0 \leq \delta \leq 1$, we have
    \begin{align}
        \prob{X \le (1-\delta)\mu} \le e^{-\frac{\delta^2 \mu}{2}}.
    \end{align}
\end{lemma}

\section{Omitted Proof of \lat} \label{apx:lat}
\lemmalatinc*
\begin{proof}
   After filtration on Line~\ref{line:threFilterV},
   it holds that, for any $x \in V_j$,
   $\marge{x}{S_{j-1}}\ge \tau$ and $\marge{x}{V_j \cup S_{j-1}} = 0$.
   Therefore,
   \begin{align*}
       &|A_0| = |\{x\in V_j: \marge{x}{S_{j-1}}< \tau\}|=0,\\
       &|A_{|V_j|}| = |\{x\in V_j: \marge{x}{V_j \cup S_{j-1}}< \tau\}|=|V_j|.
   \end{align*}
   For any $x\in A_i$, it holds that $\marge{x}{T_i \cup S_{j-1}} < \tau$.
   Due to submodularity, $\marge{x}{T_{i+1} \cup S_{j-1}} \le \marge{x}{T_i \cup S_{j-1}}  < \tau$.
   Therefore, $A_i$ is subset of $A_{i+1}$, which indicates that $|A_i|\le |A_{i+1}|$.
\end{proof}
\begin{proof}[Proof of Success Probability]
    The algorithm successfully terminates if, at some point, 
    $|V_j| = 0$ or $|S_j| = k$.
    To analyze the success probability,
    we consider a variant of the algorithm in which
    it does not terminate once $|V_j| = 0$ or $|S_j| = k$.
    In this case, the algorithm keeps running with $s = 0$
    and selecting empty set after inner for loop
    in the following iterations.
    Thus, with probability $1$, 
    either $|V_{j+1}| = 0 \le (1-\beta\epsi)|V_j| = 0$,
    or $|S_j|=k$.
    Lemma~\ref{lemma:latprob} holds for all $M+1$ iterations of the outer for loop.

    If the algorithm fails, there should be no more than $m=\lceil \log_{1-\beta\epsi}(1/n) \rceil$ successful iterations where $\lamb_j^* \ge \min\{s,t\}$.
    Otherwise, by Lemma~\ref{lemma:latprob}, there exists an iteration $j$ such that $|S_j| = k$, or $|V_{M+1}| < n\left(1-\beta\epsi\right)^m \le 1$.
    Let $X$ be the number of successful iterations.
    Then, $X$ is the sum of $M+1$ dependent Bernoulli random variables,
    where each variable has a success probability of more than $1/2$.
    Let $Y$ be the sum of $M+1$ independent Bernoulli random variables
    with success probability $1/2$.
    By probability lemmata, the failure probability of Alg.~\ref{alg:threshold} is calculated as follows,
    \begin{align*}
        \prob{\text{Alg.~\ref{alg:threshold} fails}}&\le \prob{X\le m}\\
        &\le \prob{Y \le m} \tag{Lemma~\ref{lemma:indep}}\\
        &\le \prob{Y \le \log\left(\frac n \delta\right)/(\beta\epsi)} \tag{$\log(x)\ge 1-\frac 1 x, \forall x > 0$}\\
        &\le e^{-\frac 1 2 \left(1-\frac{1}{2(1+\beta\epsi)}\right)^2 \cdot 2(1+1/\beta\epsi)\log\left(\frac n \delta\right)}\tag{Lemma~\ref{lemma:chernoff}}\\
        &= \left(\frac \delta n\right)^{\frac{(2\beta\epsi+1)^2}{4\beta\epsi(\beta\epsi+1)}}\le \frac \delta n\qedhere
    \end{align*}
\end{proof}

\begin{proof}[Proof of Adaptivity and Query Complexity]
In Alg.~\ref{alg:threshold}, oracle queries incurred on Line~\ref{line:threFilterV} and~\ref{line:threIf2},
and can be done in parallel.
Thus, there are constant number of adaptive rounds within
one iteration of the outer for loop.
The adaptivity is $\oh{M}=\oh{\log(n/\delta)/\epsi^3}$.

Consider an iteration $j$, there are no more than $|V_{j-1}|+1$ queries
on Line~\ref{line:threFilterV} and $\log_{1+\epsi}(|V_j|)+2$ queries on Line~\ref{line:threIf2}.
Let $Y_i$ be the $i-$th successful iteration.
Then, for any $Y_{i-1} < j \le Y_{i}$, 
it holds that $|V_j| \le n(1-\beta\epsi)^{i-1}$.
By Lemma~\ref{lemma:indep}, for any $i\ge 1$,
$\ex{Y_i - Y_{i-1}}\le 2$.
Then, the expected number of oracle queries is as follows,
\begin{align*}
    &\ex{\#\text{queries}} \le \sum_{j=1}^{M+1}\ex{|V_{j-1}|+\log_{1+\epsi}(|V_j|)+3}\\
    &\le n+\sum_{i:n(1-\beta\epsi)^{i-1}\ge 1}\ex{Y_i - Y_{i-1}} \left(n(1-\beta\epsi)^{i-1} + \log_{1+\epsi}\left(n(1-\beta\epsi)^{i-1}\right)\right)+ 3(M+1)\\
    &\le n+2n\sum_{i\ge 1}(1-\beta\epsi)^{i-1}+2\log_{1+\epsi}(n)\cdot \log_{1-\beta\epsi}(1/n)+3\log(n/\delta)/\epsi^3\\
    &\le \left(1+\frac{2}{\beta\epsi}\right)n+2\log_{1+\epsi}(n)\cdot \log_{1-\beta\epsi}(1/n)+3\log(n/\delta)/\epsi^3\\
    &\le O(n/\epsi^3)\qedhere
\end{align*}
\end{proof}

\revone
\begin{proof}[Proof of Property (3)]
    For each iteration $j$,
    consider two cases of $\lamb^*_j$ \revtwo
    by the choice of $\lamb^*_j$ on Line~\ref{line:LAT-index2}. 
    Firstly, if $\lamb^*_j \le \left \lceil \frac{1}{\epsi} \right \rceil $,
    clearly, \[\marge{T_{\lamb^*_j}}{S_{j-1}} \ge (1-\epsi)\tau \lamb^*_j.\]
    
    Secondly, if $\lamb^*_j > \left \lceil \frac{1}{\epsi} \right \rceil $,
    let $\lamb_j =\max\left\{\lamb \in \Lamb: \lamb <\lamb_j^* \right\} $.
    Then, $B[\lamb_j] = \textbf{True}$, and
    \[\marge{T_{\lamb_j}}{S_{j-1}} \ge (1-\epsi)\tau \lamb_j.\]
    Let $\lamb_j= \left\lfloor (1+\epsi)^u \right\rfloor$,
    and $\lamb^*_j= \left\lfloor (1+\epsi)^{u+1} \right\rfloor$.
    Then, 
    \begin{equation*}
    \frac{|T_{j, \lamb_j}|}{|T_{j, \lamb_j^*}|} = \frac{\lamb_j}{\lamb_j^*}
    = \frac{\left\lfloor (1+\epsi)^u \right \rfloor}{\left\lfloor (1+\epsi)^{u+1} \right \rfloor}\ge \frac{(1+\epsi)^u -1}{ (1+\epsi)^{u+1}}\ge \frac{1}{1+\epsi} -\epsi,
    \end{equation*}
    where the last inequality follows from $(1+\epsi)^{u+1} \ge \lamb_j^*>\left\lceil \frac{1}{\epsi}\right\rceil \ge \frac{1}{\epsi}$.
    Therefore, by above two inequalities and monotonicity of $f$,
    \begin{equation*}
    \marge{T_{j,\lamb_j^*}}{S_{j-1}}\ge\marge{T_{j,\lamb_j}}{S_{j-1}}
    \ge (1-\epsi)\left(\frac{1}{1+\epsi}-\epsi\right)\tau \lamb_j^*
    \ge (1-2\epsi)\tau \lamb_j^*/(1+\epsi).
    \end{equation*}
    The objective value of the returned solution can be bounded as follows,
    \begin{align*}
        &f(S) =\sum_j \marge{T_{\lamb_j^*}}{S_{j-1}}
        \ge \sum_j (1-2\epsi)\tau \lamb_j^*/(1+\epsi)
        = (1-2\epsi)\tau |S|/(1+\epsi).\qedhere
    \end{align*}
\end{proof}
\color{black}

\begin{proof}[Proof of Property (4)]
    If the output set $S$ follows that $|S| <k$,
    for any $x\in \mathcal N$,
    it is filtered out at some point.
    Let $x$ be discarded at itetation $j_{x}$.
    Then,
   $$\marge{x}{S} \le \marge{x}{S_{j_x-1}} < \tau.\qedhere$$
\end{proof}

\section{\tg}\label{apx:tg}
 \begin{algorithm}[t]
   \caption{\tg}
   \label{alg:tg}
   \begin{algorithmic}[1]
      \State {\bfseries Input:} Evaluation oracle $ f:2^{\mathcal N} \to \mathbb{R}$, constraint $k$, constant $\alpha$, value $\Gamma$ \st $\Gamma \le f(O) \le \Gamma/\alpha$, error $\epsi$
      \State Initialize $\tau \gets \Gamma/(\alpha k)$, $S \gets \emptyset$
      \While{$\tau \ge \epsi\Gamma/k$}
          \For{$e \in \mathcal N$}
            \State \textbf{if} $\marge{e}{S} \ge \tau$ \textbf{then} $S\gets S\cup \{e\}$
            \State \textbf{if} $|S| =k$ \textbf{then return} $S$
          \EndFor
          \State $\tau \gets \tau(1-\epsi)$
      \EndWhile
      \State \textbf{return} $S$
   \end{algorithmic}
 \end{algorithm}
In this section, we introduce a variant of near-linear time
algorithm (Alg. 1) in~\citet{BadanidiyuruV14},
\tg (Alg.~\ref{alg:tg}), which requires two more parameters,
constant $\alpha$ and value $\Gamma$.
With the assumption that $\Gamma$ is an $\alpha-$approximation solution,
\tg ensures a $(1-1/e-\epsi)-$approximation with $O(n/\epsi)$ query calls.
\begin{theorem}\label{thm:tg}
    Let $(f, k)$ be an instance of \sm,
    and $O$ is the optimal solution to $(f,k)$. 
    With input $\alpha$, $\Gamma$ and $\epsi$,
    where $\Gamma \le f(O) \le \Gamma/\alpha$,
    \tg outputs solution $S$ with $O(n\epsi^{-1}\log(\alpha\epsi)^{-1})$ queries 
    such that $f(S) \ge (1-1/e-\epsi)f(O)$.
\end{theorem}
\begin{proof}
    \textbf{Query Complexity.}
    The while loop in Alg.~\ref{alg:tg} has at most
    $ \log_{1-\epsi}(\alpha\epsi)\le \epsi^{-1}\log(\alpha\epsi)^{-1} $ iterations.
    For each iteration of the while loop,
    there are $O(n)$ queries.
    Therefore, the query complexity of \tg is
    $O(n\epsi^{-1}\log(\alpha\epsi)^{-1})$.

    \textbf{Approximation Ratio.}
    If $|S| < k$ at termination,
    for any $e \in \univ \setminus S$,
    it holds that $\marge{e}{S} < \epsi \Gamma/k \le \epsi f(O)/k$.
    Then, 
    \begin{align*}
        &f(O)-f(S) \overset{(1)}{\le}\sum_{o \in O\setminus S}\marge{o}{S}
    \le \epsi f(O)\\
    \Rightarrow&\hspace{1em} f(S) \ge (1-\epsi)f(O)
    \end{align*}
    where Inequality (1) follows from monotonicity and submodularity.

    If $|S|=k$ at termination, let $S_j$ be the intermediate solution
    after iteration $j$ of the while loop,
    and $\tau_j$ is the corresponding threshold.
    Suppose that $S_j\setminus S_{j-1} \neq \emptyset$.
    Then, each element added during iteration $j$ provides a minimum incremental benefit of $\tau$,
    \[f(S_j) - f(S_{j-1}) \ge |S_j\setminus S_{j-1}|\tau_j.\]
    Next, we consider iteration $j-1$.
    If $j>1$, since the algorithm does not terminate during iteration $j-1$,
    each element in $\univ$ is either added to the solution
    or omitted due to small marginal gain with respect to the current solution.
    Therefore, for any $o \in O\setminus S_{j-1}$,
    it holds that $\marge{o}{S_{j-1}}<\tau_{j-1}$ by submodularity.
    Then,
    \[f(O)-f(S_{j-1}) \le 
    \sum_{o \in O\setminus S_{j-1}}\marge{o}{S_{j-1}}<k\tau_{j-1}
    =k \tau_j/(1-\epsi).\]
    If $j=1$, $\tau_j = \Gamma/(\alpha k)$ and $S_{j-1}=\emptyset$.
    \[f(O)-f(S_{j-1}) =f(O) \le \Gamma/\alpha = k\tau_j < k \tau_j/(1-\epsi). \]
    Therefore, for any iteration $j$ that has elements being added to the solution, it holds that 
    \begin{align*}
        f(O)-f(S_j) &\le \left(1-\frac{1-\epsi}{k}|S_j\setminus S_{j-1}|\right)\left(f(O)-f(S_{j-1})\right)\\
        \Rightarrow \hspace{1em} f(O)-f(S) &\le \prod_j \left(1-\frac{1-\epsi}{k}|S_j\setminus S_{j-1}|\right) f(O)\\
        &\le \prod_j e^{-\frac{1-\epsi}{k}|S_j\setminus S_{j-1}|} f(O) \tag{$x+1\le e^x$}\\
        & = e^{-1+\epsi} f(O)\le (1/e+\epsi)f(O) \tag{$0<\epsi<1$}
    \end{align*}
    \[\Rightarrow \hspace{3em} f(S) \ge (1-1/e-\epsi)f(O)\hspace{3em} \qedhere\]
\end{proof}
\color{black}
\begin{table*}[t] 
\caption{Small and Large Data} \label{table:largeData}.
\centering
\begin{tabular}{ |l||l|l|l|l|  }
 \hline
 \multirow{2}{*}{Application} & \multicolumn{2}{|c|}{\textbf{Small Data}}             & \multicolumn{2}{|c|}{\textbf{Large Data}} \\
 \cline{2-5}
                              &  $n$ & Edges\tnote{$\ast$}                                                   &  $n$ & Edges\tnote{$\ast$} \\
 \hline
 ImageSumm      & 10,000     & $\simeq 1.0 \times 10^8$                                        & 50,000     & $\simeq 2.5 \times 10^9$  \\
 InfluenceMax   & 26,588     & $\simeq 1.0 \times 10^5$                                        & 1,134,890  & $\simeq 6.0 \times 10^6$  \\
 RevenueMax     & 17,432     & $\simeq 1.8 \times 10^5$                                        & 3,072,441  & $\simeq 2.3 \times 10^8$  \\
 MaxCover (BA)      & 100,000    & $\simeq 5.0 \times 10^5$                                    & 1,000,000  & $\simeq 1.0 \times 10^9$  \\
 \hline
\end{tabular}
\end{table*}
\section{\datfull with no Knowledge of OPT} \label{sec:DATOPT}
\unionprob*
\begin{proof}
    By Definition ($p_x$ in Section \ref{sec:dat}), it holds that $\prob{o \in O_1} = p_o$. 
    Since $o$ is assigned to each machine randomly with probability $1/\ell$, 
    \begin{align*}
      \prob{o \in O_2} &= \sum_{i = 1}^{\ell}\prob{o \in S_i} \\
       &= \ell \cdot \prob{o \in S_1 | o \in \univ_1}
       \cdot \prob{o \in \univ_1} \\
      &\revtwo= \prob{o \in \textsc{TSMRel}(\univ_1) | o \in \univ_1} \\
      &\revtwo= \prob{o \in \textsc{TSMRel}(\univ_1\cup \{o\})}\\
      &= 1-p_o.
    \end{align*}
    Moreover, we know that any two machines selects elements independently.
    So, 
    \begin{align*}
      \prob{o \in O_2|o\in O_1}
      &\revtwo= \prob{o \in O_2|o\not \in \textsc{TSMRel}(\univ_1\cup \{o\}, q)}\\
      &\revtwo= 1-\prob{o \not \in O_2|o\not \in \textsc{TSMRel}(\univ_1\cup \{o\}, q)}\\
      &\revtwo= 1-\prod_{i = 1}^{\ell}\prob{o \not \in S_i|o\not \in \textsc{TSMRel}(\univ_1\cup \{o\}, q)}\\
      &\overset{(a)}{=}1-\prod_{i = 2}^{\ell}\prob{o \not \in S_i}\\
      &\le 1- \prob{o \not \in O_2}\\
      & = \prob{o \in O_2}=1-p_o,
    \end{align*}
    where Inequality (a) follows from \revtwo
    $\prob{o \not \in \textsc{TSMRel}(\univ_1)|o\not \in \textsc{TSMRel}(\univ_1\cup \{o\})}=1 $. \color{black}
    
    Thus, we can bound the probability by the following,
    \begin{align*}
      \prob{o \in O_1 \cup O_2} =& \prob{o \in O_1}
      + \prob{o \in O_1} \\
      &-  \prob{o \in O_1 \cap O_2}\\
      \ge& p_o + 1-p_o - p_o(1-p_o)\\ \ge& 3/4.
    \end{align*}
  \end{proof}

\textbf{Description.} \ The \dat algorithm described in Alg. \ref{alg:DAT} is a two MR-rounds algorithm using the AFD approach that runs \lat concurrently on every machine for $\log{}_{1+\epsi}(k)$ different guesses
of threshold $\tau_{i,j}$ in the range $[\frac{\alpha\Delta_i^*}{k},
\alpha\Delta_i^*]$; where $\alpha$ is the approximation of \dat and $\Delta_i^*$ is the maximum singleton in $\univ_i$. Every solution 
returned to the primary machine are placed into bins based on their 
corresponding threshold guess $\tau_{i,j}$ such that
$\frac{\Delta^*}{k}(1+\epsi)^x \le \tau_{i,j} \le \frac{\Delta^*}{k}(1+\epsi)^{x+1}$; 
where $\Delta^*$ is the maximum singleton in $\univ$. Since there must exist a
threshold $\tau^*$ that is close enough to $\alpha \opt/k$ 
; running \lat (Line \ref{alg1:T_j}) on every
bin in the range $[\frac{\alpha\Delta^*}{k},\alpha\Delta^*]$ and
selecting the best solution guarantees the $\alpha$-approximation of 
\dat in $O(\log{}(n))$ adaptive rounds.
\begin{algorithm}[t]
  \caption{Threshold-DASH with no knowledge of \opt (\dat)}
  \label{alg:DAT}
  \begin{algorithmic}[1]
     \State {\bfseries Input:} Evaluation oracle $ f:2^{\mathcal N} \to \mathbb{R}$, constraint $k$, error $\epsi$, available machines $M \gets \{1, 2,...,\ell\}$
     \State Initialize $\delta \gets 1/(\ell+1)$, $\mathbf q \gets$ a fixed sequence of random bits.
     \State Set $\alpha \gets \frac{3}{8}$,
     $(\mathbf q_{i,j})_{i \in [\ell+1], j \in [\log_{1+\epsi}(k)]} \gets \mathbf q$
     \For{$e \in \univ$ do }\label{line:assignEle}
       \State Assign $e$ to each machine independently with probability $1/\ell$
     \EndFor
      
     \For{$i \in M$}\label{alg4:distLAT}
     \LineComment{On machine $i$}
     \State Let $\univ_i$ be the elements assigned to machine $i$
     \State Set $\Delta_i^* \gets \max\{ f(e) : e \in \univ_i\}$
     \For{$j \gets 0$ to $\log_{1+\epsi}(k)$ in parallel}
        \revtwo
        \State $\tau_{i,j} \gets \frac{(\alpha +\epsi)\Delta_i^*}{k}(1+\epsi)^j$
        \State $S_{i,j}, R_{i,j} \gets \lat (f, \mathcal N_i, k, \delta, \epsi, \tau_{j}, \mathbf q_{i,j})$ \label{alg1:Ti} 
     \EndFor
          
     \State Send $\Delta_i^*$ and all $(\tau_{i,j}, S_{i,j}, R_{i,j})$ to primary machine \color{black}
     \EndFor
     \LineComment{On primary machine} 
     \State Set $\Delta^* \gets \max\{ \Delta_i^* : 1 \le i \le \ell\}$
     \For{$x \gets 0$ to $\lceil \log_{1+\epsi}(k) \rceil+1$ in parallel}
        \revtwo
        \State $\tau_{x} \gets \frac{(\alpha+\epsi)\Delta^*}{k}(1+\epsi)^x$
        \State Let $R_x \gets \{ \bigcup R_{i,j} : \tau_x \le \tau_{i,j} \le \tau_{x+1}\} $ \color{black}
        \State Let $A_x \gets \{\text{Sample a solution } S_{i,j} : \tau_x \le \tau_{i,j} \le \tau_{x+1}\}$
        \State Let $g_x(\cdot) \gets f(A_x \cup \cdot) - f(A_x)$
        \revtwo
        \State $T_x \gets \lat (g_x, R_x, k-|A_x|,  \delta, \epsi, \tau_{x},\mathbf q_{\ell+1, x})$ \color{black} \label{alg1:T_j}
        \State $T_x' \gets A_x \cup T_x$ \label{line:TFinal_j}
     \EndFor
     
     \State $T \gets \argmax \{ f(T_x') : 0 \le x \le \log_{1+\epsi}(k) \}$ \label{alg1:T}
     \State \textbf{return} $T$
  \end{algorithmic}
\end{algorithm}
\begin{theorem}
Let $(f,k)$ be an instance of \sm where $k\log(k) < \frac{\epsi\psi}{ \ell}$.
  \dat with no knowledge of \opt returns set $T'$ with two MR rounds,
  $O\left(\frac{1}{\epsi^3}\log(n)\right)$
  adaptive rounds,
  $O\left(\frac{n\log(k)}{\epsi^4}\right)$ total queries,
  $O(n)$ communication complexity,
  and probability at least $1-n^{-c}$
   such that
  \[\ex{f(T')} \ge \left(\frac{3}{8}-\epsi\right)\opt.\]

\end{theorem}
\textbf{Overview of Proof.} 
Alg.~\ref{alg:DAT} is inspired by Alg.~\ref{alg:DATOPT} in 
Section~\ref{sec:dat}, which is a version of the
algorithm that knows that optimal solution value. 
With $\Delta^* = \max\{f(e): e \in \univ\}$, there exists an $x_0$ such that
$\tau_{x_0} \le \alpha \opt(1+\epsi)/k \le \tau_{x_0+1}$
\revone
with $\oh{\log(k)/\epsi}$ guesses.
\color{black}
Then, on each machine $i$, we only consider \revtwo sets $S_{i,j}$ and $R_{i,j}$ \color{black} such that
$\tau_{x_0} \le \tau_{i,j} \le \tau_{x_0+1}$.
If this $\tau_{i,j}$ does exist, \revtwo $(S_{i,j}, R_{i,j})$ works like $(S_i, R_i)$ \color{black} in Alg.~\ref{alg:DATOPT}.
If this $\tau_{i,j}$ does not exist, then for any $e \in \univ_i$,
it holds that $f(e) < \alpha \opt/k$, which means
$\lat(\univ_i)$ with $\tau = \alpha \opt/k$ will return an empty set.
\revone
Since each call of \lat with different guesses on $\tau$ is executed in parallel, 
the adaptivity remains the same and the query complexity increases by a factor of $\log(k)/\epsi$.
\color{black}
\begin{proof}
  First,
  for $x = 0$, \revtwo $\tau_x = (\alpha+\epsi)\Delta^*/k \le (\alpha+\epsi) \opt/k$; \color{black}
  and for $x = \lceil \log_{1+\epsi}(k) \rceil +1$, \revtwo
  \[\tau_x \ge (\alpha+\epsi)(1+\epsi) \Delta^*\ge 
  (\alpha+\epsi)(1+\epsi) \sum_{o \in O}f(o)/k \ge (\alpha+\epsi)(1+\epsi) \opt/k.\] \color{black}
  Therefore, there exists an $x_0$ such that \revtwo
  $\tau_{x_0} \le (\alpha+\epsi)(1+\epsi) \opt /k \le \tau_{x_0+1}$. \color{black}
  Since $\tau_{x_0+1} = \tau_{x_0}(1+\epsi)$, it holds that \revtwo
  $\tau_{x_0} \ge (\alpha+\epsi)\opt /k$ and 
  $\tau_{x_0+1} \le (\alpha+\epsi)(1+\epsi)^2 \opt  /k$. \color{black}

  Then, we only consider $T_{x_0}'$.
  If $|T_{x_0}'| = k$, by Theorem~\ref{theorem:threshold}, it holds that,
  \revtwo
  \begin{align*}
    f(T) \ge f(T_{x_0}') \ge \frac{1-2\epsi}{1+\epsi} 
  \tau_{x_0}k
  \ge \left(\frac{3}{8}-\epsi\right) \opt .
  \end{align*}
  \color{black}
  Otherwise, in the case that $|T_{x_0}'| < k$,
  let $A_{x_0} = S_{i_0, j_0}$, \revtwo
  $O_1 = \{o \in O: o \not \in \textsc{TSMRel} (N_{i_0} \cup \{o\}, q)\}$. \color{black}
  Also, let $\tau_{i, (j_0)}$ be returned by machine $i$ that 
  $\tau_{x_0} \le \tau_{i, (j_0)} \le \tau_{x_0+1} $,
  define \revtwo
  \begin{equation*}
    (S_{i, (j_0)}', R_{i, (j_0)}') = \begin{cases}
      (S_{i, (j_0)}, R_{i, (j_0)}) & \text{, if machine $i$ returned a }
      \tau_{i, (j_0)} \\
      (\emptyset,\emptyset)   & \text{, otherwise }
    \end{cases}.
  \end{equation*} \color{black}
  On the machine which does not return a $\tau_{i,(j_0)}$,
  we consider it runs $\lat(\univ_i, \tau_{x_0+1})$ and returns
  two empty sets, and
  hence
  $\max\{f(e): e \in \univ_i\} < \tau_{x_0}$. \revtwo
  Let $R_{x_0}' = \cup_{i \in M} R_{i, (j_0)}'$, 
  $O_2 = R_{x_0}' \cap O$. \color{black}
  Then, Lemma~\ref{lemma:union_prob} 
  in Appendix~\ref{sec:DATOPT} 
  still holds in this case.
  We can calculate the approximation ratio as follows with $\epsi \le 2/3$,
  \begin{align*}
    \ex{T}\ge \ex{T_{x_0}'}
    & \ge \ex{f(O_1 \cup O_2)} - 
    k \cdot \tau_{x_0+1}\\ &\ge \left(\frac{3}{8}-\epsi\right)\opt.
  \end{align*}
\end{proof}

\section{Experiment Setup }\label{appendix:obj}

\subsection{Applications}\label{app:app} Given a constraint $k$, the objectives of the applications are defined as follows:

\subsubsection{Max Cover} \label{exp:maxcov}
Maximize the number of nodes covered by choosing a set $S$ of maximum size $k$, such that the number of nodes having at least one neighbour in the set $S$. The application is run on synthetic random BA graphs of groundset size 100,000, 1,000,000 and 5,000,000 generated using Barabási–Albert (BA) models for the centralized and distributed experiments respectively. 

\subsubsection{Image Summarization on CIFAR-10 Data} \label{exp: imgsumm}
Given large collection of images, find a subset of maximum size $k$ which is representative of the entire collection. The objective used for the experiments is a monotone variant of the image summarization from \citet{fahrbach2019non}. For a groundset with $N$ images, it is defined as follows:
\begin{align*}
f(S) = \sum_{i \in N} \max_{j \in S} s_{i,j}
\end{align*}
where $s_{i,j}$ is the cosine similarity of the pixel values between image $i$ and image $j$. 
The data for the image summarization experiments contains 10,000 and 50,000 CIFAR-10 \cite{krizhevsky2009learning} color images respectively for the centralized and distributed experiments.

\subsubsection{Influence Maximization on a Social Network.} \label{exp: infmax}
Maximise the aggregate influence to promote a topic by selecting a set of social network influencers of maximum size $k$. The probability that a random user $i$ will be influenced by the set of influencers in $S$ is given by:
\begin{align*}
    f_i(S) & = 1 \quad \textrm{, for $i \in S$} \\
    f_i(S) & = 1 - (1 - p)^{|N_S(i)|} \quad \textrm{, for  $i \notin S$ }
\end{align*}
 
where |$N_S(i)$| is the number of neighbors of node $i$ in $S$. 
We use the Epinions data set consisting of 27,000 users from \citet{rossi2015network} for the centralized data experiments and the Youtube online social network data \citet{yang2015defining} consisting more than 1 million users for distrbuted data experiments. The value of $p$ is set to 0.01.

\subsubsection{Revenue Maximization on YouTube.} \label{exp: revmax}
Maximise revenue of a product by selecting set of users $S$ of maximum size $k$, where the network neighbors will be advertised a different product by the set of users $S$. It is based on the objective function from \citet{mirzasoleiman2016fast}. For a given set of users $X$ and $w_{i,j}$ as the influence between user $i$ and $j$, the objective function can be defined by:
  \begin{align*}
      f(S) & = \sum_{i \in X} V \left(\sum_{j \in S} w_{i,j}\right) \\
      V(y) & = y^{\alpha}
  \end{align*}
  where $V(S)$, the expected revenue from an user is a function of the sum of influences from neighbors who are in $S$ and $\alpha : 0 < \alpha < 1 $ is a rate of diminishing returns parameter for increased cover. 

  We use the Youtube data set from \citet{mirzasoleiman2016fast} consisting of 18,000 users for centralized data experiments. For the distrbuted data experiments we perform empirical evaluation on the Orkut online social network data from \citet{yang2015defining} consisting more than 3 million users. The value of $\alpha$ is set to 0.3
\section{Replicating the Experiment Results}\label{appendix:repExp}

Our experiments can be replicated by running the following scripts:
\begin{itemize}
\item Install \textbf{MPICH} version \textbf{3.3a2} (DO NOT install OpenMPI and ensure \textit{mpirun} utlizes mpich using the command \textit{mpirun --version} (Ubuntu))
\item Install \textbf{pandas, mpi4py, scipy, networkx}
\item  Set up an MPI cluster using the following tutorial:\\ \textit{https://mpitutorial.com/tutorials/running-an-mpi-cluster-within-a-lan/}
\item Create and update the host file \textit{../nodesFileIPnew} to store the ip addresses of all the connected MPI machines before running any experiments (First machine being the primary machine)
	\begin{itemize}
		\item NOTE: Please place \textit{nodesFileIPnew} inside the MPI shared repository; \textit{"cloud/"} in this case (at the same level as the code base directory).  DO NOT place it inside the code base \textit{"DASH-Distributed\_SMCC-python/"} directory.
	\end{itemize}
\item Clone the \textit{DASH-Distributed\_SMCC-python} repository inside the MPI shared repository  (\textit{/cloud} in the case using the given tutorial)
	\begin{itemize}
		\item NOTE: Please clone the \textit{"DASH-Distributed\_SMCC-python"} repository and execute the following commands on a machine with \textbf{sufficient memory} (RAM); capable of generating the large datasets. This repository NEED NOT be the primary repository (\textit{"/cloud/DASH-Distributed\_SMCC-python/"}) on the shared memory of the cluster; that will be used for the experiments.
	\end{itemize}
\item Additional Datasets For Experiment 1 : Please download the Image Similarity Matrix file \textbf{"images\_10K\_mat.csv"}(\textit{https://drive.google.com/file/d/1s9PzUhV-C5dW8iL4tZPVjSRX4PBhrsiJ/view?usp=sharing)}) and place it in the \textit{data/data\_exp1/}  directory.
\item To generate the decentralized data for \textbf{Experiment 2} and \textbf{3} : Please follow the below steps:
	\begin{itemize}
		\item Execute \textit{bash GenerateDistributedData.bash \textbf{nThreads} \textbf{nNodes} } 
		\item The previous command should generate \textbf{nNodes} directories in \textit{loading\_data/} directory (with names \textit{machine$<$nodeNo$>$Data})
		\item Copy the \textit{data\_exp2\_split/} and \textit{data\_exp3\_split/} directories within each \textit{machine$<$i$>$Data} directory to the corresponding machine $M_{i}$ and place the directories outside \textit{/cloud} (directory created after setting up an MPI cluster using the given tutorial)).
	\end{itemize}	  

\end{itemize}

\textbf{To run all experiments in the apper}\\
Please read the \textit{README.md} file in the \textit{"DASH-Distributed\_SMCC-python"} (Code/Data Appendix) for detailed information.
\section{The \sg Algorithm (\citet{Nemhauser1978a})} \label{app:Greedy}
The standard \sg algorithm starts with an empty set and then proceeds by adding elements to the set over $k$ iterations. In each iteration the algorithm selects the element with the maximum marginal gain $\Delta(e | A_{i-1})$ where $ \Delta(i | A) = f(A \cup \{i\}) - f(A)$. The Algorithm \ref{algo:SG}  is a formal statement of the standard \sg algorithm. The intermediate solution set $A_i$ represents the solution after iteration $i$. The ($1-1/e \approx$ 0.632)-approximation of the \sg algorithm  is the optimal ratio possible for monotone objectives.

\begin{algorithm}[h]
   \caption{$\sg (f, \mathcal N, k)$} \label{algo:SG}
  \begin{algorithmic}[1]
     \State {\bfseries Input:} evaluation oracle $ f:2^{\mathcal N} \to \mathbb{R}$, constraint $k$
     \State Let $A_i \gets \emptyset $
     \For{$i = 1$ to $k$ } \label{line:assignEle}
       \State $u_i \gets \argmax \{ f(e | A_{i-1}) ; \text{ where } e \in \univ \textbackslash A_{i-1}\}$
       \State $A_i \gets A_{i-1} \cup u_i$
     \EndFor
     
     \State \textbf{return} \textit{$A_k$}
  \end{algorithmic}
\end{algorithm}
\section{Discussion of the $1/2$-approximate Algorithm in~\citet{liu2018submodular}} \label{app:Vondrak}
\begin{algorithm}[h]
   \caption{$\tg(S, G, \tau)$} \label{algo:TG}
  \begin{algorithmic}[1]
     \State \textbf{Input: } An input set $S$,
     a partial greedy solution $G$ with $G \le k$,
     and a threshold $\tau$
     \State $G' \gets G$
     \For{$e \in S$}
     	\If{$\marge{e}{G_0} \ge \tau$ and $|G_0| < k$}
     		\State $G_0\gets G_0 \cup \{e\}$
     		\EndIf
     \EndFor
     \State \textbf{return} $G'$
  \end{algorithmic}
\end{algorithm}
\begin{algorithm}[h]
   \caption{$\tf(S, G, \tau)$} \label{algo:TF}
  \begin{algorithmic}[1]
     \State \textbf{Input: } An input set $S$,
     a partial greedy solution $G$ with $G \le k$,
     and a threshold $\tau$
     \State $S' \gets S$
     \For{$e \in S$}
     		\State \textbf{if} $\marge{e}{G} < \tau$ 
     		\textbf{then} $S' \gets S' \setminus \{e\}$
     		\EndFor
     \State \textbf{return} $S'$
  \end{algorithmic}
\end{algorithm}
\begin{algorithm}[h]
   \caption{A simple 2-round $1/2$ approximation, assuming \opt is known} \label{algo:SimpleG}
  \begin{algorithmic}[1]
     \State \textbf{round 1:}
     \State $S\gets$ sample each $e \in \univ$ with probability $p = 4\sqrt{k/n}$
     \State send $S$ to each machine and the central machine $C$
     \State partition $\univ$ randomly into sets $V_1, V_2, \ldots, V_m$
     \State send $V_i$ to machine $i$ for each $i \in [m]$
     \For{each machine $M_i$ (in parallel)}
     	\State $\tau \gets \frac{\opt}{2k}$
     	\State $G_0 \gets \tg(S, \emptyset, \tau)$
     	\If{$|G_0| < k$}
     		\State $R_i \gets \tf(V_i, G_0, \tau)$
     	\Else 
     	\State $R_i = \emptyset$
     	\EndIf
     	\State send $R_i$ to the central machine $C$
     \EndFor
     \Statex
     \State \textbf{round 2 (only on $C$):}
     \State compute $G_0$ from $S$ as in first round
     \For{$e \in \cup_i R_i$}
     	\State \textbf{if} $\marge{e}{G} \ge \tau$ 
     		\textbf{then} $G \gets G \cup \{e\}$
     \EndFor
     \State \textbf{return} $G$
  \end{algorithmic}
\end{algorithm}
The simple two-round distributed algorithm proposed by~\citet{liu2018submodular}
achieves a deterministic $1/2-\epsi$ approximation ratio.
The pseudocode of this algorithm, assuming that we know \opt,
is provided in Alg.~\ref{algo:SimpleG}.
The guessing \opt version runs $\oh{\log(k)}$ copies of
Alg.~\ref{algo:SimpleG} on each machine, which
has $\oh{k}$ adaptivity and $\oh {nk \log(k)}$ query complexity. 

Intuitively, the result of Alg.~\ref{algo:SimpleG} is equivalent to a single run
of \tg on the whole ground set.
Within the first round, it calls \tg on $S$ to get an intermediate solution $G_0$.
After filtering out the elements with gains less than $\tau$ in $V_i$,
the rest of the elements in $V_i$ are sent to the primary machine.
By careful analysis, there is a high probability that
the total number of elements returned within the first round is
less than $\sqrt{nk}$.

Compared to the 2-round algorithms in Table~\ref{table:algs}, Alg.~\ref{algo:SimpleG} is the only algorithm that requires data duplication, with four times more elements distributed to each machine in the first round. As a result, distributed setups have a more rigid memory restriction when running this algorithm.
For example, when the dataset exactly fits into the distributed setups
which follows that $\Psi = n/\ell$,
there is no more space on each machine to store the random subset $S$.
In the meantime, other 2-round algorithms can still run with $k \le n/\ell^2$.

Additionally, since the approximation analysis of Alg.~\ref{algo:SimpleG} does not take randomization into consideration, we may be able to substitute its \tg subroutine with any threshold algorithm (such as \threshold in \citet{chen2021best}) that has optimal logarithmic adaptivity and linear query complexity.
However, given its analysis of the success probability considers $S$ as $3k$ blocks where each block has random access to the whole ground set, it may not be easy to fit parallelizable threshold algorithms into this framework.

\clearpage
\bibliography{sample-base-dblp.bib}
\bibliographystyle{plainnat} 
\end{document}